\documentclass[aps,pra,twocolumn,letterpaper,showpacs,superscriptaddress,floatfix,10pt,nofootinbib]{revtex4-2}
\usepackage[pdftex]{graphicx}
\graphicspath{ {./images/} }
\usepackage{dcolumn}
\usepackage{bm}
\usepackage{bbm}
\usepackage[usenames, dvipsnames]{color}
\usepackage{comment}
\usepackage[colorlinks]{hyperref}

\usepackage{layouts}

\usepackage[T1]{fontenc}

\usepackage{amsfonts}
\usepackage{amssymb}
\usepackage{amsmath}
\usepackage{mathtools}
\usepackage{tensor}
\usepackage{xcolor}
\usepackage[normalem]{ulem}

\usepackage{amsthm}
\theoremstyle{thmstyleone}%
\newtheorem{theorem}{Theorem}[section]
\newtheorem{lemma}[theorem]{Lemma}%

\theoremstyle{thmstyletwo}%

\theoremstyle{thmstylethree}%
\newtheorem{definition}{Definition}[section]%

\usepackage{multirow}
\usepackage[scr=boondoxo, scrscaled=1.05]{mathalfa}

\usepackage{diagbox}
\usepackage{array}
\newcolumntype{L}[1]{>{\raggedright\let\newline\\\arraybackslash\hspace{0pt}}m{#1}}
\newcolumntype{C}{>{\centering\arraybackslash}X}
\newcolumntype{s}{>{\centering\arraybackslash\hsize=\hsize}X}
\newcolumntype{R}[1]{>{\raggedleft\let\newline\\\arraybackslash\hspace{0pt}}m{#1}}
\usepackage{tabularx}

\newcommand{\n}{\nonumber \\ }
\renewcommand{\d}{\mathrm{d}}

\newcommand{\Bm}{{\vec{m}}}
\newcommand{\Bw}{{\vec{\omega}}}
\newcommand{\Bt}{{\vec{\theta}}}

\newcommand{\N}{{\mathbb{N}}}
\newcommand{\R}{{\mathbb{R}}}
\newcommand{\Q}{{\mathbb{Q}}}
\newcommand{\Z}{{\mathbb{Z}}}
\newcommand{\C}{{\mathbb{C}}}
\newcommand{\T}{{\mathbb{T}}}

\newcommand{\QCA}{\mathrm{QCA}}
\newcommand{\inv}{\mathrm{inv}}

\def\bra#1{\langle#1|}
\def\ket#1{|#1\rangle}

\def\qexp#1#2{\bra{#2}#1\ket{#2}}
\def\cexp#1{\langle#1\rangle}

\def\map#1#2{\left[#1,#2\right]}
\def\ho#1#2{\pi_0\left[#1,#2\right]}
\def\tr#1{\mathrm{Tr}\left[#1\right]}

\begin{document}

\title{Topological Phases of Many-Body Localized Systems: Beyond Eigenstate Order}

\author{David M. Long}
\email{dmlong@stanford.edu}
\affiliation{Condensed Matter Theory Center and Joint Quantum Institute,\\Department of Physics, The University of Maryland, College Park, Maryland 20742, USA}
\affiliation{Department of Physics, Stanford University, Stanford, California 94305, USA}

\author{Dominic V. Else}
\affiliation{Perimeter Institute for Theoretical Physics, Waterloo, Ontario N2L 2Y5, Canada}

\date{\today}

\begin{abstract}
    Many-body localization (MBL) lends remarkable robustness to nonequilibrium phases of matter.
    Such phases can show topological and symmetry breaking order in their ground and excited states, but they may also belong to an \emph{anomalous localized topological phase} (ALT phase).
    All eigenstates in an ALT phase are trivial, in that they can be deformed to product states, but the entire Hamiltonian cannot be deformed to a trivial localized model without going through a delocalization transition.
    Using a correspondence between MBL phases with short-ranged entanglement and locality preserving unitaries---called quantum cellular automata (QCA)---we reduce the classification of ALT phases to that of QCA.
    This method extends to periodically (Floquet) and quasiperiodically driven ALT phases, and captures anomalous Floquet phases within the same framework as static phases.
    We considerably develop the study of the topology of QCA, allowing us to classify static and driven ALT phases in low dimensions.
    The QCA framework further generalizes to include symmetry-enriched ALT phases (SALT phases)---which we also classify in low dimensions---and provides a large class of soluble models suitable for realization in quantum simulators.
    In systematizing the study of ALT phases, we both greatly extend the classification of interacting nonequilibrium systems and clarify a confusion in the literature which implicitly equates nontrivial Hamiltonians with nontrivial ground states.
\end{abstract}

\maketitle

\section{Introduction}
\label{sec:intro}

    Many-body localization (MBL)~\cite{Anderson1958,Gornyi2005,Basko2006,Oganesyan2007,Pal2010,Huse2014fullyMBL,Schreiber2015,Smith2016,Imbrie2016MBL,deRoeck2023rigorous,Sierant2024review} can protect robust, topological phases even in systems which are far from equilibrium~\cite{Huse2013order,Parameswaran2018MBLsymm}. Localized systems can support spontaneous symmetry breaking~\cite{Huse2013order,Pekker2014eigord,Kjall2014,Parameswaran2018MBLsymm}, topological order~\cite{Wen2017review,Bauer2013TO,Slagle2015SPT,Chan2020SPT,Wahl2020mblto,Venn2022TO2d}, and nontrivial phases with external driving~\cite{Rudner2013,Po2016QCA,Titum2016,Nathan2019AFI,Zhang2021U1,Nathan2021hierarchy,Else2019longlived,Friedman2022qpedge,Dumitrescu2022,Long2021ALTP,Nathan2021pump,Sacha2015,Khemani2016,Else2016,vonKeyserlingk2016,Choi2017,Khemani2019brief,Else2020,Mi2021}. In fact, any Hamiltonian built out of commuting terms, of the kind frequently employed to study topological order~\cite{Kitaev2003anyons,Levin2005levinwen,Chamon2005fracton,Bravyi2011fracton,Haah2011code,Levin2012spt,Walker2012walkerwang,Burnell20143fermi}, becomes robustly localized by choosing strongly disordered coefficients for the commuting terms~\cite{Wahl2020mblto}.

    \begin{figure}
        \centering
        \includegraphics[width=\linewidth]{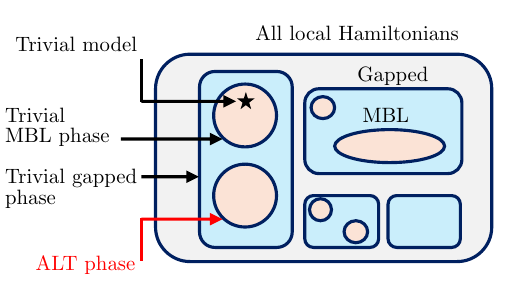}
        \caption{The topological classification of gapped or MBL Hamiltonians aims to identify the path connected components of gapped or MBL models (respectively), viewed as a subspace of all local Hamiltonians. MBL Hamiltonians are a subspace of gapped Hamiltonians. The path component (in either classification) containing a specified trivial model is called the trivial phase. An anomalous localized topological phase (ALT phase) is a nontrivial MBL phase with no eigenstate order, and so in particular no ground state order, and thus is contained inside the trivial gapped phase. There are also nontrivial gapped phases with no MBL Hamiltonians.}
        \label{fig:MBLGapped}
    \end{figure}
    
    Two MBL systems that can be deformed into one another without going through a delocalization transition belong to the same \emph{localized phase}~\cite{Titum2016,Po2016QCA,Roy2017loops,Nathan2019AFI,Friedman2022qpedge,Lapierre2022TLI,Long2021ALTP}. Two Hamiltonians certainly belong to distinct localized phases when they support distinct \emph{eigenstate order}~\cite{Huse2013order,Pekker2014eigord,Kjall2014,Parameswaran2018MBLsymm}---their eigenstates cannot be deformed into one another without their correlation length diverging. Indeed, the lack of a delocalization transition throughout a deformation of the Hamiltonian implies \emph{all} eigenstates retain only short-ranged correlations.
    
    However, in this article, we argue eigenstate order is not the right concept to obtain a general understanding of localized topological phases. For one thing, not all ground state topological phases are compatible with MBL. More importantly, there exist \emph{anomalous localized topological phases}\footnote{The label ``anomalous'' descends from Refs.~\cite{Rudner2013,Titum2016}. Strictly, only the edges of these phases are anomalous in a similar sense to what is usually meant for topological phases of matter.} (ALT phases)~\cite{Rudner2013,Titum2016,Po2016QCA,Roy2017loops,Long2021ALTP,Lapierre2022TLI}. ALT phases have \emph{no} eigenstate order, but nonetheless cannot be deformed to a trivial model---in which the commuting terms have disjoint support~\cite{Hastings2013torus}---without going through a delocalization transition.

    The situation is illustrated in \autoref{fig:MBLGapped}, with the example eigenstate being the ground state. Within the space of all local Hamiltonians, we identify those with a finite energy gap above the ground state. The topological classification of ground states aims to enumerate and describe all connected components (phases) of this space of gapped Hamiltonians. The phase containing a specified trivial model, consisting of a sum of single-site operators, is called the trivial phase. Among these gapped models, some are actually MBL. The topological classification of MBL phases, in complete analogy to gapped phases, aims to enumerate and describe all the connected components of this space of MBL Hamiltonians. However, there is not just one MBL phase within each gapped phase. There can be many, or none at all. Some or all of the nontrivial MBL phases which are contained within the trivial gapped phase are ALT phases.
    
    The existence of ALT phases in driven systems is recognized~\cite{Rudner2013,Nathan2017magnet,Po2016QCA,Roy2017loops,Long2021ALTP,Long2022layer,Nathan2021hierarchy,Zhang2021U1,Zhang2022SPTentangler}, but in static systems, ALT phases have only been discussed in the single-particle context~\cite{Lapierre2022TLI}. Nonetheless, we propose that the static interacting models of Refs.~\cite{Walker2012walkerwang,Burnell20143fermi,Shirley2022semionQCA,Haah20233d,Fidkowski2023pumping} in fact represent ALT phases upon disordering the Hamiltonian.
    
    \begin{figure}
        \centering
        \includegraphics{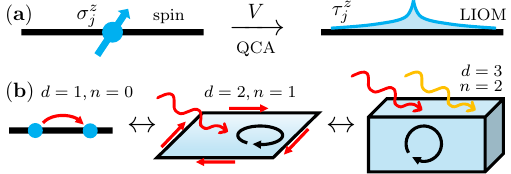}
        \caption{(a)~Many-body localization is characterized by commuting LIOMs, \(\tau^z_j\)~\cite{Huse2014fullyMBL}. The LIOMs of an MBL model with short-ranged entanglement can be prepared from single-site operators \(\sigma^z_j\) by a QCA, \(\tau^z_j = V \sigma^z_j V^\dagger\). ALT phases cannot be deformed back to a trivial model without delocalizing, even though they have no eigenstate order. (b)~Multi-tone driven ALT phases correspond to loops or tori of QCA. The \(\Omega\)-spectrum property \(\QCA_{d-1} \simeq \Omega \QCA_{d}\)~\eqref{eqn:Omega_spectrum} implies that \(n\)-tone-driven ALT phases in \(d\) dimensions correspond to \((d-1)\)-dimensional phases with \(n-1\) drives. The translation QCA in \(d=1\)~\cite{Gross2012QCA} corresponds with the anomalous Floquet phase in \(d=2\)~\cite{Po2016QCA}, which in turn corresponds with a two-tone-driven phase in \(d=3\) with circulating surface states.}
        \label{fig:ALT-QCA}
    \end{figure}
    
    This article systematizes the study of ALT phases in strongly interacting systems. Specifically, we focus on localized topological phases satisfying a condition we refer to as \emph{short-ranged entanglement}. We argue that instead of eigenstate order, the right way to characterize \emph{all} topological localized phases with short-ranged entanglement is via  \emph{quantum cellular automata} (QCA)~\cite{Gross2012QCA,Farrelly2020qca,Arrighi2019qcareview,Haah20233d}---the locality preserving unitaries which map physical degrees of freedom to local integrals of motion (LIOMs) [\autoref{fig:ALT-QCA}(a)] (\autoref{sec:MBLphases}).
    
    As QCA in spin systems have been classified in one~\cite{Gross2012QCA} and two~\cite{Freedman2020higherd} dimensions, this correspondence allows for a classification of ALT phases in low dimensional systems. 
    The same formalism extends to periodically (Floquet~\cite{Floquet1883}) and quasiperiodically driven ALT phases~\cite{Long2021ALTP}, which are classified by loops or tori of QCA, respectively. 
    We argue that loops of QCA are completely classified through a picture of \emph{pumping} lower-dimensional QCA to their boundary [\autoref{fig:ALT-QCA}(b)].
    More technically, we conjecture that QCA form a mathematical structure known as an \emph{\(\Omega\)-spectrum}~\cite{Kitaev2015SPT,Xiong2018minimalist,Gaiotto2019gencohomology}.
    This has wide-reaching implications for the classification of driven ALT phases, and in particular allows us to classify all \(d\)-dimensional ALT phases with \(n\) incommensurate driving frequencies such that \(d-n \leq 2\) (\autoref{sec:Classification}).

    Our methods extend immediately to symmetry-enriched ALT phases (SALT phases), which correspond to symmetric QCA~\cite{Gong2020MPU,Zhang2021U1,Zhang2022SPTentangler,Zhang2023note}.
    The \(\Omega\)-spectrum proposal provides a conjecture for the classification of SALT phases in terms of a \emph{generalized cohomology theory}~\cite{Kitaev2015SPT,Xiong2018minimalist,Gaiotto2019gencohomology}.
    With this perspective, we compute the classification of \(n\)-fold loops of symmetric QCA (that is, the \(n\)th homotopy group) for any group \(G\) and \(d-n \leq 1\). 
    This generalizes existing partial classifications for certain \(G\) and \(d=1\), \(n=0\)~\cite{Gong2020MPU,Zhang2021U1,Zhang2022SPTentangler,Zhang2023note}. 
    For these low dimensions, this calculation is not dependent on the \(\Omega\)-spectrum conjecture always being true (\autoref{sec:SymmetryEnrichment}).
    
    Finally, the pumping picture for QCA loops naturally provides explicit soluble models of driven ALT phases. 
    These models reveal the \emph{anomalous edge dynamics} of driven ALT phases, which cannot be realized in a short-range entangled MBL model of the same dimension as the edge.
    Similarly to short-ranged entangled phases of states, this anomalous edge dynamics is the generic signature of an ALT phase. Being finite depth circuits, these models are well suited to realization in current quantum simulators~\cite{Eckardt2017,Gross2017,Kjaergaard2020}  (\autoref{sec:Models}).

\section{Anomalous localized topological phases}
\label{sec:MBLphases}

    \subsection{Localized phases}

    We aim to classify \emph{many-body localized} phases. We use a strong notion of localization, namely the existence of a complete set of local integrals of motion (LIOMs). A localized phase is then defined as the set of LIOM Hamiltonians which can be deformed into one another while still possessing a complete set of LIOMs.
    
    More formally, a Hamiltonian \(H\) is many-body localized when it possesses a complete set of commuting LIOMs \(\tau^z_j\) such that~\cite{Huse2014fullyMBL}
    \begin{equation}
        H = \sum_j h_j \tau^z_j + \sum_{jk} h_{jk} \tau^z_j \tau^z_k + \cdots.
        \label{eqn:MBL_H}
    \end{equation}
    Here, \(j\) indexes LIOMs, ``\(\cdots\)'' indicates higher-body interactions, and \(h_{jk\ldots l}\) are real numbers which decay with the maximum distance between any of \(\tau^z_j, \ldots, \tau^z_l\). For the sake of simplicity, we consider local degrees of freedom on the integer lattice in \(d\) dimensions, \(\Z^d\).
    
    Equation~\eqref{eqn:MBL_H} identifies MBL Hamiltonians with commuting models, which are independently used as exactly solvable models of topological order~\cite{Kitaev2003anyons,Levin2005levinwen,Chamon2005fracton,Bravyi2011fracton,Haah2011code} and symmetry-protected topological (SPT) phases~\cite{Wen2017review,Chen20111dspt,Chen20112DSPT,Chen2012spt,Levin2012spt,Chen2013SPT,Burnell20143fermi,Chen2014decorated}. In the MBL context, the LIOMs may have exponential tails.
    
    MBL arising from strong random disorder is believed to be asymptotically unstable for \(d>1\)~\cite{deRoeck2017avalanche}, and doubt has arisen regarding the LIOM description even for \(d=1\)~\cite{Suntajs2020,Schulz2020,Sels2021obstruction,Sels2023dilute,Morningstar2022avalanche,Sels2021bathavalnce,Sierant2024review}. The case of correlated (including quasiperiodic) disorder remains open~\cite{Khemani2017qpuniv,Agrawal2020qp,Agrawal2022qphighd,Tu2023avalanche,Tu2024nee}. Regardless, the Hamiltonian~\eqref{eqn:MBL_H} continues to function as an effective model for astronomical prethermal timescales~\cite{Suntajs2020,Morningstar2022avalanche,Sels2021bathavalnce,Long2023phenomenology,deRoeck2023rigorous,Sierant2024review}. Indeed, the timescale for topological response in the Hamiltonian~\eqref{eqn:MBL_H} (which is finite in the LIOM model~\cite{Long2021ALTP,Lapierre2022TLI,Po2016QCA,Nathan2021pump,Nathan2019AFI,Friedman2022qpedge,Dumitrescu2022}) is parametrically separate from the thermalization time (which diverges exponentially as perturbations away from the LIOM model are reduced~\cite{Sierant2024review}). This vast parametric separation of time scales justifies treating the topological classification of LIOM Hamiltonians and the eventual thermalization process as separate problems. This article addresses the former. We will henceforth assume the stability of MBL, with the understanding that in some cases this may be an approximate statement applying to an effective model.
    
    We say two MBL Hamiltonians \(H_0\) and \(H_1\) are in the same \emph{localized topological phase} if there is a continuous path of MBL Hamiltonians \(H_t\) (\(t \in [0,1]\)) which connects them. We identify the trivial phase as the phase which contains an MBL Hamiltonian with LIOMs supported on single sites (\autoref{fig:MBLGapped}). We denote these trivial LIOMs \(\sigma^z_j\), even when the local Hilbert space dimension is not 2.
    
    We seek to construct a continuous path of \emph{locality preserving} unitaries \(V_t\) such that \(V_t\) conjugates the LIOMs of \(H_0\) into LIOMs of \(H_t\).
    A unitary \(V\) is called locality preserving if \(V A V^\dagger\) is localized near the support of \(A\) for all local operators \(A\).
    In a physical context, this means \(V A V^\dagger\) has some operator decomposition where the norm of a term in the decomposition decays with its distance from the support of \(A\). Any finite-time evolution under a local Hamiltonian is locality preserving in this sense.
    Mathematically, it is convenient to deal with \emph{strictly} locality preserving unitaries, where the support of \(V A V^\dagger\) grows by a finite distance \(r\). This includes all finite depth local quantum circuits.
    Locality-preserving unitaries are also called quantum cellular automata (QCA), and we will use this terminology in the rest of the paper.
    
    Ignoring, for the moment, long-range resonances in the spectrum of \(H_t\), the path of QCA \(V_t\) we need may be constructed from the adiabatic evolution of the many-body eigenstates~\cite{Berry1993friction,Jarzynski1995geometric,Hastings2010localityquantumsystems,Claeys2019floquetengineer,Pandey2020AGP,Gopalakrishnan2015lowfreq}.
    Including long-range resonances prevents an adiabatic construction of \(V_t\)~\cite{Khemani2015adiabatic}, but it is expected that sufficiently dilute resonances can be incorporated such that there is still a \(V_1\) that conjugates LIOMs of \(H_0\) into those of \(H_1\) and which is continuously connected to the identity. This is essentially the statement of the stability of MBL, as in Ref.~\cite{Imbrie2016MBL}.
    Thus, if \(H_0\) and \(H_1\) are in the same localized phase, their LIOMs are related by a continuous path of QCA.
    Conversely, if the LIOMs of \(H_0\) and \(H_1\) are related by a path of QCA, it is straightforward to find a path between \(H_0\) and \(H_1\). Indeed, just set \(H_t = V_{2t} H_0 V_{2t}^\dagger\) for \(t \in [0,1/2]\) and \(H_t = 2(1-t) V_1 H_0 V_1^\dagger + (2t-1) H_1\) for \(t\in[1/2,1]\). The latter is localized as \(V_1 H_0 V_1^\dagger\) and \(H_1\) commute.

    An MBL Hamiltonian is said to have eigenstate order if any of its eigenstates cannot be continuously deformed to a product state~\cite{Huse2013order}. A familiar example is the toric code Hamiltonian with disordered coefficients~\cite{Kitaev2003anyons}. This model is explicitly constructed from commuting terms---the LIOMs---and is well known to possess topological order.

    The remainder of this article will restrict to MBL models with \emph{short-ranged entanglement}. In the context of ground state phases, short-ranged entanglement means that all excitations on top of the state are topologically trivial. That is, there are local unitaries which can remove the excitations. We extend this notion to MBL models similarly. We say an MBL model has short-ranged entanglement if, for each LIOM \(\tau^z_j\), there is a local unitary \(\tau^x_j\) which can create and destroy excitations of \(\tau^z_j\) without affecting other LIOMs \(\tau^z_{k \neq j}\).
    We note that this property implies all eigenstates are short-range entangled in the usual sense. An excitation on top of any eigenstate can be removed by applying \(\tau^x_j\) to all the excited LIOMs.
    
    An example of a nontrivial short-ranged-entangled MBL model is the nontrivial fixed point of the Kitaev Majorana chain~\cite{Kitaev2001majorana}. 
    In contrast, while the toric code has LIOMs (the plaquette and star operators), the short-ranged entanglement condition is not met: the anyons can only be annihilated in pairs, and an isolated anyon cannot be removed locally~\cite{Kitaev2003anyons}. Similarly, symmetry breaking phases are not short-ranged entangled, as domain walls cannot be annihilated locally.

    \subsection{Anomalous phases}

    Not all MBL phases are characterized by eigenstate order, as we describe in this section. Rather, the full classification of short-ranged entangled MBL phases can be obtained from the path-connected components of the space of quantum cellular automata in the same dimension, \(\pi_0(\QCA_d)\).
    
    First, we note that eigenstate order, despite being a statement about individual states, is a robust property of the entire MBL Hamiltonian. This is an immediate consequence of the existence of the path of unitaries \(V_t\) connecting LIOMs of Hamiltonians \(H_0\) and \(H_1\) in the same localized phase. Indeed, if \(\ket{\psi_0}\) was an eigenstate of \(H_0\), then \(\ket{\psi_t} = V_t \ket{\psi_0}\) is a path of short-range correlated states demonstrating that \(\ket{\psi_0}\) can be continuously deformed into an eigenstate \(\ket{\psi_1}\) of \(H_1\). Thus, \(H_0\) and \(H_1\) have the same eigenstate order. In contrapositive, Hamiltonians with distinct eigenstate order cannot be in the same localized phase.
    
    However, the converse is not true: even if they have the same eigenstate order, it need not be possible to deform \(H_0\) to \(H_1\). This phenomenon is exemplified by ALT phases, which have \emph{no} eigenstate order, but which cannot be deformed to a trivial MBL model with LIOMs supported on a single lattice site.
    
    An example of a nontrivial ALT phase is the three-dimensional 3-fermion Walker-Wang model~\cite{Walker2012walkerwang,Burnell20143fermi,Haah20233d}. In the absence of symmetry, this model is believed to have no eigenstate order~\cite{Burnell20143fermi,Fidkowski2020beyond}, but nonetheless it supports topologically ordered surface states protected by localization. 
    Phrased in the language of MBL, \cite[Theorem II.4]{Haah20233d} showed that the 3-fermion model, and, in fact, any MBL Hamiltonian \(H\) with short-ranged entanglement, can be prepared from a trivial model by a QCA. That is, there is a QCA \(V\) such that \( V^\dagger H V \) has single-site LIOMs \(\sigma^z_i\). (We review the proof in Appendix~\ref{app:ALTtoQCA}.)
    Note that, given an MBL Hamiltonian $H$, the choice of QCA $V$ is not unique. If we replace $V$ by $VP$, where $P$ is any QCA that maps $\sigma_i^z$ to a linear combination of products of $\sigma_i^z$'s (we refer to such QCA $P$ as \emph{generalized permutations}, see Appendix~\ref{app:ALTtoQCA}), then $V^\dagger H V$ also has $\{ \sigma_i^z \}$ as LIOMs.

    Moreover, using the property that paths of non-degenerate MBL Hamiltonians induce paths of unitaries that relate their LIOMs, as described above, it follows that if $H$ and $H'$ are in the same short-range entangled MBL phase, then there exists a QCA $V'$ which can be continuously connected to $V$ such that $(V')^\dagger H' V'$ has $\{ \sigma_i^z \}$ as LIOMs.
    We see that every short-range entangled MBL Hamiltonian can be prepared from $\{ \sigma_i^z \}$ by a QCA, and those Hamiltonians in the same phase can be prepared by QCA in the same connected component.
    It follows that short-range entangled MBL phases can be classified by connected components of QCA (modulo generalized permutations).

    For the purposes of the classification, it is convenient to allow \emph{stabilization} by the stacking (tensor product) of additional ancillae. The physical picture is that, for any system, there is an environment of additional degrees of freedom, and the classification of ALT phases (or MBL phases more generally) should be stable to the presence of such ancillae. Formally, the QCA \(V_0\) and \(V_0 \otimes \mathbbm{1}\) are said to be stably equivalent, and are considered to prepare Hamiltonians in the same phase~\cite{Gross2012QCA,Hastings2013torus,Piroli2020qcatn,Gong2020MPU}. The space of stable equivalence classes of \(d\)-dimensional QCA is denoted \(\QCA_d\), and its connected components \(\pi_0(\QCA_d)\) form an Abelian group under tensor products of QCA~\cite{Gross2012QCA,Freedman2022group}. Short-range entangled localized phases are classified by the quotient of this group by generalized permutation QCA.
    
    The QCA formalism also lets us straightforwardly identify which phases are anomalous. For any simultaneous eigenstate \(\ket{0}\) of all \(\sigma^z_j\), \(V \ket{0}\) is an eigenstate of \(H\). This defines a map \(e_0(V) = V\ket{0}\) from \(\QCA_d\) to the space of invertible states~\cite{Kitaev2013SRE} (short-range entangled states) in \(d\)-dimensions, \(\inv_d\). Our earlier observation that MBL Hamiltonians in the same phase have the same eigenstate order may be improved to the statement that there is an induced homomorphism from the classification of QCA to the classification of invertible states
    \begin{equation}
        e^*_0: \pi_0(\QCA_d) \to \pi_0(\inv_d),
    \end{equation}
    where addition of states is again defined by stacking.
    The crucial observation is that the homomorphism \(e_0^*\) is, in general, neither surjective nor injective. The lack of surjectivity means that not all invertible states can arise as eigenstates of commuting models (for instance, chiral states~\cite{Kitaev2006honeycomb,Kapustin2020hall}). The lack of injectivity (even accounting for equivalence by permutations) precisely means that nontrivial ALT phases exist. Indeed, ALT phases may be identified as the nontrivial phases in the kernel of \(e^*_0\).

    Ultimately, this homomorphism is just a mathematical formalization of the situation illustrated in \autoref{fig:MBLGapped}. MBL phases are a subspace of gapped phases (and similarly, MBL phases with short-ranged entanglement are a subspace of gapped phases with short-ranged entanglement). The inclusion map from MBL phases to gapped phases induces a homomorphism between the classifications thereof, which is just to say that each MBL phase belongs to a definite gapped phase and the stacking operation on both sides of the map is the same. Some gapped phases have no MBL representatives, so that the homomorphism is not surjective. Further, distinct MBL phases might belong to the same gapped phase, so that the homomorphism is not injective.

\subsection{Driven phases}
\label{sec:DrivenPhases}

    Most examples of ALT phases in the literature are driven systems~\cite{Titum2016,Po2016QCA,Roy2017loops,Nathan2019AFI,Friedman2022qpedge,Lapierre2022TLI,Long2021ALTP}. 
    We specialize to the case of multi-tone time dependence~\cite{Ho1983,Else2019longlived,Long2022QPMBL} (which includes Floquet driving as a special case), such that the Hamiltonian can be parametrized as
    \begin{equation}
        H(t) = H(\theta_1(t), \ldots, \theta_n(t)),
    \end{equation}
    where \(\theta_k(t) = \omega_k t\) is a drive phase defined modulo \(2 \pi\). 
    We show that short-range entangled MBL phases with \(n\)-tones can be classified using the higher homotopy groups of the space of QCA, denoted \(\pi_n(\QCA_d)\).
    
    It is convenient to assemble the drive phases into a vector \(\Bt_t = (\theta_1(t),\ldots, \theta_n(t))\). Then we say that the system is MBL with short-ranged entanglement if it has a locality preserving \emph{generalized Floquet decomposition} (\(\hbar =1\))~\cite{Else2019longlived,Long2022QPMBL}
    \begin{equation}
        U(t) = V(\Bt_t) \exp\bigg[ -i t \bigg( \sum_j h_j \sigma^z_j + \cdots \bigg) \bigg] V(\Bt_0)^\dagger.
        \label{eqn:generalized_Floquet}
    \end{equation}
    Here, \(U(t)\) is the time evolution operator, the Floquet Hamiltonian \(H_F = \sum_j h_j \sigma^z_j + \cdots\) is trivially MBL, and the micromotion operator \(V(\Bt)\) is a family of QCA parametrized by \(\Bt\). Thus, in complete analogy with the static case, we have a correspondence between multi-tone-driven localized phases and parametrized families of QCA, \(V(\Bt)\), up to continuous deformations. That is, homotopy classes of continuous maps from the \(n\)-dimensional torus to \(\QCA_d\). This reduces to the \(n\)th homotopy group \(\pi_n(\QCA_d)\) when ignoring \emph{weak invariants}~\cite{Kitaev2009periodic,Roy2017periodic} corresponding to low dimensional cross sections of the torus. Collapsing some of these cross-sections to a point reduces the torus to a sphere, and the \(n\)th homotopy group precisely corresponds to families \(V(\Bt)\) parametrized by a variable on a sphere, up to continuous deformations.
    
    Once more, there is an induced homomorphism
    \begin{equation}
        e^*_0 : \pi_n(\QCA_d) \to \pi_n(\inv_d).
    \end{equation}
    The group \(\pi_n(\inv_d)\) can be interpreted as the classification of individual quasienergy states (states for which time evolution synchronizes with the drive). The map \(e_0^*\) identifies the eigenstate order associated to individual quasienergy states, and ALT phases are again identified as the nontrivial elements of \(\ker e^*_0\).
    
    As in the case of static systems, the parametrization by QCA is not unique, so that the classification of ALT phases is properly a quotient of \(\ker e^*_0\). Further details are provided in Appendix~\ref{app:ALTtoQCA}.

\section{Classification}
\label{sec:Classification}

    The classification of static ALT phases in \(d\) dimensions is determined by the classification of QCA---\(\pi_0(\QCA_d)\). There has been great recent progress in classifying QCA in spin systems---a complete classification now exists in one~\cite{Gross2012QCA} and two~\cite{Freedman2020higherd} dimensions, and several nontrivial phases are known in three dimensions~\cite{Haah20233d,Shirley2022semionQCA,Fidkowski2023pumping}. The classification of ALT phases inherits this progress.
    
    In this section, we use these known facts about the classification of QCA, and argue for new results regarding the higher homotopy groups of QCA, to compute the classification of both static and driven MBL phases in low dimensions. 
    Namely, we propose and motivate a conjecture that QCA form an \(\Omega\)-spectrum, which can be viewed as a bulk-boundary correspondence between driven MBL phases and phases with one less drive in one lower dimension (\autoref{sec:OmegaSpectrum}).
    This conjecture is supported by an explicit construction mapping QCAs in \(d-1\) dimensions to loops of QCAs (driven ALT phases) in one higher dimension, called the \emph{swindle map} (\autoref{sec:Swindle}).
    Assuming the \(\Omega\)-spectrum conjecture is true, we calculate the complete classification of \(n\)-tone-driven ALT phases in \(d\)-dimensions with \(d-n \leq 2\) (\autoref{sec:ComputeClassification}).

    \subsection{Loops and \texorpdfstring{\(\Omega\)}{Omega}-spectra}
    \label{sec:OmegaSpectrum}
    
    In terms of the correspondence between MBL phases and QCA in \autoref{sec:MBLphases}, driven systems can be understood as relating to loops or tori of QCA. Higher homotopy groups of QCA have received little attention in the literature (though, Ref.~\cite{Roy2017loops} deals with essentially the same context in a different langauge), and these will be the subject of most of our new mathematical results. In particular, we propose that QCA form a mathematical object known as an \(\Omega\)-spectrum~\cite{Kitaev2006honeycomb,Kitaev2013SRE,Kitaev2015SPT,Xiong2018minimalist,Gaiotto2019gencohomology}, which in more plain language means that there is a bulk-boundary correspondence between \(n\)-tone driven phases in \(d\) dimensions and \((n-1)\)-tone driven phases in \(d-1\) dimensions.
    
    We motivate this proposal by recalling that, for Floquet systems \((n=1)\), there is another correspondence, distinct to that of \autoref{sec:DrivenPhases}, known between localized phases and QCA. Reference~\cite{Po2016QCA} showed that a periodically driven localized phase can be classified by its action on the edge when placed on open boundary conditions. This action is described by a QCA of one dimension lower, \(\pi_0(\QCA_{d-1})\). On the other hand, both our construction and Ref.~\cite{Roy2017loops} show that Floquet localized phases should be classified by \(\pi_1(\QCA_d)\). We propose that these groups are isomorphic,
    \begin{equation}
        \pi_0(\QCA_{d-1}) \cong \pi_1(\QCA_d).
        \label{eqn:pi1_pi0}
    \end{equation}
    
    The relationship~\eqref{eqn:pi1_pi0} is reminiscent of a widely held conjecture that invertible states form an \(\Omega\)-spectrum~\cite{Kitaev2006honeycomb,Kitaev2013SRE,Kitaev2015SPT,Xiong2018minimalist,Gaiotto2019gencohomology}. 
    In topology, an \(\Omega\)-spectrum refers to a sequence of topological spaces \((Y_d)_{d \in \Z}\) together with (weak) homotopy equivalences (that is, continuous maps which induce isomorphisms on all homotopy groups)
    \begin{equation}
        Y_{d-1} \simeq \Omega Y_{d}
    \end{equation}
    where \(\Omega Y_{d}\) is the space of loops in \(Y_d\) which begin and end at a specified base point. That is, \(\Omega Y_d\) is the space of maps \(\gamma:[0,1] \to Y_d\) such that \(\gamma(0)=\gamma(1) = *_d \in Y_d\) is the base point of \(Y_d\). Thus, continuous paths in \(\Omega Y_d\) correspond to homotopies of loops, and so \(\pi_0(Y_{d-1}) \cong \pi_0(\Omega Y_d) \cong \pi_1(Y_d)\). In fact, a similar equation holds for all higher homotopy groups, \(\pi_{n-1}(Y_{d-1}) = \pi_{n}(Y_d)\). Iterating this equation shows that the homotopy groups \(\pi_n(Y_d)\) only depend on \(q = d-n\). Thus, knowing all the connected components \(\pi_0(Y_d)\) already provides the classification of all higher homotopy groups.
    
    \(\Omega\)-spectra play a special role in topology, as they serve to define \emph{generalized cohomology theories}~\cite{Kitaev2013SRE,Kitaev2015SPT,Gaiotto2019gencohomology,Xiong2018minimalist}. Given the prevalence of cohomology classifications in the study of invertible topological states, it should be plausible that \emph{generalized} cohomology may find application in this problem, as indeed it does. 
    Explicitly, the \(d\)th generalized cohomology group of a space \(X\) with coefficients in the \(\Omega\)-spectrum \((Y_d)_{d\in \Z}\) is defined to be
    \begin{equation}\label{eqn:GeneralizedCohomology}
        h^d_Y(X) = \ho{X}{Y_d}
    \end{equation}
    (where \(\map{X}{Y_d}\) is the space of continuous functions from \(X\) to \(Y_d\)).
    
    The hypothesis that invertible states form an \(\Omega\)-spectrum implies immediately that \(\pi_n(\inv_d) = \pi_{n-1}(\inv_{d-1})\). Further, it is believed that the classification of symmetry-enriched invertible states with symmetry group \(G\) are completely classified by the generalized cohomology of the classifying space \(BG\) [i.e.,~\(h^d_{\inv}(BG)\)]~\cite{Kitaev2013SRE,Kitaev2015SPT,Gaiotto2019gencohomology,Xiong2018minimalist}.

    We conjecture that the spaces of quantum cellular automata on the integer lattice \(\Z^d\) stabilized by adding ancillae, \(\QCA_d\), form an \(\Omega\)-spectrum,
    \begin{equation}
        \QCA_{d-1} \simeq \Omega \QCA_d \implies
        \pi_{n-1}(\QCA_{d-1}) \cong \pi_n(\QCA_d).
        \label{eqn:Omega_spectrum}
    \end{equation}
    The base point for the loops in \(\Omega \QCA_d\) is the identity.
    For \(d < 0\), we define \(\QCA_d = \Omega^{-d} \QCA_0\), making the negative degrees of this statement trivial.
    
    We unravel the physical interpretation of Eq.~\eqref{eqn:Omega_spectrum} in terms of a bulk-boundary correspondence through a picture of \emph{pumping}~\cite{Kitaev2013SRE,Gaiotto2019gencohomology,Xiong2018minimalist}. Consider a loop \(\gamma \in \Omega\QCA_d\). If the path \(\gamma(t)\) is differentiable, it can be realized by Hamiltonian evolution under some \(H(t)\). Then we can restrict \(\gamma\) to open boundary conditions by taking the partial trace of \(H(t)\) outside of the specified boundaries. The path of QCA generated by this partial trace, \(\tilde{\gamma}(t)\), reproduces the action of \(\gamma(t)\) far from the boundary. In particular, \(\tilde{\gamma}(1)\) acts as the identity on any local operator which is supported far from the boundaries. However, on the \((d-1)\)-dimensional edge, \(\tilde{\gamma}(1)\) may have a nontrivial action described by some \(v \in \QCA_{d-1}\)~\cite{Po2016QCA}. We say that \(\gamma\) pumps \(v\) across the system. Equation~\eqref{eqn:Omega_spectrum} hypothesizes that the loop space \(\Omega\QCA_d\) is characterized (up to continuous deformations) entirely by the pumped QCA \(v\).
    
    Indeed, continuous deformations of the bulk loop \(\gamma\) produce continuous deformations of \(v\)~\cite{Harper2017bosonsandfermions}, so the topological invariants of \(v\) may also be considered invariants for \(\gamma\). The difficult steps needed to prove Eq.~\eqref{eqn:Omega_spectrum} are to show the following:
    \begin{enumerate}
        \item\label{itm:vOccurs} all \(v \in \QCA_{d-1}\) can occur as edge actions of some \(\gamma\in \Omega\QCA_d\), and
        \item\label{itm:HomotopyEquiv} every topological feature of the loop space is captured by the edge action.
    \end{enumerate}

    \subsection{The swindle map}
    \label{sec:Swindle}
    
    The \emph{swindle map} \(S: \QCA_{d-1} \to \Omega \QCA_{d}\) demonstrates that \ref{itm:vOccurs} is true. Given some \(v \in \QCA_{d-1}\), the swindle \(S(v,t) \in \Omega\QCA_d\) pumps \(v\) to the edge.

    \begin{figure}
        \centering
        \includegraphics{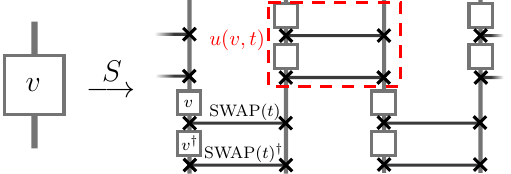}
        \caption{The swindle map \(S: \QCA_{d-1} \to \Omega \QCA_{d}\) maps a QCA \(v\) to a loop of QCA \(S(v,t)\) in one higher dimension. \(S(v,t)\) is a finite depth circuit built from \(v\), \(v^\dagger\) (squares), and a path from the identity to a swap gate, \(\mathrm{SWAP}(t)\) (line with crosses).}
        \label{fig:swindle}
    \end{figure}
    
    We can construct the swindle through a finite depth quantum circuit as follows~\cite{Kitaev2013SRE,Arrighi2011localizability,Arrighi2019qcareview,Farrelly2020qca} (\autoref{fig:swindle}). Let \(\mathrm{SWAP}_{x,x+1}(t)\) be a path of finite depth circuits from the identity to the swap gate between two \((d-1)\)-dimensional lattices.
    Then the path of QCA on a two-layer system
    \begin{equation}
        u_{x,x+1}(v,t) := 
        v_x  \mathrm{SWAP}_{x,x+1}(t) v_x^\dagger \mathrm{SWAP}_{x,x+1}(t)^\dagger 
        \label{eqn:uxxp1}
    \end{equation}
    interpolates between the identity and \(v_{x} v_{x+1}^\dagger\), where \(v_x\) acts as \(v\) on the \((d-1)\)-dimensional surface \(\{x\}\times \Z^{d-1}\). Indeed, \(u_{x,x+1}\) is a finite depth circuit, as we see by replacing the gates in the first \(\mathrm{SWAP}\) by their conjugation by \(v_x\). This operation preserves locality, so the result is still a finite depth circuit. Finally, the swindle \(S(v,t)\) is the brickwork circuit of \(u(v,t)\) gates shown in \autoref{fig:swindle}.
    
    The swindle \(S(v,t)\) demonstrates the pumping of the QCA \(v\). On an length \(L\) (even) system, the first layer of \(u\) gates implements the QCA \(v_0 v_1^{\dagger} v_2 v_3^\dagger \ldots v_L^\dagger\). The second layer of \(u\) gates cancels all the \(v\) factors in the bulk of the system, but leaves dangling unitaries \(v_0 v_L^\dagger\) acting on the left and right boundaries. The unitary \(v\) has been pumped to the left of the system, leaving behind its inverse \(v^\dagger\) on the right, thus proving \ref{itm:vOccurs}.
    
    The more precise statement capturing both \ref{itm:vOccurs} and \ref{itm:HomotopyEquiv} is that there is a homotopy equivalence between \(\QCA_{d-1}\) and \(\Omega \QCA_d\). That is, a map which---up to smooth deformations---has an inverse. 
    We propose that the swindle \(S\) is such an equivalence, and can sketch a proof as follows. Indeed, the homotopy inverse can be constructed through the pumping action of the loop \(\gamma \in \Omega\QCA_d\). Restrict \(\gamma\) to open boundaries (with a fixed choice of boundary for all \(\gamma\)), and define \(\bar{S}(\gamma) = v\), where \(v\in \QCA_{d-1}\) is the edge action induced by the restricted loop. It is straightforward to show that \(\bar{S}\) can be chosen such that \(\bar{S} S(v) = v\). Showing that \(S \bar{S}\) can be continuously deformed to the identity is more involved, but we present an explicit deformation in Appendix~\ref{app:Omega_spectrum}.
    
    The more formal argument that \(S\) is a homotopy equivalence for strictly local QCA is presented in Appendix~\ref{app:Omega_spectrum}, assuming that any loop of QCA can be restricted to open boundary conditions.
    Making this argument rigorous requires verifying that the restriction to open boundaries always exists, which we leave for future work.
    (Our previous argument involving a Hamiltonian evolution is, rigorously, unsatisfying because it allows for exponential tails in the QCA to develop parallel to the boundary for \(d>1\). Physically, this is not a problem, but it is an obstacle for proofs.)

    \subsection{Computing the classification}
    \label{sec:ComputeClassification}
    
    The \(\Omega\)-spectrum condition [Eq.~\eqref{eqn:Omega_spectrum}] implies that the bulk-boundary correspondence defined by pumping completely captures all multi-tone phases. In particular, iteratively applying Eq.~\eqref{eqn:Omega_spectrum}, the classification of driven localized phases in any dimension is reduced to either the classification of static QCA, or a driven classification of zero-dimensional QCA,
    \begin{equation}
        \pi_n(\QCA_d) \cong
        \left\{
        \begin{array}{l l}
            \pi_0(\QCA_{d-n}) & \quad\text{for } d\geq n, \\
            \pi_{n-d}(\QCA_0) & \quad\text{for } d < n.
        \end{array}
        \right.
    \end{equation}

    The connected components \(\pi_0(\QCA_{d})\) are known (for spin systems) for \(d= 1\)~\cite{Gross2012QCA} and \(d=2\)~\cite{Freedman2020higherd}. There are sound conjectures for \(d=3\)~\cite{Haah20233d,Haah2021clifford,Shirley2022semionQCA}.
    The space \(\QCA_0\) is a stabilization of the projective unitary group by stacking ancillae. We recall existing classifications and compute the homotopy groups of \(\QCA_0\) in Secs.~\ref{sec:connected_components} and \ref{sec:zero_dim} respectively.
    
    Thus, we have all the homotopy groups \(\pi_n(\QCA_d)\) for spin systems with \(d-n \leq 2\), and can make a strong conjecture for \(d-n = 3\). The results are listed in \autoref{tab:homotopy_groups}, and the corresponding classification of localized phases (excluding permutation QCA) is listed in \autoref{tab:ALT_phases}. Comparing to the classification of invertible states shows that \emph{all} the nontrivial phases in the tables are ALT phases.
    
    \renewcommand{\arraystretch}{1.2}
    \begin{table}[t]
        \begin{center}
        \begin{tabularx}{\linewidth}{|C | s s s s s s s s s|} 
         \hline
         \(q\) & \(3\) & \(2\) & \(1\) & \(0\) & \(-1\) & \(-2\) & \(-3\) & \(-4\) & \(-5\) \\ \hline
         \(h_\QCA^q\) & Witt & \(0\) & \(\log \Q^\times\) & \(0\) & \(\Q/\Z\) & \(0\) & \(\Q\) & \(0\) & \(\Q\) \\ 
         \(h_{\inv}^q\) & \(0\) & \(\Z\) & \(0\) & \(0\) & \(0\) & \(\Z\) & \(0\) & \(0\) & \(0\) \\ 
         \hline
        \end{tabularx}
        \end{center}
        \caption{The homotopy groups \(\pi_n(\QCA_d) =: h_\QCA^q\) are a function of \(q = d-n\), as is \(\pi_n(\inv_d) =: h_{\inv}^q\). For spin systems \(h_\QCA^3\) is conjectured to be the Witt group of modular tensor categories~\cite{Haah20233d,Haah2021clifford,Shirley2022semionQCA,Haah2023invertible,Fidkowski2023pumping}; \(h_\QCA^2\) is trivial~\cite{Freedman2020higherd}; \(h_\QCA^1\) is isomorphic to \(\log\Q^\times\)~\cite{Gross2012QCA}; and \(h_\QCA^q\) for \(q \leq 0\) is trivial for even \(q\), \(\Q/\Z\) for \(q=-1\), and \(\Q\) for odd \(q \leq -3\).
        In all groups in the table, the group operation is addition.}
        \label{tab:homotopy_groups}
    \end{table}
    \renewcommand{\arraystretch}{1.0}
    \begin{table}[t]
        \begin{center}
        \begin{tabularx}{\linewidth}{|C | C C C C C C|} 
         \hline
         \diagbox[height=1.5\line]{\(d\;\;\)}{\(\;\;n\)} & \(0\) & \(1\) & \(2\) & \(3\) & \(4\) & \(5\) \\ \hline
         \(0\) & \(0\) & \(0\) & \(0\) & \(\Q\) & \(0\) & \(\Q\) \\ 
         \(1\) & \(0\) & \(0\) & \(\Q/\Z\) & \(0\) & \(\Q\) & \(0\) \\ 
         \(2\) & \(0\) & \(\log \Q^\times\) & \(0\) & \(\Q/\Z\) & \(0\) & \(\Q\) \\ 
         \(3\) & Witt & \(0\) & \(\log \Q^\times\) & \(0\) & \(\Q/\Z\) & \(0\) \\ 
         \hline
        \end{tabularx}
        \end{center}
        \caption{The classification of \(n\)-tone driven \(d\)-dimensional localized phases of spins with short-range entanglement.
        The classifying group differs from \(\pi_n(\QCA_d)\) for \((d,n) \in \{(0,1), (1,0)\}\) due to the non-uniqueness of the parametrization by QCA (Appendix~\ref{subapp:nonunique}). 
        \emph{All} nontrivial elements of this table correspond to ALT phases.}
        \label{tab:ALT_phases}
    \end{table}

    \subsubsection{Connected components}
    \label{sec:connected_components}

    In general, the connected components \(\pi_0(\QCA_{d})\) are not known. However, these groups have been calculated for \(d \in \{1,2\}\) (and \(d=0\), which is trivial), and the groups are known to be nontrivial in some higher dimensions. Note that while \(\pi_0(X)\) has no canonical group structure for general spaces \(X\), \(\pi_0(\QCA_{d})\) inherits a group structure from either composition of QCA or tensor product. The resulting group operations are the same.

    For one-dimensional QCA, the complete classification of \(\pi_0(\QCA_1)\) was obtained by Ref.~\cite{Gross2012QCA}, which also explicitly defined an invariant which diagnoses which connected component a given QCA belongs to. All one-dimensional QCA are path connected to a \emph{shift}---a QCA which acts as a translation on some subset of the degrees of freedom. For example, in a chain of qubits, each QCA is connected to a QCA which translates all qubits by \(r\) sites either right or left. For qubit chains, QCA are classified by the group \(\Z\).

    However, we \emph{stabilize} QCA by allowing ancillae of any Hilbert space dimension to be added to any site by tensor product. Thus, we must assume that the chain also has qutrit (three-state) degrees of freedom, and other higher level qudits. While a 4-level qudit can be decomposed into two qubits, and so does not give new invariants compared to the qubit system, qutrits and other prime qudits all produce new invariants. A shift of qutrits cannot be connected to a shift of qubits.

    The resulting classifying group is succinctly presented as
    \begin{equation}
        \pi_0(\QCA_1) \cong \log \Q^\times,
    \end{equation}
    where \(\Q^\times\) is the set of positive rationals with multiplication as the group operation, and the group operation in \(\log\Q^\times\) is addition. Connecting to our discussion, this can be written as
    \begin{equation}
        \log\Q^\times \cong \bigoplus_{p \text{ prime}} \Z \log p.
    \end{equation}
    Each prime Hilbert space dimension is labeled by \(\log p\), and the integer coefficients give the shift index for that qudit species.

    The classification of two-dimensional QCA on spheres was shown to be trivial in Ref.~\cite{Freedman2020higherd},
    \begin{equation}
        \pi_0(\QCA_2) \cong 0.
    \end{equation}
    There is some subtlety in this statement. For QCA on a general two-dimensional manifold---that is, the lattice is obtained from a triangulation of a manifold---the classification may be nontrivial. However, all invariants relate to one-dimensional cycles in the manifold, and so belong to a lower dimensional part of the classification. For appropriate choices of topology\footnote{This includes the topology used in Appendix~\ref{app:Omega_spectrum}.} on \(\QCA_2\) (where the manifold is \(\R^2\)) this allows all two-dimensional QCA on the plane to be connected by continuous paths.

    The classification of three-dimensional QCA is still open. It is conjectured that \(\pi_0(\QCA_3)\) is isomorphic to the Witt group of modular tensor categories~\cite{Haah20233d,Haah2021clifford,Shirley2022semionQCA,Haah2023invertible}---that is, it has the same group structure as stacking two-dimensional topological orders, modulo an equivalence relation which identifies phases which may share a gapped edge. Presently, Clifford QCA (which map Pauli operators to Pauli operators) have been classified~\cite{Haah2021clifford}, and the classification is consistent with the Witt group conjecture for \(\pi_0(\QCA_3)\), but it has not been formally verified that all these phases of Clifford QCA remain nontrivial as non-Clifford QCA. Nonetheless, the existence of nontrivial elements of \(\pi_0(\QCA_3)\) is strongly supported by the constructions of Refs.~\cite{Haah20233d,Shirley2022semionQCA}. The models of those references can, indeed, be identified with two-dimensional topologically ordered phases. In fact, the MBL Hamiltonians corresponding to those models exhibit that topological order (specifically, the associated anyonic excitations) on their two dimensional surfaces~\cite{Burnell20143fermi,Haah20233d,Shirley2022semionQCA,Haah2023invertible}.

    \subsubsection{Zero dimensions}
    \label{sec:zero_dim}

    The space of zero-dimensional QCA with finite Hilbert space dimension, say \(N\), is naturally identified with the projective unitary group,
    \begin{equation}
        \QCA_0(N) \cong PU(N).
    \end{equation}
    The complication in finding the homotopy groups of \(\QCA_0\) is the stabilization by adding ancillae. That is, the \(N \to \infty\) limit.

    \emph{Unstable invariants.}---
    We first recall the homotopy structure of projective unitaries for fixed \(N\). The homotopy groups of \(PU(N)\) for large but finite \(N\) are all known,
    \begin{equation}
        \pi_n(PU(N)) \cong
        \left\{
        \begin{array}{l l}
            \Z_N & \quad \text{for }n=1, \\
             0 & \quad \text{for }n<2N\text{ and even,} \\
            \Z & \quad \text{for }n<2N\text{ and odd.}
        \end{array}
        \right.
    \end{equation}
    There are even integral formulas for winding number invariants indexing the elements of these groups. These formulas involve first making a lift from projective unitaries to linear unitaries. Indeed, given any continuous map \(V:S^n \to PU(N)\) (where \(S^n\) is the \(n\)-dimensional sphere) with \(n\neq 2\) it is possible to find a lift \(\tilde{V}: S^n \to U(N)\) to the usual unitary group. 
    (For \(n=2\), we already have that \(\pi_2(PU(N)) \cong 0\), so we do not need any formula to compute the homotopy class.)
    That is, if \(q: U(N) \to PU(N)\) is the quotient map from the unitary group to the projective unitaries, we have \(q\tilde{V} = V\).\footnote{In general, the existence of such a lift is guaranteed by \(V^*(\omega) \in H^2[S^n,\Z]\) being trivial, where \(\omega \in H^2[PU(N),\Z]\) is the Euler class for \(q: U(N) \to PU(N)\) and \(V^*\) is the map induced by \(V\) on singular cohomology. But \(H^2[S^n,\Z] \cong 0\) for \(n\neq 2\), so the map always exists for our cases of interest.}
    Defining \(A = (\mathrm{d} \tilde{V}) \tilde{V}^\dagger\), where \(\mathrm{d}\) is the exterior derivative on \(S^n\), the winding number invariants are given by~\cite{Long2021ALTP,Yao2017winding}
    \begin{equation}
        W_n[\tilde{V}] = C_n \int_{S^n} \tr{A^{\wedge n}},
        \label{eqn:winding_def}
    \end{equation}
    where \(C_n\) is a constant independent of \(N\) and \(A^{\wedge n} = A \wedge \cdots \wedge A\) is the \(n\)-fold wedge product of \(A\) with itself. For even \(n\) we have \(W_n = 0\), while for odd \(n \geq 3\) the winding number is an integer which only depends on the homotopy class of \(V\), so we can denote them \(W_n[V]\). For \(n=1\), only
    \begin{equation}
        w_1[V] = W_1[\tilde{V}] \bmod N \in \Z_N
    \end{equation}
    acts as an invariant for \(V\). Indeed, redefining \(\tilde V \to e^{i \theta} \tilde{V}\) does not affect \(V\), but takes \(W_1 \mapsto W_1 + N\).

    \emph{Stable invariants.}---
    Moving to larger \(N\), the complication is that these winding numbers are invariant under the direct sum with an ancilla, \emph{not} the tensor product. Indeed, denoting  \(\mathbbm{1}_M\) as the identity in \(PU(M)\), we have from Eq.~\eqref{eqn:winding_def}
    \begin{equation}
        W_n[V \otimes \mathbbm{1}_M] = M W_n[V].
    \end{equation}
    More generally, for \(Q \in PU(M)\),
    \begin{equation}
        W_n[V \otimes Q] = M W_n[V] + N W_n[Q].
        \label{eqn:unstable_Wadd}
    \end{equation}
    To find invariants which are stable under the addition of ancillae, we simply need to rescale the winding numbers as
    \begin{equation}
        \widetilde{W}_n[V] = C_n \int_{S^n} \frac{\tr{A^{\wedge n}}}{\tr{\mathbbm{1}}}.
    \end{equation}
    These are no longer integer invariants, but rather \emph{rational} invariants (modulo 1 for \(n=1\)). Then Eq.~\eqref{eqn:unstable_Wadd} shows that \(\widetilde{W}_n[V\otimes Q] = \widetilde{W}_n[V] + \widetilde{W}_n[Q]\).

    \emph{Classification.}---
    At a more formal level, for all \(N,M \in \N\) we have injective homomorphisms
    \begin{align}\label{eqn:nosymm_directed}
        \iota_{N \to NM} : PU(N) &\to PU(NM) \\
        V &\mapsto V \otimes \mathbbm{1}_M \nonumber
    \end{align}
    which define a directed system of topological groups. \(\QCA_0\) is the direct limit
    \begin{equation}
        \QCA_0 = \varinjlim PU(N)
        \label{eqn:QCA0_lim}
    \end{equation}
    with respect to this system. (We require that the multiplicity \(n_p\) of every prime \(p\) in the factorization \(N = 2^{n_2} 3^{n_3} \ldots p^{n_p} \ldots\) approaches infinity, \(n_p \to \infty\).) The homotopy groups of \(\QCA_0\) can then be computed from the induced system
    \begin{align}\label{eqn:dir_pi}
        \iota^*_{N \to NM} : \pi_n(PU(N)) &\to \pi_n(PU(NM)) \\
        W_n[V] &\mapsto M W_n[V]. \nonumber
    \end{align}
    By rescaling the winding numbers, we find that
    \begin{align}
        \pi_n(\QCA_0) &\cong \varinjlim \pi_n(PU(N)) \\
        &\cong \left\{
        \begin{array}{l l}
            \Q/\Z & \quad\text{for }n=1, \\
            0 & \quad\text{for }n\text{ even,} \\
            \Q & \quad\text{for }n\text{ odd.}
        \end{array}
        \right.
    \end{align}
    The group operation is given by the usual addition of rational numbers, and corresponds to the tensor product of QCA.
    
    \emph{Limited ancillae.}---
    We obtained a classification by rationals because we allowed ancillae of any dimension to be appended to the system. Thus, all Hilbert space dimensions can occur in the denominator of \(\widetilde{W}_n\). If, instead, we considered only a limited collection of Hilbert space dimensions for the ancillae, we would obtain a different classification. In fact, this classification will be a useful warm-up for the case of QCA with symmetry.
        
    Say our available ancillae dimensions are given by
    \begin{equation}
        \mathcal{R} = \{\rho \in \N\,:\, \text{ancilla of dimension }\rho\text{ exists}\},
        \label{eqn:univ_triv}
    \end{equation}
    and the system is built out of degrees of freedom with these dimensions. Then \(\mathcal{R}\) forms a submonoid (similar to a group, except without inverses) of \(\N\) under multiplication, as if \(\mathbbm{1}_\rho\) and \(\mathbbm{1}_\sigma\) can be ancillae, then so can \(\mathbbm{1}_{\rho \sigma} = \mathbbm{1}_\rho \otimes \mathbbm{1}_\sigma\). \(\mathcal{R}\) also contains \(1\), which corresponds to adding no ancillae.

    A similar directed system to Eqs.~\eqref{eqn:nosymm_directed} and \eqref{eqn:QCA0_lim} defines the limit \(\QCA_0^{\mathcal{R}}\) describing the stabilization of QCA by the ancillae in \(\mathcal{R}\). The only difference is that now the system proceeds through increasing \(\rho \in \mathcal{R}\), rather than \(N \in \N\). Thus, not all denominators may occur in the rescaled winding numbers. Defining the \emph{algebraic localization} (not to be confused with localization in physical space)
    \begin{equation}
        \mathcal{R}^{-1} \Z := \left\{\frac{p}{q} \in \Q \,:\,p \in \Z, q \in \mathcal{R}\right\},
    \end{equation}
    then \(\mathcal{R}^{-1} \Z\) is an additive group (indeed, a ring) and we have
    \begin{equation}
        \pi_n(\QCA^{\mathcal{R}}_0) \cong \left\{
        \begin{array}{l l}
            (\mathcal{R}^{-1} \Z)/\Z & \quad\text{for }n=1, \\
            0 & \quad\text{for }n\text{ even,} \\
            \mathcal{R}^{-1} \Z & \quad\text{for }n\text{ odd.}
        \end{array}
        \right.
    \end{equation}

\section{Symmetry enrichment}
\label{sec:SymmetryEnrichment}

    Physical systems frequently express symmetries which are crucial to their phenomenology, such as \(U(1)\)-symmetry leading to charge conservation. It is well known in many contexts that the presence of symmetry can affect topological classifications. Hamiltonians that were trivial may become nontrivial if some symmetry is imposed (called \emph{symmetry protection}), or more generally phases without symmetry may fracture into several distinct phases if symmetry is introduced (called \emph{symmetry enrichment})~\cite{Chen20112DSPT,Levin2012spt,Wen2017review}.

    These notions are equally important for localized phases. We say two systems belong to the same symmetric localized phase if they can be deformed into one another without breaking some symmetry or delocalizing. If the short-ranged entangled MBL Hamiltonian \(H\) commutes with an on-site unitary Abelian symmetry, then it can still have stable and complete LIOMs~\cite{Chen20112DSPT,Levin2012spt,Zhang2022SPTentangler,Potter2016_nonabsymm,Protopopov2017su2}, and the correspondence with QCA continues to hold. Thus, \emph{symmetry-enriched ALT phases} (SALT phases, symmetric localized phases without eigenstate order) can be classified by symmetric QCA. That is, QCA which commute with the symmetry action. In fact, at the level of QCA, there is no need to restrict to Abelian groups, and for the rest of this section we will consider the symmetry action to be a continuous unitary on-site representation of any topological group.

    The classification of SALT phases, and symmetry-enriched localized phases more generally, inherits progress from the classification of symmetric QCA, just as in the case without symmetry. The most significant progress has been in one dimension~\cite{Gong2020MPU,Zhang2021U1,Zhang2023note}. Of course, many symmetric QCA (including finite depth circuits) are also known in dimensions higher than one, and some progress has been made on the classification here as well~\cite{Zhang2022SPTentangler}.

    In this section, we unify and generalize this progress in several directions. First, we argue that symmetric QCA also form an \(\Omega\)-spectrum, and thus that the classification of driven phases can be deduced from the classification of static phases, or zero-dimensional driven phases (\autoref{sec:LoopsWithSymmetry}). We then make a complete classification of symmetric QCA in one dimension, generalizing the classifications of Refs.~\cite{Gong2020MPU,Zhang2021U1,Zhang2023note} to arbitrary groups (\autoref{sec:SymmetryOneDimension}). The classification in zero dimensions is analogous to the case without symmetry, and is deferred to Appendix~\ref{app:homotopy_symmetry}. Together with the \(\Omega\)-spectrum conjecture, this gives a classification of symmetric QCA (and subsequently SALT phases) with \(d-n \leq 1\).

    Before proceeding to technical details, we summarize several important points regarding the classification of symmetric QCA.
    We remarked in \autoref{sec:OmegaSpectrum} that the classification of \(d\)-dimensional invertible states with symmetry group \(G\) is believed to be given by \(h^d_{\inv}(BG)\), the \(d\)th generalized cohomolgy group of the classifying space \(BG\) with coefficients in the space of invertible states [as defined in Eq.~\eqref{eqn:GeneralizedCohomology}]~\cite{Kitaev2015SPT,Xiong2018minimalist,Gaiotto2019gencohomology}. The corresponding conjecture for QCA is that \(G\)-symmetric QCA are classified by \(h_\QCA^d(BG)\).
    Indeed, given the \(\Omega\)-spectrum condition, this conjecture can be motivated by a general analysis of symmetry defects, which is relatively insensitive to details regarding whether the objects being classified are states or QCA~\cite{Xiong2018minimalist,Gaiotto2019gencohomology}.
    
    However, the conjecture that the symmetric classification is given by \(h_\QCA^d(BG)\) must be carefully interpreted.
    An important observation of previous literature is that the classification of symmetric QCA depends not only on the abstract symmetry group \(G\), but also on the representation of the symmetry action \(\rho\), even when that representation is on-site~\cite{Gong2020MPU,Zhang2021U1,Zhang2023note}. 
    Thus, to understand the general classification of symmetric QCA, the symmetry representation---and the representations of any ancillae used in the stabilization procedure---must be specified, and not just the group \(G\).

    In the analysis of this section, we will not assume that ancillae carrying any representation of \(G\) are necessarily available. Instead, we define a countable set of representations which occur in the ancillae,
    \begin{equation}
        \mathcal{R} := \{ \rho \in \mathrm{Rep}(G) \,:\, \text{some ancillae carry }\rho \},
        \label{eqn:univ_R}
    \end{equation}
    where \(\mathrm{Rep}(G)\) is the tensor category of unitary equivalence classes of finite-dimensional unitary representations of \(G\). (This can be replaced with projective representations if appropriate.)

    The set \(\mathcal{R}\) defines the stabilization procedure for symmetric QCA. We say that two QCA \(V_1\) and \(V_2\) which carry a group action by \(G\) are \(\mathcal{R}\)-equivalent if ancillae from \(\mathcal{R}\) can be appended to each of them such that they become equal, and carry the same representation. Denote the space of symmetric QCA in \(d\) dimensions modulo \(\mathcal{R}\)-equivalence by \(\QCA_d^{\mathcal{R}}\). In the literature, two QCA have been called \emph{strongly equivalent} if they are connected by a symmetric path modulo \(\mathcal{R}\)-equivalence~\cite{Gong2020MPU,Zhang2023note}, which we see is just the statement that they belong to the same path component of \(\QCA_d^{\mathcal{R}}\).
    
    The set \(\mathcal{R}\) is a multiplicative submonoid of \(\mathrm{Rep}(G)\) under tensor product. If \(\rho\) and \(\sigma\) are representations carried by some ancillae, then \(\rho \otimes \sigma\) is the representation of the ancilla formed by stacking both ancillae. \(\mathcal{R}\) also contains the trivial one-dimensional representation \(1\), which corresponds to appending no ancillae. This definition of \(\mathcal{R}\) generalizes Eq.~\eqref{eqn:univ_triv}, for which we can regard each ancilla as a representation of the trivial group.\footnote{One can compare \(\mathcal{R}\) to the closely related notion of a \emph{\(G\)-universe}~\cite{Greenlees1995equivariant}, which is defined similarly, except that a \(G\)-universe is usually required to be closed under direct sum, rather than tensor product.}
    
    Indeed, the ingredient which allows us to complete the classification of symmetric QCA in one-dimension is to recognize that the elements of \(\mathcal{R}\) can be treated analogously to the dimensions of available qudits.
    We can define a shift index as for one-dimensional QCA without symmetry, only rather than considering ratios of dimensions we deal with formal ratios of representations.
    The resulting classification is similar to that in \autoref{tab:homotopy_groups}, except that rational numbers \(p/q\) are replaced by such formal ratios, \(\rho/\sigma\). In particular, if ancillae carrying any finite dimensional representation can be added to the system, one-dimensional symmetric QCA are completely classified by the \emph{Grothendieck group} of projective representations of \(G\). 
    
    We summarize a few examples. The classification of QCA without symmetry~\cite{Gross2012QCA} can be regarded as a special case where \(G=0\) is the trivial group. In this case, projective representations are just (projective) vector spaces classified by their (nonzero) dimension \(N \in \N^\times\), which multiply under tensor product. The Grothendieck group consists of formal ratios of dimensions, and is just the group of positive rationals with multiplication, \(\Q^\times\). In \autoref{tab:homotopy_groups}, we express this as \(\log \Q^\times\) so that the group operation is addition. 
    Number-conserving, that is \(U(1)\)-symmetric, QCA were found to be classified by rational polynomials in Ref.~\cite{Zhang2021U1}. This also follows from the observation that projective representations of \(U(1)\) are classified by Laurent polynomials with positive integer coefficients, and that the Grothendieck group thereof precisely consists of rational polynomials. 
    Finally, the classification with finite \(G\) and all ancillae available collapses to the classification of SPT states, \(H^2[G;U(1)]\), times the shift index classification, \(\log \Q^\times\)~\cite{Gong2020MPU,Zhang2023note}. This occurs because the representation ring of a finite group has zero divisors, which causes many elements to become identified in the Grothendieck group.

    Finally returning to the conjecture for the general classification, we suspect that taking \(\mathcal{R} = \mathrm{Rep}(G)\) to include all representations of \(G\) will give
    \begin{equation}\label{eqn:BGclassification}
        \pi_{n}(\QCA_d^{\mathrm{Rep}(G)}) \cong h^{d-n}_{\QCA}(BG).
    \end{equation}
    While \(h^{d-n}_{\QCA}(BG)\) is difficult to compute in general, this conjecture appears to be compatible with the classification for \(d-n \leq 1\). 

    \subsection{Loops with symmetry}
    \label{sec:LoopsWithSymmetry}

    In this section, we extend the calculations of \autoref{sec:Swindle} to the case of symmetric QCA, and so justify the conjecture that the \(\Omega\)-spectrum condition
    \begin{equation}
        \QCA^{\mathcal{R}}_{d-1} \simeq \Omega\QCA^{\mathcal{R}}_d
    \end{equation}
    is satisfied.
    Indeed, we will show that if the swindle map \(S: \QCA_{d-1} \to \Omega\QCA_d\) is in fact a homotopy equivalence, then it also provides the homotopy equivalence for the case with symmetry.
    Together with \autoref{sec:SymmetryOneDimension} and Appendix~\ref{app:homotopy_symmetry}, the \(\Omega\)-spectrum condition gives us the classification of \(\pi_n(\QCA^{\mathcal{R}}_{d})\) for \(d-n \leq 1\).

    We show that the swindle map can be constructed so as to respect any on-site symmetry \(\rho\), and so we get a map
    \begin{equation}
        S^\rho : \QCA^{\mathcal{R}}_{d-1} \to \Omega\QCA^{\mathcal{R}}_d
    \end{equation}
    between symmetric QCA in \(d-1\) dimensions and the loop space of symmetric QCA in \(d\) dimensions.
    
    The swindle map is constructed in terms of a path \(\mathrm{SWAP}(t)\) from the identity map to a swap gate between two \((d-1)\)-dimensional layers with the same local Hilbert space dimension. Thus, in the context of QCA with an on-site \(G\)-action, we can repeat all the same arguments we made with the swindle map, provided we can find a path \(\mathrm{SWAP}^\rho(t)\) which commutes with the on-site representation \(\rho\).

    Consider two sites with Hilbert space \(\mathcal{H}_j \otimes \mathcal{H}_k\) and a representation \(\rho \otimes \rho\) of the group \(G\). It is a well known fact that the representation \(\rho \otimes \rho\) decomposes into a symmetric and alternating subspace. Indeed, we have
    \begin{equation}
        \mathrm{swap}_{jk}(1) \rho_g \otimes \rho_g = \rho_g \otimes \rho_g \mathrm{swap}_{jk}(1)
    \end{equation}
    for all \(g \in G\). Thus, all \(\rho_g \otimes \rho_g\) become block diagonal in an eigenbasis of the swap gate, \(\mathrm{swap}_{jk}(1)\). Decomposing
    \begin{equation}
        \mathcal{H}_j \otimes \mathcal{H}_k \cong \mathrm{Sym}(\mathcal{H}_j \otimes \mathcal{H}_k) \oplus \mathrm{Alt}(\mathcal{H}_j \otimes \mathcal{H}_k),
    \end{equation}
    where \(\mathrm{Sym}(\mathcal{H}_j \otimes \mathcal{H}_k)\) is the \((+1)\)-eigenspace of \(\mathrm{swap}_{jk}(1)\) and \(\mathrm{Alt}(\mathcal{H}_j \otimes \mathcal{H}_k)\) its \((-1)\)-eigenspace, we have
    \begin{equation}
        \rho\otimes \rho \sim \rho_{\mathrm{Sym}} \oplus \rho_{\mathrm{Alt}},
    \end{equation}
    where ``\(\sim\)'' is unitary equivalence.
    Then the path of operators
    \begin{equation}
        \mathrm{swap}^\rho_{jk}(t) = \mathbbm{1}_{\mathrm{Sym}} \oplus (e^{i \pi t} \mathbbm{1}_{\mathrm{Alt}})
    \end{equation}
    is continuous, connects \(\mathbbm{1}\) to the swap gate, and commutes with \(\rho_g\otimes \rho_g\) for all \(g \in G\) and \(t \in [0,1]\).

    Then, given some on-site group action \(\rho\) on \(\QCA_{d-1} \ni v\), we can construct a loop of QCA on a doubled Hilbert space
    \begin{equation}
        u^\rho_{x,x+1}(v,t) := 
        v_x  \mathrm{SWAP}^\rho_{x,x+1}(t) v_x^\dagger \mathrm{SWAP}^\rho_{x,x+1}(t)^\dagger,
    \end{equation}
    where \(\mathrm{SWAP}^\rho_{x,x+1}(t)\) is a product of \(\mathrm{swap}^\rho_{jk}(t)\) paths implementing a swap of the entire lattice. Multiplying these together gives an instance of the swindle map
    \begin{equation}
        S^\rho(v,t) = \left(\prod_{x \in 2\Z+1} u^\rho_{x,x+1}(v,t)\right) \left(\prod_{x \in 2\Z} u^\rho_{x,x+1}(v,t)\right),
    \end{equation}
    which respects the group action in the sense that
    \begin{equation}
        \left(\prod_x \rho_x\right) S^\rho(v,t) \left(\prod_x \rho^\dagger_x\right) = S^\rho(\rho v \rho^\dagger,t).
    \end{equation}
    That is, the group action on the \(d\)-dimensional loop agrees with the action on the \((d-1)\)-dimensional QCA.
    
    The pumping argument showing that \(S^\rho\) is a homotopy equivalence can also be applied without alteration, as it is still possible to restrict \(S^\rho\) to open boundaries. Indeed, \(S^\rho\) is composed of symmetric gates, so just removing gates which cross some interface in the lattice produces such a restriction.
    (The more formal argument in Appendix~\ref{app:Omega_spectrum} proceeded exclusively through interleaving finite depth circuits with swap gates, which can simply be replaced by symmetric swaps.)
    Thus, \(S^\rho\) gives a \(G\)-homotopy equivalence---a homotopy equivalence which respects the \(G\) action.

    In mathematical language, this makes \((\QCA_d)_{d\in\Z}\) a \emph{\(G\)-spectrum}, or more specifically a (naive) \(\Omega\)-spectrum with \(G\)-action. This defines a (Bredon) equivariant cohomology theory. On general grounds, it is known that the fixed point subspaces of the spectrum---QCA which commute with the group action---themselves form an \(\Omega\)-spectrum~\cite{Greenlees1995equivariant}. (An equivariant map must map fixed points to fixed points.)  The fixed point spaces are \(\QCA_d^{\mathcal{R}}\). We may also define \(\QCA_d^{\mathcal{R}}\) for negative \(d\) by \(\QCA_d^{\mathcal{R}} = \Omega^{-d}\QCA_0^{\mathcal{R}}\). Then we have
    \begin{equation}
        \pi_n(\QCA^{\mathcal{R}}_d) \cong \pi_0(\QCA^{\mathcal{R}}_{d-n}).
    \end{equation}
    More generally, there is a generalized cohomology theory
    \begin{equation}
        h^d_{\mathcal{R}}(X) := \ho{X}{ \QCA^{\mathcal{R}}_d}
    \end{equation}
    which classifies parametrized families of symmetric QCA.

    \subsection{Connected components in one dimension}
    \label{sec:SymmetryOneDimension}

    The study of the connected components of symmetric QCA has been considered in Refs.~\cite{Gong2020MPU,Zhang2021U1,Zhang2023note}, where complete classifications were found for the symmetry group \(U(1)\) and finite cyclic groups of prime order. In this section, we use some technology of representation theory to generalize these classifications to arbitrary topological groups \(G\) and countable collections of ancillae representations \(\mathcal{R}\).

    \subsubsection{Structure of one-dimensional symmetric QCA}

    Two one-dimensional QCA \(V_1\) and \(V_2\) are path equivalent if and only if \(V_1 V_2^\dagger\) is a finite depth circuit~\cite{Gross2012QCA,Farrelly2020qca}. Leveraging this fact simplifies the definition of invariants and proofs of classifications. Indeed, all one-dimensional QCA can be expressed as circuits of isometries~\cite{Cirac2017canform}. That is, possibly after coarse graining by blocking sites together, any QCA \(V\) on a one-dimensional chain with sites \(x\) and local Hilbert space dimensions \(N_x\) is equal to a product
    \begin{equation}
        V = \left( \prod_{x \in 2\Z+1} u_{x,x+1}^\dagger \right) \left( \prod_{x \in 2\Z} u_{x,x+1} \right)
    \end{equation}
    where
    \begin{equation}
        u_{x,x+1} \in U(N_x N_{x+1} \to N'_x N'_{x+1}).
    \end{equation}
    That is, \(V\) can be written as a finite depth circuit where the internal legs of the circuit do not necessarily have the same dimension as the external legs---the gates of the circuit are isometries, rather than unitaries. However, for the entire circuit to be unitary, we must have \(N_x N_{x+1} = N'_x N'_{x+1}\).
    
    Then \(V\) commutes with the on-site symmetry \(\prod_x \rho_x\) if and only if
    \begin{equation}
        u_{x,x+1} \rho_x(g) \rho_{x+1}(g) = v_x(g) v_{x+1}(g) u_{x,x+1}
        \label{eqn:symm_V}
    \end{equation}
    for all \(g \in G\) and \(x \in \Z\), where \(v_x(g) \in U(N'_x)\) is a unitary on the internal leg~\cite{Cirac2017canform}.
    
    The product of two consecutive \(v_x v_{x+1}\) is, by Eq.~\eqref{eqn:symm_V}, unitarily equivalent to the representation \(\rho_x \rho_{x+1}\). However, this only implies that \(v_x\) and \(v_{x+1}\) are \emph{projective} representations of \(G\). The classifications of Refs.~\cite{Gong2020MPU,Zhang2021U1,Zhang2023note} rely on the following observation: \(V\) is equal to a finite depth circuit with symmetric unitary gates (which implies they are connected by a symmetric path) and the same circuit geometry if and only if there are functions \(\phi_x:G \to U(1)\) such that \(e^{i \phi_x} v_x\) is unitarily equivalent to \(\rho_x\) for all \(x\). That is, the projective representation \(v_x\) can be lifted to an ordinary unitary representation, and that representation is precisely \(\rho_x\). A more concise statement of the condition is that \(v_x\) and \(\rho_x\) are isomorphic when both are regarded as projective representations.
    
    Indeed, suppose that there are unitaries \(w_x\) such that
    \begin{equation}
        e^{i \phi_x(g)}w_x v_x(g) w_x^\dagger =  \rho_x(g) \quad\text{for all }g \in G.
    \end{equation}
    From the fact that \(v_x v_{x+1}\) is unitarily equivalent to \(\rho_x \rho_{x+1}\), we conclude that \(\phi_x = -\phi_{x+1}\).
    Further, \(N_x = N'_x\) as \(w_x\) is unitary, so we can write \(V\) as a circuit of unitary gates
    \begin{equation}
        V = \left( \prod_{x \in 2\Z+1} u'^\dagger_{x,x+1} \right) \left( \prod_{x \in 2\Z} u'_{x,x+1} \right)
        \label{eqn:symm_decomp}
    \end{equation}
    with
    \begin{multline}
        u'_{x,x+1} = w_x w_{x+1} u_{x,x+1} \\
        = (e^{i \phi_x} w_x)(e^{i \phi_{x+1}} w_{x+1}) u_{x,x+1}
    \end{multline}
    such that
    \begin{align}
        u'_{x,x+1} \rho_x \rho_{x+1} &= w_x w_{x+1} v_x v_{x+1} u_{x,x+1} \\
        &= (e^{i \phi_x} w_x v_x w_x^\dagger) (e^{i \phi_{x+1}} w_{x+1} v_{x+1} w_{x+1}^\dagger)  \n
        &\qquad\qquad\qquad \times w_x w_{x+1} u_{x,x+1} \\
        &= \rho_x \rho_{x+1} u'_{x,x+1}
    \end{align}
    is symmetric.
    
    Conversely, if \(V\) can be written as a circuit with symmetric unitary gates in the same circuit geometry, as in Eq.~\eqref{eqn:symm_decomp}, then we must have that
    \begin{equation}
        \left( \prod_{x \in 2\Z+1} u'_{x,x+1} u^\dagger_{x,x+1} \right) \left( \prod_{x \in 2\Z} u_{x,x+1} u^{\prime \dagger}_{x,x+1} \right) = \mathbbm{1}.
    \end{equation}
    We see by, for instance, the fact that the operator entanglement across the \((x,x+1)\) bond in this product must vanish, that
    \begin{equation}
         u_{x,x+1} u^{\prime \dagger}_{x,x+1} = w_x w_{x+1}
    \end{equation}
    must be a tensor product of isometries \(w_x \in U(N_x \to N_x')\). Further, as the whole circuit is the identity, we have \(w_x^\dagger w_x = \mathbbm{1}\) and thus \(N_x' \geq N_x\) for all \(x\). However, \(N_x N_{x+1} = N'_x N'_{x+1}\), and these conditions can only be simultaneously satisfied if \(N_x = N'_x\) for all \(x\), and thus \(w_x\) is a unitary gate. Note that we did not reference symmetry here: two brickwork decompositions of \(V\) necessarily only differ by acting by unitaries on the internal legs, regardless of whether \(V\) is symmetric or the gates in either decomposition are symmetric.
    
    However, when \(u'\) is symmetric, we have
    \begin{align}
        v_x v_{x+1} u_{x,x+1} &= u_{x,x+1} \rho_x \rho_{x+1} \\
        &= w_x w_{x+1} u'_{x,x+1} \rho_x \rho_{x+1} \\
        &= w_x w_{x+1} \rho_x \rho_{x+1} u'_{x,x+1} \\
        &= (w_x \rho_x w^\dagger_x) (w_{x+1} \rho_{x+1} w^\dagger_{x+1}) u_{x,x+1},
    \end{align}
    so that \(v_x = e^{-i \phi_x} w_x \rho_x w^\dagger_x\), as we needed.

    Thus, an invariant for symmetric QCA in one dimension must somehow measure the difference between \(v_x\) and \(\rho_x\). We wish to use a formal ratio \(v_x / \rho_x\) as such an invariant. The next subsection explains this procedure.

    \subsubsection{Invariants}

    We aim to construct an invariant which measures the difference between \(v_x\) and \(\rho_x\). This is easier if both are regarded as projective representations, and so we will proceed with the construction treating them as such.

    Denote the set of unitary equivalence classes of finite-dimensional projective representations of \(G\) by \(\mathrm{Rep}_{\mathrm{proj}}(G)\). This becomes a commutative monoid when equipped with the tensor product operation.
    The result of the previous section is that the decomposition of \(V\) in Eq.~\eqref{eqn:symm_decomp} is a symmetric circuit if and only if
    \begin{equation}
        v_x = \rho_x
    \end{equation}
    as elements of \(\mathrm{Rep}_{\mathrm{proj}}(G)\) for all \(x\).

    However, \(v_x \in \mathrm{Rep}_{\mathrm{proj}}(G)\) cannot be used as an invariant: it is unstable to stacking ancillae, and presumably depends on position \(x\) and on re-partitioning the circuit geometry. These issues are dealt with through an a construction known as the \emph{algebraic localization} (which already appeared in a limited form in \autoref{sec:zero_dim}).
 
    The localization of \(\mathrm{Rep}_{\mathrm{proj}}(G)\) away from the submonoid \(\mathcal{R}\) is defined as the quotient
    \begin{multline}
        \mathcal{R}^{-1} \mathrm{Rep}_{\mathrm{proj}}(G) = \mathrm{Rep}_{\mathrm{proj}}(G) \times \mathcal{R}/\sim \\
        \text{where } (v_1, \rho_1) \sim (v_2, \rho_2) \iff \rho_2 v_1 \sigma = \rho_1 v_2 \sigma
    \end{multline}
    for some \(\sigma \in \mathcal{R}\). Denote the equivalence class of \((v, \rho)\) by \(v/\rho\). The composition rule \((v_1/\rho_1)(v_2/\rho_2) = v_1 v_2 / \rho_1 \rho_2\) makes \(\mathcal{R}^{-1} \mathrm{Rep}_{\mathrm{proj}}(G)\) a commutative monoid.
    This gives a precise meaning to the notion of a formal ratio.
    
    We propose that
    \begin{equation}\label{eqn:symm_invariant_1d}
        v_0/\rho_0 \in \mathcal{R}^{-1} \mathrm{Rep}_{\mathrm{proj}}(G),
    \end{equation}
    is a complete stable invariant for \(V \in \QCA_1^{\mathcal{R}}\). 
    
    Note that the invariant is well defined if we restrict to a particular circuit geometry. Choosing different circuit representations with the same geometry just conjugates each \(v_x\) by some unitary as \(w_x^\dagger v_x w_x\), as we saw in the last section. This, by definition, does not change the unitary equivalence class of \(v_x\).
    
    Next, we check stability to stacking ancillae. When appending additional ancillae to site \(x\) with representation \(\sigma_x\), \(V\otimes \mathbbm{1}\) acts trivially on the ancilla, so \(v_x\) just becomes \(v_x \otimes \sigma_x\). Then \(v_x \sigma_x/ \rho_x \sigma_x = v_x/\rho_x\) is unchanged. More generally, if we stack two QCAs with the same circuit geometry (but not necessarily the same dimensions of internal legs), we have that their putative invariants multiply. Thus, the map from QCAs to the proposed invariant is a homomorphism.
    
    Further, because \(v_x v_{x+1} = \rho_x \rho_{x+1} \in \mathcal{R}\), we have
    \begin{equation}
        \frac{v_x}{\rho_x} \frac{v_{x+1}}{\rho_{x+1}} = 1.
    \end{equation}
    So the invariant for \(x+1\) is just the inverse of that for \(x\), and has no new information. [In particular, the invariant always has an inverse for an actual QCA, even though \(\mathcal{R}^{-1} \mathrm{Rep}(G)\) is in general only a monoid.] Thus, we can just focus on \(x=0\), and henceforth drop the \(x\) index. Similarly, blocking three sites together (which maintains the brickwork circuit structure) does not change the invariant. This further implies that the invariant does not depend on the choice of brickwork geometry to express the circuit for \(V\), provided that it is nearest-neighbour. Indeed, given any two brickwork circuits for \(V\), just block sites until they have the same geometry. This does not change the invariant associated to either circuit, and once they are on the same geometry, we know that the invariants must agree.

    Now we must show that Eq.~\eqref{eqn:symm_invariant_1d} is, in fact, an invariant. That is, if we consider a symmetric path of QCA, \(V_t\), then the invariant is constant all along the path. Assume the QCA along the path have a maximum range, \(r\). Then we will write \(V_t\) as a circuit in a fixed geometry where the gates \(u_{x,x+1}(t)\) evolve continuously. Once we have done this, we have that
    \begin{equation}
        v_{xt} v_{(x+1)t} = u_{x,x+1}(t) \rho_x \rho_{x+1} u^\dagger_{x,x+1}(t)
    \end{equation}
    also evolves continuously, which gives that each \(v_{xt}\) is a continuous family of projective representations (that is, \(v_{xt}(g) \in PU(N_x)\) is continuous as a function of \((t,g) \in [0,1] \times G\) with the product topology). This evolution cannot change the unitary equivalence class of \(v_{xt}\)---one can construct a unitary relating different \(t\) by adiabatic continuation, which is always successful in a finite dimensional Hilbert space (even with accidental degeneracies~\cite{Kato1980}). Thus, the \(v_{x}/\rho_x\) invariant is, in fact, invariant along the path.
    
    To construct such a circuit decomposition of \(V_t\), first consider the Hamiltonian \( (i\partial_t V_t) V_t^\dagger = H_t\). As \(V_t\) has range \(r\), \(H_t\) can be written as a sum of terms supported on a set of diameter at most \(4r\). To see this, use that \(\|i\partial_t[V_t^\dagger a V_t, b]\| = \|[[a,H_t], V_t b V_t^\dagger]\| = 0\) for operators \(a\) and \(b\) further than \(r\) apart, and that all operators \(c\) further than \(2r\) from \(a\) can be written as \(V_t b V_t^\dagger\) for some \(b = V_t^\dagger c V_t\) at least \(r\) from \(a\). Then \([[a,H_t], c] = 0\) for all \(a\) and \(c\) at least \(2r\) apart. We can extract all terms in the Hamiltonian with support on site \(x\) as 
    \begin{align}
        H_{xt} &= H_t - \mathrm{tr}_x(H_t) \\
        &= \int \mathrm{d}U_x\, [H_t,U_x]U_x^\dagger
    \end{align}
    where \(\mathrm{tr}_x\) is the normalized partial trace, which we expressed as an integral over the normalized Haar measure on unitary operators on site \(x\). By directly substituting the integral formula, one finds that \([H_{xt}, c] = 0\) for all \(c\) supported at least \(2r\) from \(x\), so \(H_{xt}\) is supported in a ball of radius \(2r\) around \(x\), which has diameter at most \(4r\). Repeat this to iteratively define \(H_{x_nt} = \mathrm{tr}_{X_{n-1}}(H_t) - \mathrm{tr}_{X_n}(H_t)\) where \(X_n = \{x_i\}_{i=1}^n\) for some enumeration \((x_n)\) of the sites of the lattice with \(x_1 = x\). Then \(H_t\) is expressed as a telescoping sum
    \begin{equation}
        H_t = \sum_{n = 1}^\infty H_{x_n t},
    \end{equation}
    where each term is supported on a ball of radius \(2r\) around \(x_n\).

    Now we use \(H_t = \sum_{x \in \Z} H_{x t}\) to construct the circuit representation of \(V_t\). Block \(V_0\) into a brickwork circuit with blocks much larger than \(4r\), and only overlapping near their boundaries. Let \(H_{bt}\) be the sum of those \(H_{xt}\) which are supported entirely within a single block \(b\) in the first layer of the circuit. The isometry \(V_{bt}\) defined by solving  \(i \partial_t V_{bt} = H_{bt}V_{bt}\) with initial condition \(V_{b0}\) (a block of \(V_0\)) reproduces the action of \(V_t\) on operators \(a\) in the middle of the block, as we have the differential equation for \(a(t) = V_t a V_t^\dagger \),
    \begin{equation}
        i \partial_t a(t)  = [H_t, a(t)] = [H_{bt}, a(t)]
    \end{equation}
    where we used that no other term in the Hamiltonian has support in a radius \(r\) ball around the support of \(a\). Both \(a(t) = V_t a V_t^\dagger\) and \(a(t) = V_{bt} a V_{bt}^\dagger\) solve this equation with the same initial condition, \(a(0) = V_0 a V_0^\dagger = V_{b0} a V_{b0}^\dagger\), so \(a(t)\) is the same whether we evolve it with \(H_t\) or just \(H_{bt}\). We also see that evolution under \(H_{bt}\) does not couple opposite ends of the block---for instance, when evolving \(a\) supported near one boundary of \(b\), we can ignore any Hamiltonian terms in the middle of the block because \(H_{bt}\) reproduces \(V_t\) there and \(V_t\) does not couple the middle of the block to the distant boundary. In one dimension, this means that the evolved operator \(a\) cannot cross the region where we could drop the Hamiltonian terms. We then conclude that
    \begin{equation}
        V_t = \left( V_t \prod_{b} V_{bt}^\dagger \right) \left( \prod_{b} V_{bt} \right)
    \end{equation}
    is a brickwork circuit decomposition of \(V_t\), where the product over \(b\) is over the lower layer of blocks in the brickwork for \(V_0\). (Note that \(V^\dagger_{bt} V_{bt} = \mathbbm{1}\) when including both internal legs of the isometry.) Indeed, the \(\prod_b V_{bt}\) factor is clearly a product of disjoint isometries, each supported on the same blocks as the circuit for \(V_0\). The other factor has no effect in the middle of each block \(b\), and has a range much smaller than the size of \(b\). Thus, it is also a product of disjoint isometries supported near the boundaries of the blocks. This is the decomposition of \(V_t\) we sought.

    Finally, we show that the invariant is complete---two QCA \(V\) and \(V'\) with the same value of the invariant can be joined by a path of symmetric QCA. Because the values of the invariants achieved by some QCA form a group, it is sufficient to show that \(v/\rho = 1\) implies that \(V\) can be expressed as a unitary circuit with symmetric gates. By definition, \(v/\rho = 1\) if and only if there is a \(\sigma \in \mathcal{R}\) and unitary \(w\) such that
    \begin{equation}
        v \otimes \sigma = w^\dagger (\rho  \otimes \sigma) w.
    \end{equation}
    This is exactly the condition we showed resulted in \(V\) being a unitary circuit composed of symmetric gates in the last section, with the addition of an ancilla from \(\mathcal{R}\). So \(V\) is indeed a unitary circuit with symmetric gates.

    \subsubsection{Classification}

    We have an invariant valued in \(\mathcal{R}^{-1} \mathrm{Rep}_{\mathrm{proj}}(G)\).
    However, not all values of \(v/\rho\)
    can actually be achieved with a fixed \(\mathcal{R}\). Indeed, \(\mathcal{R}^{-1}\mathrm{Rep}_{\mathrm{proj}}(G)\) is not generally a group, but the connected components of symmetric QCA do form a group under tensor product. In the remainder of this section we characterize the isomorphism class of the classifying group in more detail.
    
    We observe that the only numerators \(v\) which may occur in the invariant \(v/\rho\) are those which occur as tensor factors in \(\sigma\) for some \(\sigma\in \mathcal{R}\). That is, the numerators belong to the \emph{saturation}
    \begin{multline}
        \hat{\mathcal{R}} = \{v \in \mathrm{Rep}_{\mathrm{proj}}(G)\,:\, \exists \sigma \in \mathrm{Rep}_{\mathrm{proj}}(G) \\
        \text{such that } v \sigma \in \mathcal{R}\}.
    \end{multline}
    Indeed, for \(v = v_0\), this element is \(\sigma = v_1\), which gives us \(\sigma v = \rho_0 \rho_1 \in \mathcal{R}\). The unitary similarity transformation relating these is \(u_{0,1}\).
    Note also that \(\hat{\mathcal{R}}\) is closed under multiplication by \(\mathcal{R}\), so the localization \(\mathcal{R}^{-1}\hat{\mathcal{R}}\) is well defined. Happily, \(\mathcal{R}^{-1}\hat{\mathcal{R}}\) is actually an Abelian group, with \((v/\rho)^{-1} = \rho \sigma/ v \sigma\), where the denominator is in \(\mathcal{R}\) and the numerator is in \(\hat{\mathcal{R}}\) [\(v\) demonstrates that there is an element of \(\mathrm{Rep}_{\mathrm{proj}}(G)\) such that \(v \rho \sigma \in \mathcal{R}\)].
    
    All the elements of \(\mathcal{R}^{-1}\hat{\mathcal{R}}\) may occur as invariants of some QCA. Suppose \(v_0/\rho_0 \in \mathcal{R}^{-1}\hat{\mathcal{R}}\). Then, as \(v_0\) belongs to \(\hat{\mathcal{R}}\), there are representations \(v_1 \in \mathrm{Rep}_{\mathrm{proj}}(G)\) and \(\rho_1 \in \mathcal{R}\) and a unitary \(u_{0,1}'\) such that
    \begin{equation}
        v_0 \otimes v_1 u_{0,1}' = u_{0,1}' \rho_1.
    \end{equation}
    We tensor in \(\rho_0\) to get
    \begin{equation}
        v_0 \otimes (v_1\otimes \rho_0) u_{0,1} = u_{0,1} \rho_0 \otimes \rho_1.
        \label{eqn:u01}
    \end{equation}
    Then we define a QCA \(V\) which commutes with the representation \(\prod_{x \in 2\Z} \rho_{0x} \otimes \rho_{1({x+1)}}\) by
    \begin{equation}
        V = \left( \prod_{x \in 2\Z+1} u_{x,x+1}^\dagger \right) \left( \prod_{x \in 2\Z} u_{x,x+1} \right),
    \end{equation}
    where \(u_{x,x+1}\) is a translation of \(u_{0,1}\) from Eq.~\eqref{eqn:u01}. 
    The invariant for the QCA \(V\) is, by construction, \(v_0/\rho_0\), or its inverse \(v_1 \rho_0/\rho_1\).

    Thus, the full classifying group is
    \begin{equation}
        \pi_0(\QCA_1^{\mathcal{R}}) \cong \log \mathcal{R}^{-1}\hat{\mathcal{R}},
    \end{equation}
    where we took a logarithm just to formally make the group operation addition. If \(\mathcal{R} = \mathrm{Rep}_{\mathrm{proj}}(G)\) consists of all projective representations of \(G\) this becomes the Grothendieck group of \(\mathrm{Rep}_{\mathrm{proj}}(G)\), which is by definition
    \begin{equation}\label{eqn:Grothendieck}
        \pi_0(\QCA_1^{\mathrm{Rep}_{\mathrm{proj}}(G)}) \cong \log \mathrm{Rep}_{\mathrm{proj}}(G)^{-1}\mathrm{Rep}_{\mathrm{proj}}(G).
    \end{equation}

    This classifying group can be realized as an extension of a projective part by a linear part. The linear part is a localization of the submonoid of \(\mathrm{Rep}_{\mathrm{proj}}(G)\) consisting of projective representations that admit lifts to linear representations. This submonoid is isomorphic to \(\mathrm{Rep}(G)/X_G\), where \(X_G\) is the group of units (invertible elements) in \(\mathrm{Rep}(G)\)---that is, one-dimensional representations---which all become identified as projective representations~\cite{Zhang2023note}.
    The projective part is classified by \(H^2[G;U(1)]\), as is well known.
    The induced homomorphism from evaluation on a symmetric product state, \(e_0^*\), just corresponds to sending a projective representation to its classifying cohomology class in \(H^2[G;U(1)]\), and thus reproduces the group cohomology classification of invertible states in one dimension~\cite{Chen20111dspt,Gong2020MPU}.

    The case of \(G = U(1)\) serves as a simple example~\cite{Zhang2021U1,Nathan2021hierarchy}. We take \(\mathcal{R} = \mathrm{Rep}_{\mathrm{proj}}(U(1)) \cong \N[x,x^{-1}]^\times/X_{U(1)}\) to consist of all representations of \(U(1)\), where the group of units \(X_{U(1)}\) consists of all monomials \(x^m\). (The quotient identifies representations which differ by a global phase.) Then \(\hat{\mathcal{R}} = \mathrm{Rep}_{\mathrm{proj}}(U(1))\) and the Grothendieck group is isomorphic to the multiplicative group of rational polynomials with positive coefficients, quotient by the group of units,
    \begin{equation}
        \mathcal{R}^{-1} \hat{\mathcal{R}} \cong \mathrm{Frac}\,\N[x]^\times/X_{U(1)}.
    \end{equation}
    Expressing the quotient in the additive notation, this is
    \begin{equation}
        \log \mathcal{R}^{-1} \hat{\mathcal{R}} \cong (\log \mathrm{Frac}\,\N[x]^\times)/\Z\log x.
    \end{equation}
    This is the complete classification of \(U(1)\)-symmetric QCA in one dimension~\cite{Zhang2021U1}.

    It has also been observed that including the regular representation for a finite group \(G\) in \(\mathcal{R}\), \(\rho_{\mathrm{reg}} \in \mathcal{R}\),  reduces the classification to just \(H^2[G;U(1)] \oplus \log \Q^\times\)~\cite{Gong2020MPU}. Indeed, in the current presentation it is clear that this is because \(\mathrm{Rep}(G)\) is not cancellative when \(G\) is a finite group. If \(\rho_{\mathrm{reg}}\) is the regular representation, then in \(\mathrm{Rep}(G)\) we have \(\rho_{\mathrm{reg}} \sigma = \mathrm{dim}(\sigma) \rho_{\mathrm{reg}}\), where \(\dim(\sigma)\) is the trivial representation with dimension \(\dim(\sigma)\). Then
    \begin{equation}
        \rho_{\mathrm{reg}} \sigma \mathrm{dim}(\sigma') =  \rho_{\mathrm{reg}} \sigma' \mathrm{dim}(\sigma),
    \end{equation}
    for any \(\sigma, \sigma' \in \mathrm{Rep}(G)\).
    Thus, \(\sigma/\sigma' =\mathrm{dim}(\sigma)/\mathrm{dim}(\sigma') \in \mathcal{R}^{-1}\hat{\mathcal{R}}\), and any two linear representations of the same dimension define the same element of the algebraic localization, which is then isomorphic to an extension of the projective part \(H^2[G;U(1)]\) by ratios of positive integers. The extension can be seen to split by observing that there is an \(H^2[G;U(1)]\) subgroup formed by ratios of a projective representation and a linear representation of the same dimension. Thus, we have \(\log \mathcal{R}^{-1} \hat{\mathcal{R}} \cong H^2[G;U(1)] \oplus \log \Q^\times\). [More correctly, \(\Q^\times\) here should itself be a localization of the natural numbers, \(\mathrm{dim}(\mathcal{R})^{-1} \mathrm{dim}(\hat{\mathcal{R}})\).]

\section{Models}
\label{sec:Models}

    The swindle map \(S(v,t)\) has been used until now as a theoretical tool. In fact, it also provides a useful construction for explicit models of \(d\)-dimensional driven (S)ALT phases, given a QCA (or family of QCA) \(v\) in one dimension lower. If \(v\) is nearest-neighbor, so is \(S(v,t)\) (in the \(\ell_\infty\) distance---it may involve diagonal couplings in the square lattice). Even better, a slightly modified model only involves coupling two qubits at a time in one dimension. Thus, these models are well suited to experimental realization in a variety of quantum simulation platforms~\cite{Eckardt2017,Gross2017,Kjaergaard2020}.

    Further, the swindle models are exactly solvable, in the sense that we can predict all observables for all times. Thus, they can explicitly demonstrate what, precisely, is anomalous about ALT phases. Namely, these models demonstrate anomalous edge dynamics---impossible in a short-range entangled MBL model of the same dimension as the edge---as a consequence of pumping nontrivial QCA. 
    
    That the swindle construction produces all \(n\)-tone phases relies on the conjectured \(\Omega\)-spectrum property, Eq.~\eqref{eqn:Omega_spectrum}, but the constructed phases being nontrivial and distinct does not.

    While the swindle map does not provide models of static ALT phases, we observe that Ref.~\cite{Haah2023invertible} implies that the edges of static ALT phases also cannot correspond to short-range entangled MBL models. Thus, anomalous edge dynamics are a generic feature of all ALT phases.

    In \autoref{sec:ChiralFloquet} we briefly describe a simple class of examples, building on the \emph{anomalous Floquet phases} of Refs.~\cite{Rudner2013,Po2016QCA}. These examples are built from swindles applied to a translation QCA, and give a simple picture illustrating the anomalous edge dynamics of ALT phases. Section~\ref{subapp:modified_swindle} constructs a more experimentally viable implementation of the swindle map, and \autoref{subsec:QP_pump} applies it to a very detailed investigation of a one-dimensional SALT phase with two drives and \(U(1)\)-symmetry, the \emph{quasiperiodic energy pump} (QP pump)~\cite{Long2021ALTP,Nathan2021pump,Long2022layer}. We include explicit expressions for the Hamiltonian and quantized observables which reveal the anomalous edge dynamics. The model we obtain through the swindle construction is relatively simple, at least in its operator content, paving the way for experiments on this phase.

    \subsection{Anomalous edge dynamics}
    \label{sec:ChiralFloquet}
    
    As a basic example, taking \(v\) in the swindle \(S(v,t)\) to be a one-dimensional translation---or shift---we construct a two-dimensional Floquet phase from a finite depth circuit composed of four layers of \(\mathrm{SWAP}\) gates. [A \(\mathrm{SWAP}\) gate conjugated by a translation as in Eq.~\eqref{eqn:uxxp1} is still a \(\mathrm{SWAP}\) gate.] This is precisely the model of the anomalous Floquet phase constructed in Refs.~\cite{Rudner2013,Po2016QCA}, which has circulating bulk orbitals and chiral edge dynamics---each Floquet cycle implements a unidirectional translation \(v\) on the left edge (\autoref{fig:ALT-QCA}). This behavior is impossible in a localized model in one dimension, both because particles are clearly not localized under this dynamics, and more fundamentally because no Hamiltonian evolution can implement a unidirectional translation.
    
    Iterating the swindle construction on the translation \(v\) produces a previously-unidentified phase of three-dimensional two-tone-driven MBL systems, with an explicit finite depth circuit model \(S(S(v,t),\nu t)\). This model pumps anomalous Floquet unitaries \(S(v,t)\) every period \(1/\nu\). To understand this, it is instructive to consider \(\nu \ll 1\), which allows us to consider \(\nu t\) as a slow adiabatic parameter deforming the period-1 evolution of \(S(v,t)\). For \(\nu t \approx 0\), the evolution everywhere in the lattice is a small-amplitude perturbation of a trivially localized model. The evolution deforms as \(\nu t\) grows, but remains localized. By \(\nu t \approx 1\) the evolution in the bulk is again close to trivial, but the evolution of the boundary is now by the anomalous Floquet unitary \(S(v,t)\). By \(\nu t \approx 2\) the boundary evolution is by \(S(v,t)^2\)---a Floquet evolution with twice the index of \(S(v,t)\) in \(\pi_1(\QCA_2)\), as witnessed by the circulation of orbitals on the surface~\cite{Titum2016,Po2016QCA,Nathan2021hierarchy}. The circulation at the two-dimensional surface continues to grow with time. As the circulation is an invariant of two-dimensional localized Floquet models, this edge evolution is impossible in a strictly two-dimensional localized model.

    \subsection{Modified swindle}
    \label{subapp:modified_swindle}

    If \(v \in \QCA_{d-1}\) has a range \(r=1\), then \(S(v,t)\) also has range \(1\). However, say in one dimension, it typically couples a site \(j\) to both \(j-1\) and \(j+1\) simultaneously. This makes experimental implementation of the swindle models in some quantum simulation platforms more complicated. Further, the swindle model is translationally invariant, and so does not enjoy the robustness of disordered systems to perturbations.

    We address both of these issues by defining a modified swindle map \(\tilde{S}(v,t)\). We split the loop defining the swindle unitary into three segments, \(t \in [0,\tfrac{1}{3})\), \(t \in [\tfrac{1}{3},\tfrac{2}{3})\), and \(t \in [\tfrac{2}{3},1]\). Within the first segment, sites are coupled along even numbered bonds along the first dimension of the lattice, while odd numbered bonds are coupled in the second segment. The third segment introduces on-site disorder to the model.

    \begin{figure}
        \centering
        \includegraphics{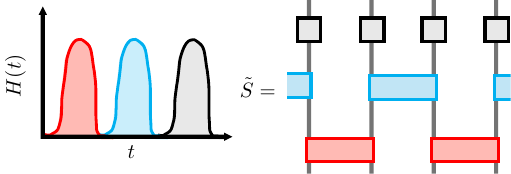}
        \caption{Sketch of the pulse schedule for the models defined by the modified swindle map \(\tilde{S}\). Sites are coupled along even bonds (red, \(t \in [0,\tfrac{1}{3})\)) and then along odd bonds (blue, \(t \in [\tfrac{1}{3},\tfrac{2}{3})\)). Finally, an on-site disorder potential may be applied to promote localization (gray, \(t \in [\tfrac{2}{3},1]\)).}
        \label{fig:modified_swindle}
    \end{figure}
    
    We pick arbitrary smooth functions \(\eta_{1,2,3}(t)\) which respectively map \(t \in [0,\tfrac{1}{3}]\), \([\tfrac{1}{3},\tfrac{2}{3}]\), and \([\tfrac{2}{3},1]\) to \(\eta \in [0,1]\). Sensible choices for these functions are monotonically increasing. We also choose disordered \((d-1)\)-dimensional evolution operators \(u_j^{\mathrm{dis}}(\eta)\), which are required to commute with any desired symmetry of the model. Then, on open boundary conditions with \(L \in 2\Z\) sites along the first dimension, we have
    \begin{equation}
        \tilde{S}(v,t) = \prod_{j=0}^{L/2-1} u_{2j,2j+1}(v, \eta_1(t))
    \end{equation}
    [where \(u_{j,j+1}\) is defined as in Eq.~\eqref{eqn:uxxp1}] for \(t \in [0,\tfrac{1}{3})\),
    \begin{equation}
        \tilde{S}(v,t) = v_0 \left(\prod_{j=1}^{L/2-1} u_{2j-1,2j}(v, \eta_2(t)) v^\dagger_{2j-1} v_{2j} \right) v^\dagger_L
    \end{equation}
    for \(t \in [\tfrac{1}{3},\tfrac{2}{3})\), and
    \begin{equation}
        \tilde{S}(v,t) = u_j^{\mathrm{dis}}(\eta_3(t)) v_0 v^\dagger_L
    \end{equation}
    for \(t \in [\tfrac{2}{3},1]\) (\autoref{fig:modified_swindle}).

    Away from the boundaries, \(\tilde{S}\) acts as the identity at \(t = \tfrac{2}{3}\). The LIOMs are then determined by \(u_j^{\mathrm{dis}}\). At the edges, \(\tilde{S}\) does not act as the identity, and instead acts as \(v\) or \(v^\dagger\).

    The Hamiltonian associated to the unitary \(\tilde{S}\) can be obtained by differentiation,
    \begin{equation}
        H(t) = [i \partial_t \tilde{S}(v,t)] \tilde{S}(v,t)^\dagger.
    \end{equation}

    \subsection{QP pump}
    \label{subsec:QP_pump}

    We carry through the modified swindle construction explicitly for the QP pump. This phase is of particular interest, as its edge states can be used to prepare highly excited non-classical states in quantum cavities~\cite{Long2022boosting}. Thus, an experimental realization of this model is especially desirable.

    Our model is based on a one-dimensional chain of qubits with open boundary conditions and length \(L \in 2\Z\). The construction for odd numbers of qubits is essentially the same. The resulting model will be qualitatively similar to that in Refs.~\cite{Kolodrubetz2018pump,Nathan2021pump}.

    The model has a \(U(1)\) symmetry generated by 
    \begin{equation}
        Q = \sum_j Q_j = \sum_j \tfrac{1}{2}(\mathbbm{1}+\sigma^z_j),
    \end{equation}
    so we choose a path to the swap gate which preserves this symmetry. Indeed, we may take
    \begin{equation}
        \mathrm{swap}_{j,j+1}(t) = \exp\left[-i \frac{2\pi}{8} t\, \vec{\sigma}_j \cdot \vec{\sigma}_{j+1}  \right]
    \end{equation}
    up to a global phase, where \(\vec{\sigma}_j \cdot \vec{\sigma}_{j+1} = \sigma^x_j \sigma^x_{j+1} + \sigma^y_j \sigma^y_{j+1} + \sigma^z_j \sigma^z_{j+1}\).

    The lower dimensional QCA to which the swindle map will be applied will actually be a loop of QCA in zero dimensions. That is, a loop of unitary matrices,
    \begin{equation}
        v_j(\theta) = \exp\left[ -i \frac{W}{2} \theta \sigma^z_j \right],
    \end{equation}
    where \(W \in \Z\) is the winding number and \(\theta \in [0,2\pi)\) parameterizes the loop. All \(W \in \Z\) give distinct winding numbers with \(U(1)\) symmetry. Without this symmetry, all odd (even) winding numbers would be connected by homotopies (Appendix~\ref{subapp:QPpump_invariant}).

    For the model to be quasiperiodically driven, we will take \(\theta(t) = \omega t \bmod 2\pi\) to be a linear function of \(t\) with \(\omega \not\in \Q\) irrational. 

    We take the disorder in the model to be on-site and diagonal in the \(\sigma^z\) basis,
    \begin{equation}
        u^\mathrm{dis}_j(t) = \exp\left[ -i t \delta_j \sigma^z_j \right],
    \end{equation}
    where \(\delta_j\) is, say, uniformly random in \([-\Delta, \Delta]\).

    The final elements of the modified swindle model we need are the functions \(\eta_{1,2,3}(t)\). These are essentially the pulse shapes which implement each unitary as a function of \(t\), and their specific choice should be tailored to the experimental platform in mind. We leave them unspecified, beyond requiring that they are smooth, monotonic, and their derivatives vanish at their end points. This ensures that there are no discontinuities in the pulses, which also assists in stabilizing localization~\cite{Else2019longlived,Long2022QPMBL}.

    The Hamiltonian is then defined piecewise as
    \begin{equation}
        H(t, \theta) = 
        \left\{
        \begin{array}{l l}
            \omega h^{1}(\eta_1(t),\theta) + \dot{\eta}_1(t) h^{2}(\eta_1(t), \theta), & t \in [0,\tfrac{1}{3}), \\
            \omega h^{3}(\eta_2(t),\theta) + \dot{\eta}_2(t) h^{4}(\eta_2(t),\theta), & t \in [\tfrac{1}{3},\tfrac{2}{3}), \\
            \omega h^5 + \dot{\eta}_3(t) h^{6}, & t \in [\tfrac{2}{3},1],
        \end{array}
        \right.
    \end{equation}
    where \(\dot{\eta}_k  = \partial_t \eta_k\) and \(H(t,\theta)\) has period \(T=1\) in \(t\) and period \(2\pi\) in \(\theta\).
    Defining \(\sigma^\pm_j = \sigma_j^x \pm i \sigma_j^y\), the individual terms for the first pulse are
    \begin{align}
        h^1(\eta_1, \theta) &= \sum_{j=0}^{L/2-1} h^1_{2j,2j+1}(\eta_1,\theta) \\
        h^2(\eta_1,\theta) &= \sum_{j=0}^{L/2-1} h^2_{2j,2j+1}(\eta_1,\theta),
    \end{align}
    where
    \begin{widetext}
    \begin{align}
        h_{j,j+1}^1(\eta_1, \theta) &= \frac{W}{4} \left[(1-\cos(\pi \eta_1))(\sigma^z_{j} - \sigma^z_{j+1}) 
        - \frac{i}{2} \sin(\pi \eta_1) e^{-i W \theta} \sigma_{j}^+ \sigma_{j+1}^{-} + \text{H.c.} \right], \text{ and} \\
        h_{j,j+1}^2(\eta_1,\theta) &= \frac{2\pi}{8} \left[ 
        \sin(W \theta) \sin(\pi \eta_1) (\sigma^z_{2j} - \sigma^z_{2j+1})
        + \frac{e^{-iW \theta}}{2}(1-\cos(W\theta) - i \sin(W\theta)\cos(\pi \eta_1)) \sigma_{j}^+ \sigma_{j+1}^{-} + \text{H.c.}\right].
    \end{align}
    For the second pulse, we have
    \begin{equation}
        h^3(\eta_2, \theta) = \frac{W}{2}(\sigma^z_0 - \sigma^z_{L-1}) +\sum_{j=1}^{L/2-1} h^3_{2j-1,2j}(\eta_2,\theta), 
        \quad\text{and}\quad
        h^4(\eta_2,\theta) = \sum_{j=1}^{L/2-1} h^2_{2j-1,2j}(\eta_2,\theta),
    \end{equation}
    where
    \begin{multline}
        h_{j,j+1}^3(\eta_2, \theta) = \frac{W}{4}\bigg[[1+\cos(W\theta) + \cos(2\pi \eta_2)(1-\cos(W\theta))]
        (\sigma^z_{j} - \sigma^z_{j+1}) \\
        - \frac{i}{2}e^{-i W\theta} [\sin(\pi \eta_2) (1-2 i \sin(W\theta)) + \sin(2\pi \eta_2)(1-\cos(W\theta))]
         \sigma_{j}^+ \sigma_{j+1}^{-} + \text{H.c.} \bigg].
    \end{multline}
    \end{widetext}
    Finally, the disorder pulse has
    \begin{equation}
        h^5 = \frac{W}{2}(\sigma^z_0 - \sigma^z_{L-1}), 
        \quad\text{and}\quad
        h^6 = \sum_{j=0}^{L-1} \delta_j \sigma^z_j.
    \end{equation}

    Additionally, rotating frame transformations generated by combinations of \(\sigma^z_j\) can be used to eliminate some of the complex time dependence of the coupling terms. For instance, making the transformation
    \begin{equation}
        \tilde{S}(t) \mapsto e^{i \omega t \frac{W}{2} \sum_j j \sigma^z_j} \tilde{S}(t)
    \end{equation}
    eliminates the \(e^{-i W\theta}\) factors in all the expressions above, at the cost of introducing a time-independent term 
    \begin{equation}
        H' = -\omega \frac{W}{2} \sum_{j=0}^{L-1} j \sigma^z_j
    \end{equation}
    to the Hamiltonian.
    That is, the time-dependent vector potential in the model can be traded for a linear scalar potential. On the other hand, if implementing time-dependent detunings \(\sigma^z_j - \sigma^z_{j+1}\) is more challenging in some architecture, moving to a frame co-rotating with those terms of the Hamiltonian moves this time dependence to the \(\sigma^+_j \sigma^-_{j+1}\) exchange term. 

    The main interesting feature of the QP pump is its edge states, which transfer energy from one drive to the other. 
    To compute the pumped power in the edge state, we evaluate the expectation value of 
    \begin{equation}
        P(t,\theta) = -\omega \partial_\theta H(t, \theta),
    \end{equation}
    which is the work done by the qubit chain on the drive with phase \(\theta\).
    We evaluate this operator in a mixture of quasienergy states, with the initial density matrix being
    \begin{equation}
        \rho_{0\pm, (L-1) \pm} = \frac{1}{2^L}(\mathbbm{1} \pm \sigma^z_0) (\mathbbm{1} \pm \sigma^z_{L-1}).
        \label{eqn:pump_state}
    \end{equation}
    This is maximally mixed in the bulk LIOM states, but has a definite state at both edges. The expectation value of the power operator is thus
    \begin{equation}
        \cexp{P(t,\theta)}_{\pm \pm} = \tr{P\tilde{S} \rho_{0\pm, (L-1) \pm} \tilde{S}^\dagger}
        \label{eqn:power_state}
    \end{equation}
    where we suppressed the dependence on \(t\) and \(\theta\) in the right hand side.

    The model is constructed so that \(\tau_j(t=0) = \sigma^z_j\), and the Heisenberg picture operator \(\sigma^{z\mathrm{H}}_j(t)\) is a function of \(\Bt\). We can thus find the edge LIOMs within a period by ``reverse Heisenberg evolution''~\cite{Else2019longlived} of the \(\sigma^z_0\) operator. We have, for \(t \in [0,\tfrac{1}{3})\)
    \begin{equation}
        \tau^z_0(t, \theta) = u_{0,1}(v(\theta),\eta_1(t)) \sigma^z_0 u_{0,1}^\dagger(v(\theta),\eta_1(t))
    \end{equation}
    and \(\tau^z_0 = \sigma^z_0\) otherwise. Equation~\eqref{eqn:power_state} then gives that 
    \begin{equation}
        \cexp{P(t,\theta)}_{\pm \pm} = \frac{1}{2^L} \tr{P (\mathbbm{1} \pm \tau^z_0) (\mathbbm{1} \pm \tau^z_{L-1})}.
        \label{eqn:power_LIOM}
    \end{equation}
    Further using that the power is a sum of traceless local operators, \(P = \sum_j P_{j,j+1}\), and that \(\tau^z_{0,L-1}\) have disjoint support, we can expand Eq.~\eqref{eqn:power_LIOM} as
    \begin{equation}
        \cexp{P(t,\theta)}_{\pm \pm} = \pm \frac{1}{2^L} \tr{P \tau^z_0} \pm \frac{1}{2^L} \tr{P \tau^z_{L-1}}.
    \end{equation}
    
    Calculating the LIOMs from the definition of \(u_{0,1}\) gives
    \begin{widetext}
    \begin{multline}
        u_{0,1}(v(\theta),\eta_1) \sigma^z_0 u_{0,1}^\dagger(v(\theta),\eta_1) 
        = \frac{1}{2}\bigg[ \sigma^z_0 + \sigma^z_1 + \frac{1}{2}[1+ \cos(W \theta) + \cos(2 \pi \eta_1)(1-\cos(W\theta))] (\sigma^z_0 - \sigma^z_1) \\
        +\frac{e^{-iW \theta}}{4}[2\sin(\pi \eta_1) \sin(W\theta_1) + i \sin(2\pi \eta_1) (1-\cos(W\theta))]
        \sigma^+_0 \sigma^-_1 + \text{H.c.} \bigg].
    \end{multline}
    Similarly, \(\tau^z_{L-1}(t,\theta)\) is given by \(\sigma^z_{L-1}\) except during the first pulse, which is determined by
    \begin{multline}
        u_{L-2,L-1}(v(\theta),\eta_1) \sigma^z_{L-1} u_{L-2,L-1}^\dagger(v(\theta),\eta_1) 
        = \frac{1}{2}\bigg[ \sigma^z_{L-2} + \sigma^z_{L-1} - \frac{1}{2}[1+ \cos(W \theta) + \cos(2 \pi \eta_1)(1-\cos(W\theta))] (\sigma^z_{L-2} - \sigma^z_{L-1}) \\
        -\frac{e^{-iW \theta}}{4}[2\sin(\pi \eta_1) \sin(W\theta_1) + i \sin(2\pi \eta_1) (1-\cos(W\theta))]
        \sigma^+_0 \sigma^-_1 + \text{H.c.} \bigg].
    \end{multline}
    The power pumped by the left (smaller \(j\)) edge state is
    \begin{multline}
        \frac{1}{2^L} \tr{P \tau^z_0} = -\frac{\omega}{2^L} \tr{(\omega \partial_\theta h^1_{0,1} + \dot{\eta}_1 \partial_\theta h^2_{0,1})\tau^z_0}
        = \frac{\omega W}{4}\bigg[ \frac{\omega W}{2} \sin(W\theta) (1-\cos(2 \pi \eta_1 ))  \\
        + \pi \dot{\eta}_1 (\sin(2\pi \eta_1)(1-\cos(W\theta))-\sin(\pi \eta_1))
        \bigg],
    \end{multline}
    for \(t\in [0,\tfrac{1}{3})\), and \(\tr{P \tau^z_0} = 0\) otherwise. Further, \(\tr{P \tau^z_0} = -\tr{P \tau^z_{L-1}}\), so the pumped power is nonzero only when the two edge states are opposite. We will subsequently let \(\pm\) refer only to the occupation of the LIOM at site \(j=0\), and assume the other edge LIOM is in the opposite state. Then the total pumped power is
    \begin{equation}
        \cexp{P(t,\theta)}_\pm = \pm \frac{\omega W}{2}\bigg[ \frac{\omega W}{2} \sin(W\theta) \left(1-\cos(2 \pi \eta_1(t) )\right) 
        + \pi \dot{\eta}_1(t) \left(\sin(2\pi \eta_1(t))(1-\cos(W\theta))-\sin(\pi \eta_1(t))\right)
        \bigg]
    \end{equation}
    \end{widetext}
    for \(t\in [0,\tfrac{1}{3})\). This function is plotted as a function of \(t\) in \autoref{fig:QPpump}.

    \begin{figure}
        \centering
        \includegraphics[width=\linewidth]{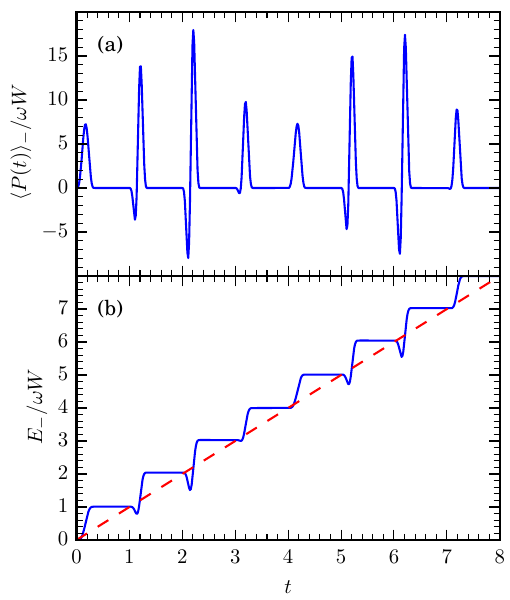}
        \caption{(a)~The power pumped into the \(\theta\) drive in the state \(\rho_{0-,(L-1)+}\)~\eqref{eqn:pump_state}. The average value is quantized to \(\omega W\), where \(W\) is the winding number of the unitary used to construct the QP pump model and \(\omega\) is the frequency of the \(\theta\) drive. The power is only nonzero when \(t \bmod 1 \in [0,\tfrac{1}{3})\). (b)~The quantized power appears as an average slope (red dashed) in the energy absorbed by the \(\theta\) drive (blue). In fact, the modified swindle model has \(E_\pm(t) = \mp \omega W t\) at integer times. \emph{Parameters:} \(\omega = (1+\sqrt{5})/2\), \(W=1\), \(\eta_1 = \sin^2(3\pi t/2)\).}
        \label{fig:QPpump}
    \end{figure}

    We can compute the average pumped power by integrating \(\cexp{P(t,\theta)}_\pm\) over \(t\) and \(\theta\). This gives
    \begin{equation}
        \bar{P}_{\pm} = \frac{1}{2\pi}\int_0^{2\pi} \d\theta \int_0^1 \d t\, \cexp{P(t,\theta)}_\pm = \mp \omega W.
    \end{equation}
    That is, \(W\) photons of frequency \(\omega\) are pumped into (or out of, if \(W< 0\)) the \(\theta\) drive per period of \(t\) (which we have set to 1) when \(\cexp{\tau^z_0} = -1\) and \(\cexp{\tau^z_{L-1}} = +1\).

\section{Discussion}
\label{sec:Discussion}

    In the context of exactly soluble commuting models and driven phases of localized systems, authors have repeatedly expressed confusion regarding the position of their models within existing classifications~\cite{Po2016QCA,Friedman2022qpedge,Zhang2021U1,Fidkowski2020beyond}. Most of these are classifications of states~\cite{Chen20111dspt,Else2016cohomology,Else2019longlived}, while localized models (including any commuting projector model) impose nontrivial structure on the entire spectrum. This structure need not be captured by any single state, making such classifications of limited applicability. This point is implicit in the literature~\cite{Po2016QCA,Fidkowski2020beyond,Long2021ALTP,Lapierre2022TLI}, but the lack of a systematic and general study seems to have allowed confusion to persist.
    
    Indeed, both the study of ground state phases and localized phases are rich mathematical problems with important implications for disparate physical systems.
    While the two problems are related (by \(e_0^*\), with short-ranged entanglement), neither can be regarded as a special case of the other.
    
    Our work demonstrates a general relationship between ALT phases and QCA. This extends both to models with on-site symmetries---forming symmetry-enriched ALT phases, or SALT phases---and to fermionic systems~\cite{Haah20233d} essentially without modification. 
    The \(\Omega\)-spectrum property we have proposed also allows for the complete classification of all homotopy groups of one-dimensional QCA and SALT phases with any on-site symmetry, reproducing and extending the classifications of Refs.~\cite{Gong2020MPU,Zhang2021U1,Zhang2023note}.
    
    We expect that the \(\Omega\)-spectrum structure can also be exploited to classify SALT phases with lattice symmetries~\cite{Thorngren2018crystal,Xiong2018minimalist,Gaiotto2019gencohomology} (Appendix~\ref{app:manifolds}). The more challenging direction of generalization is to include topological order (long-ranged entanglement), where the correspondence with QCA breaks down~\cite{Wahl2020mblto}. 
    This case is of particular interest given recent proposals of using driven topologically ordered localized models to perform quantum computation~\cite{Po2017radical,Hastings2021dynamically,Aasen2022hampath,Paetznick2023majorana,Aasen2023measurement}.

    Beyond the construction of a rigorous proof of the \(\Omega\)-spectrum conjecture, our work raises several other mathematical questions. For instance, an important extension of our work would be to include a rigorous treatment of \emph{tails} of LIOMs, as generically arise in any disordered system. This requires working with QCAs which can map local operators to quasilocal operators~\cite{Haah2023invertible,Ranard2022converseLRbound}.
    We expect that the classification is not altered in any essential way by including tails, but the technical tools used in any proof would be quite different.
    
    The models constructed through the swindle map are finite depth circuits, and are in immediate reach of current quantum simulation experiments~\cite{Eckardt2017,Gross2017,Kjaergaard2020}. 
    While the three-dimensional geometry of the phase shown in \autoref{fig:ALT-QCA} is difficult to achieve, the model of the QP pump in \autoref{subsec:QP_pump} should be straightforward to realize in most popular platforms for quantum simulation.
    Applying the swindle construction once more to the QP pump model produces a two-dimensional model with three-tone quasiperiodic driving, which has received little prior study~\cite{Long2022layer}, in part because exact numerical simulations of many-body quantum dynamics in two dimensions are very challenging. The swindle model provides a solvable point for the associated topological phase, but experimental simulation away from this fixed point may outperform achievable numerics.

\begin{acknowledgements}

    The authors are particularly grateful to Ryohei Kobayashi for collaboration during the early stages of this project. The authors also thank Maissam Barkeshli, Anushya Chandran, Lukasz Fidkowski, Jeongwan Haah, Salvatore Pace, Ryan Thorngren, Yi-Ting Tu, Brayden Ware, and Carolyn Zhang for useful discussions. This work is supported by the Laboratory for Physical Sciences, a Stanford Q-FARM Bloch Fellowship (DML), and a Packard Fellowship in Science and Engineering (DML, PI: Vedika Khemani). Research
    at Perimeter Institute is supported in part by the Government of Canada through the
    Department of Innovation, Science and Economic Development, and by the Province of
    Ontario through the Ministry of Colleges and Universities (DVE).

\end{acknowledgements}

\appendix

\section{From ALT phases to QCA}
    \label{app:ALTtoQCA}

    Our primary tool for studying localized topological phases is a mapping from quantum cellular automata (QCA) to many-body localized (MBL) Hamiltonians with short-ranged entanglement. Explicitly, we define a locality preserving unitary which maps physical on-site degrees of freedom to the conserved local integrals of motion (LIOMs).

    We formalize our definitions of ``MBL Hamiltonian'' and ``short-ranged entanglement'' thereof in Appendix~\ref{subapp:setup}. The main point of this section and Appendix~\ref{subapp:existence} is that these physical notions reproduce the formal requirements of \cite[Theorem II.4]{Haah20233d}. Thus, we realize abstract restrictions on commuting projector models as physical properties of localized systems. \cite[Theorem II.4]{Haah20233d} then states that the LIOMs of the MBL Hamiltonian can be prepared by a QCA. We make an equivalent proof in Appendix~\ref{subapp:existence}. The QCA in question is not unique, and we characterize the redundancy in Appendix~\ref{subapp:nonunique}.
    We also discuss localized phases of driven systems, which we also characterize (non-uniquely) by parametrized families of QCA.

     While our discussion is phrased for spin systems, the results we review all hold for fermions as well~\cite{Haah20233d}.

    \subsection{Setup}
        \label{subapp:setup}

        While all classification results in this work are for spin systems on the integer lattice \(\Z^d\), the correspondence between localized phases and QCA holds for either spins or fermions and any locally finite set of sites. That is, any subset of space with finite diameter contains finitely many degrees of freedom. We will call the set of points in space supporting some degree of freedom a lattice, though it does not need to have any further regularity.
        
        Formally, we consider a metric space \(X\) and a lattice \(K \subseteq X\) (which may include the same point multiple times) with associated finite dimensional local Hilbert spaces \(\mathcal{H}_k\). This specializes our discussion to spin systems, though a parallel discussion holds for systems with fermions~\cite{Haah20233d}. The total Hilbert space is the tensor product of the local Hilbert spaces,
        \begin{equation}
            \mathcal{H} = \bigotimes_{k \in K} \mathcal{H}_k.
        \end{equation}
        (We will not need to be concerned with precise definitions of the infinite tensor product.)
        
        A Hamiltonian on \(X\) is a Hermitian operator \(H\) on \(\mathcal{H}\). We will consider \emph{fully localized} Hamiltonians with strictly local LIOMs. The generic case is for the LIOMs to be quasilocal (with exponential tails), but for the purpose of mathematical proofs it is convenient to require that LIOMs have finite support on the lattice.

        \begin{definition}\label{def:MBL}
            The Hamiltonian \(H\) is range-\(r\) fully many-body localized if there is a lattice of points \(J \subseteq X\) (points may occur more than once) and Hermitian operators \(\{\tau^z_j\}_{j \in J}\), called LIOMs, of the form
            \begin{equation}
                \tau^z_j \in \mathbbm{1} \otimes \bigotimes_{k \in B_r(j) \cap K} \mathrm{End}(\mathcal{H}_k),
            \end{equation}
            [where \(B_r(j)\) is the ball of radius \(r\) around \(j\)] such that \([\tau^z_j, H] = [\tau^z_j, \tau^z_{j'}] = 0\). Further, the set of LIOMs must be complete, such that a local operator \(A\) commutes with all LIOMs only if \(A \in \cexp{\tau^z_j}_{j\in J}\).
        \end{definition}

        This definition should be compared to \cite[Definition II.1]{Haah20233d}.

        We further say that \(H\) is short-range entangled if all its excitations can be created and destroyed locally. This excludes models with nontrivial anyons, where there are local excitations which cannot be removed except by annihilating them with other, potentially distant, excitations. It also excludes symmetry breaking, where domain walls also cannot be destroyed locally. We formalize this notion with the following convenient definition.

        \begin{definition}\label{def:no_topo}
            The Hamiltonian \(H\) is range-\(r\) fully many-body localized with short-ranged entanglement if it is range-\(r\) fully many-body localized and additionally there are unitary operators
            \begin{equation}
                \tau^x_j \in \mathbbm{1} \otimes \bigotimes_{k \in B_r(j) \cap K} \mathrm{End}(\mathcal{H}_k),
            \end{equation}
            such that \([\tau^x_j, \tau^z_{j'}] = [\tau^x_j, \tau^x_{j'}] = 0\) for \(j \neq j'\) and \(\tau^x_j\) cyclically permutes the eigenspaces of \(\tau^z_j\). That is, denoting the ordered set of distinct eigenvalues of \(\tau^z_j\) as \((\lambda^j_0, \ldots \lambda^j_{N_j-1})\) and the associated eigenspaces \(\mathrm{eig}(\lambda^j_l,\tau^z_j)\), we have
            \begin{equation}
                \tau^x_j \mathrm{eig}(\lambda^j_l,  \tau^z_j) = \mathrm{eig}(\lambda^{j}_{l+1 \bmod N_j}, \tau^z_j).
            \end{equation}
        \end{definition}

        This definition should be compared to \cite[Definition II.3]{Haah20233d}. One can also compare this notion of ``short-ranged entanglement'' for Hamiltonians to the usual notion of short-ranged entanglement for states. A state \(\ket{\psi}\) may be said to be short-range entangled if, for any other state \(\ket{\psi'}\) which differs from \(\ket{\psi}\) only in some ball \(B_r(x)\) [\(\qexp{A}{\psi} = \qexp{A}{\psi'}\) for any local \(A\) not supported in \(B_r(x)\)] there is a local unitary \(\tau^x_{\psi \psi'}\) supported in \(B_r(x)\) such that \(\tau^x_{\psi \psi'} \ket{\psi'} = \ket{\psi}\). Any eigenstate of an \(H\) satisfying Definition~\ref{def:no_topo} is short-range entangled in this sense (the required \(\tau^x_{\psi \psi'}\) can be constructed from the \(\tau^x_j\) operators), but the converse---that \(\tau^x_j\) operators exist for any MBL Hamiltonian \(H\) whose eigenstates are all short-range entangled---is not obvious.

        Note that a commuting model is characterized less by the Hamiltonian \(H\), and more by the algebra of the operators \(\{\tau^z_j, \tau^x_j\}_{j \in J}\). Indeed, focusing on LIOM operator algebras will be much more useful than the Hamiltonian itself in this context. In this way, we also avoid needing to be specific regarding the resolution of degeneracies in the Hamiltonian. A classification more directly based on the topology of the space of MBL Hamiltonians would need a definition of MBL which excludes highly degenerate Hamiltonians, such as the zero Hamiltonian, from being MBL. Otherwise all MBL Hamiltonians \(H_0\) and \(H_1\) would be connected by a continuous path with \(H_t = (1-2t)H_0\) for \(t \in [0,\tfrac{1}{2})\) and \(H_t = (2t-1)H_1\) for \(t \in [\tfrac{1}{2},1]\). Focusing on LIOMs captures the continuous change of local properties of the model while allowing for accidental degeneracies, so we will not address how to formulate a definition of MBL which avoids this issue.
        
        For convenience, we define some abbreviated notation for operator algebras. For \(S \subseteq X\) define
        \begin{equation}
            \mathcal{O}_S = \mathbbm{1}\otimes \bigotimes_{k \in S\cap K} \mathrm{End}(\mathcal{H}_k) \cong \bigotimes_{k \in S \cap K} \mathrm{Mat}_{M_k}(\C),
        \end{equation}
        where \(M_k = \mathrm{dim}(\mathcal{H}_k)\) and \(\mathrm{Mat}_{M_k}(\C)\) is the algebra of \(M_k \times M_k\) matrices over the complex numbers. Similarly define the algebra generated by the \(\tau\) operators within some region as
        \begin{equation}
            \mathcal{T}_S = \cexp{ \tau^z_j, \tau^x_j}_{j \in S \cap J}.
        \end{equation}

        \begin{definition}\label{def:QCA}
            A unitary operator \(V\) on \(\mathcal{H}\) is a range-\(r\) quantum cellular automaton (QCA) if
            \begin{equation}
                V \mathcal{O}_S V^\dagger \subseteq \mathcal{O}_{B_r(S)}
            \end{equation}
            for all \(S \subseteq X\).
        \end{definition}

        Here, \(B_r(S)\) is the ball of radius \(r\) around the set \(S\). Properly, a QCA should be a projective unitary, but the distinction will not be important here.
        Reference~\cite{Hastings2013torus} defines a commuting model to be ``without intrinsic topological order'' when the LIOMs \(\tau^z_j\) can be prepared from single-site operators
        \begin{equation}
            \sigma^z_k \in \mathbbm{1} \otimes \mathrm{End}(\mathcal{H}_k)
        \end{equation}
        by a QCA \(V\). Any such Hamiltonian also clearly has short-ranged entanglement by our Definition~\ref{def:no_topo}, as we can construct \(\tau_k^x = V \sigma^x_k V^\dagger\) from a matrix \(\sigma^x_k \in \mathrm{End}(\mathcal{H}_k)\) which acts as a cyclic permutation on the eigenvectors of \(\sigma^z_k\). The result of Appendix~\ref{subapp:existence} is that the QCA \(V\) always exists given Definition~\ref{def:no_topo} as well, and as such the two notions short-ranged entanglement and lack of intrinsic topological order are equivalent.

    \subsection{Existence}
        \label{subapp:existence}

        In this appendix, we prove the following theorem.

        \begin{theorem}\label{thm:QCA_exists}
            If \(H\) is range-\(r\) many-body localized with short-ranged entanglement, then there is a range-\(2r\) QCA \(V\) such that each \(V^\dagger \tau^z_j V\) is a sum of products of commuting single-site operators \(\sigma^z_k\).
        \end{theorem}

        In fact, after a redefinition of LIOMs to another equivalent set, we will have that \(V^\dagger \tau^z_j V\) is a single-site operator.

        The proof is essentially the same as \cite[Theorem II.4]{Haah20233d}. The main insight here is that the abstract requirements of that theorem correspond to natural physical requirements for an MBL Hamiltonian.

        The strategy will be to show that each of the algebras
        \begin{equation}
            \mathcal{T}_j = \cexp{\tau^z_j, \tau^x_j}
        \end{equation}
        is isomorphic to a matrix algebra. This implies that there is an isomorphism between \(\mathcal{T}_j\) and \(\mathcal{O}_k\) whenever they have the same dimension. [All isomorphisms will be \(*\)-isomorphisms---that is an isomorphism \(\alpha\) such that \(\alpha(A^\dagger) = \alpha(A)^\dagger\) for all operators \(A\). All such automorphisms of this kind between matrix algebras can be expressed as a conjugation by a unitary.] To find a QCA \(V\) which implements this isomorphism simultaneously for all sites, we need to match the algebras \(\mathcal{O}_{k}\) to \(\mathcal{T}_{j(k)}\) locally and bijectively. This can be phrased as a graph problem, and Hall's marriage theorem~\cite{Hall1935marriage,Hall1998combinatorial} implies that this problem always has a solution.

        \begin{lemma}\label{lem:matrix_iso}
            There is an isomorphism
            \begin{equation}
                \mathcal{T}_j \cong \mathrm{Mat}_{N_j}(\C)
            \end{equation}
            for each \(j \in J\).
        \end{lemma}

        \begin{proof}
            Wedderburn's classification of finite dimensional semisimple algebras implies that \(\mathcal{T}_j\) is isomorphic to a full matrix algebra if and only if it is central. That is, \(A \in \mathcal{T}_j\) commutes with all of \(\mathcal{T}_j\) only if \(A \in \C \mathbbm{1}.\)

            Suppose \(A \in Z(\mathcal{T}_j)\) is in the center of \(\mathcal{T}_j\). Then it must commute with all the LIOMs---it commutes with \(\tau^z_j\) by assumption and it commutes with the other \(\tau^z_l\) because it is contained in \(\mathcal{T}_j\). Then, by Definition~\ref{def:MBL}, \(A\) is a sum of products of LIOMs. However, as \(A \in \mathcal{T}_j\), it must be a polynomial in \(\tau^z_j\). In particular,
            \begin{equation}
                A = \bigoplus_{l=1}^{N_j-1} a_l \mathbbm{1}_{\mathrm{eig}(\lambda^j_l,  \tau^z_j)}
            \end{equation}
            for some \(a_l \in \C\).
            Evaluating the commutator with \(\tau^x_j\), we have
            \begin{equation}
                \tau^{x\dagger}_j A \tau^x_j = \bigoplus_{l=1}^{N_j-1} a_{l+1\bmod N_j}\mathbbm{1}_{\mathrm{eig}(\lambda^j_{l},  \tau^z_j)} 
                = A,
            \end{equation}
            so all the \(a_l = a\) are equal and \(A \in \C\mathbbm{1}\).
        \end{proof}

        It is convenient to require that all the \(N_j\) and \(M_k = \mathrm{dim}(\mathcal{H}_k)\) are prime. This can always be achieved by redefining the lattices \(J\) and \(K\). Indeed, focusing on \(N_j\), we can make a prime factorization
        \begin{equation}
            N_j = p_0 p_1 \cdots p_{m_j-1}
        \end{equation}
        (not all necessarily distinct) and express
        \begin{equation}
            \mathcal{T}_j \cong \mathrm{Mat}_{N_j}(\C) \cong \bigotimes_{n=0}^{m_j-1} \mathrm{Mat}_{p_n}(\C).
        \end{equation}
        Then we have that \(\mathcal{T}_j = \prod_{n=0}^{m-1} \mathcal{T}_{j_n}\) with each tensor factor being of prime dimension. Further, we can make this factorization such that the eigenbasis of \(\tau^z_j\) becomes a product basis. Indeed, we have that any \(l \in \{0,\ldots, N_j-1\}\) can be uniquely expressed as a sum of integers
        \begin{equation}
            l = \sum_{n=0}^{m_j-1} \mu^l_n \prod_{q=0}^{n-1} p_q
        \end{equation}
        with \(\mu^l_n \in \{0,\ldots, p_n-1\}\). (Empty sums are taken to be 0 and empty products are 1.) Then we can define Hermitian projectors \(\pi^j_{\mu n}\) (with \(\mu \in \{0,\ldots, p_n-1\}\) and \(n \in \{0,\ldots, m_j-1\}\)) onto the direct sum of eigenspaces \(\mathrm{eig}(\lambda^j_l, \tau^z_j)\) such that \(\mu^l_n = \mu\). Note that all \(\pi^j_{\mu n}\) with the same \(n\) have the same rank, specifically \(N_j/p_n\). As \(\mathcal{T}_j\) is a full matrix algebra, we can find \(\tau^x_{j_n}\) such that 
        \begin{equation}
            \tau^x_{j_n} \pi^j_{\mu n} \tau^{x\dagger}_{j_n} = \pi^j_{(\mu+1\bmod p_n) n}
        \end{equation}
        and such that \(\tau^x_{j_n}\) commutes with all the other projectors. Indeed, \(\tau^x_{j_n}\) is defined by cyclically permuting the value of \(\mu^l_n\) in the decomposition of \(l\) for \(\mathrm{eig}(\lambda^j_l, \tau^z_j)\). By essentially the same calculation as Lemma~\ref{lem:matrix_iso}, \(\cexp{\tau^x_{j_n}, \pi^j_{\mu n}}_{\mu=0}^{p_n-1}\) is isomorphic to a full matrix algebra, and so occurs as a tensor factor in \(\mathcal{T}_j\). Then we make the decomposition of \(\mathcal{T}_j\) as
        \begin{equation}
            \mathcal{T}_j \cong \bigotimes_{n=0}^{m_j-1} \cexp{\tau^x_{j_n}, \pi^j_{\mu n}}_{\mu=0}^{p_n-1}.
        \end{equation}
        We pick new LIOMs \(\tau^z_{j_n} = \sum_{\mu=0}^{p_n-1} \lambda^{j_n}_\mu \pi^j_{\mu n}\) with \(\lambda^{j_n}_\mu \in \R\) all distinct. These commute with \(\tau^z_j\), and so are each actually functions of \(\tau^z_j\) by Definition~\ref{def:MBL}. On the other hand, they generate all diagonal operators in \(\mathcal{T}_j\) in the \(\tau^z_j\) eigenbasis, so \(\tau^z_j\) can also be recovered from these new LIOMs. Thus, the new LIOMs are equivalent to the old LIOM.
        
        Replace the lattice \(J\) with a new lattice \(J'\) which consists of the old lattice sites \(j\) each repeated \(m_j\) times, and associated LIOMs \(\tau^z_{j_n}\) and algebras \(\mathcal{T}_{j_n}\). This lattice has the properties we need. We henceforth drop primes, and just use \(J\). Similarly, we can ensure that each of the local Hilbert spaces \(\mathcal{H}_k\) is of prime dimension.

        As \(\mathcal{T}_j\) and \(\mathcal{O}_k\) are both full matrix algebras, there is an isomorphism between them whenever they have the same dimension. To build these individual isomorphisms into a QCA, we need to make a matching between \(\mathcal{T}_j\) and \(\mathcal{O}_k\) algebras which are supported near one another. Define the bipartite graph \(B=((J,K),E)\) with edges
        \begin{equation}
            (j,k) \in E \iff N_j = M_k \text{ and } \mathrm{dist}(j,k) \leq r.
        \end{equation}
        Then we are looking for a perfect matching in this graph. That is, a bijection \(k \mapsto j(k)\) such that \((j(k), k) \in E\) for all \(k \in K\). Hall's marriage theorem~\cite{Hall1935marriage} (specifically, the infinite version due to Hall~\cite{Hall1998combinatorial}---no relation) gives necessary and sufficient conditions for a perfect matching to exist. Define the neighborhood of \(S \subseteq B\) as
        \begin{equation}
            \partial S = \{b \in B\,:\, (s,b) \in E\text{ for some }s \in S\}.
        \end{equation}
        Then the marriage theorem states that if: (1)~every vertex has finite degree, (2)~for every finite \(S \subseteq J\) we have
        \begin{equation}
            |S| \leq |\partial S|, 
            \label{eqn:marriage}
        \end{equation}
        and similarly (3)~for each finite \(S \subseteq K\) we have \(|S| \leq |\partial S|\), 
        then the required perfect matching bijection exists. The criterion~\eqref{eqn:marriage} is called the \emph{marriage condition}.

        \begin{lemma}\label{lem:matching}
            The graph \(B\) has a perfect matching.
        \end{lemma}

        \begin{proof}
            We verify the marriage condition, Eq.~\eqref{eqn:marriage}.
            
            Note that every vertex in \(B\) has finite degree by the assumption that \(K\) and \(J\) are locally finite. Further observe that the graph \(B\) splits into distinct connected components labeled by primes \(p\) such that \(N_j = M_k = p\). Thus, we can consider each component separately. Henceforth, we assume that we are working in a given component.
            
            Let \(S \subseteq J\) be a finite subset of LIOM sites. Then we have that \(\mathcal{T}_S\) is a subalgebra in \(\mathcal{O}_{B_r(S)}\). As both \(\mathcal{T}_S\) and \(\mathcal{O}_{B_r(S)}\) are full matrix algebras, \(\mathcal{O}_{B_r(S)}\) must be a tensor product of \(\mathcal{T}_S\) and its commutant in \(\mathcal{O}_{B_r(S)}\). Comparing dimensions then shows that \(\prod_{j \in S} N_j\) divides \(\prod_{k \in B_r(S)} M_k\). Then, all of the primes \(N_j = p\) must occur among the primes \(M_k\), repeated according to their multiplicity. The sites in \(B_r(S)\) with local Hilbert space dimension \(M_k = p\) are precisely the elements of \(\partial S\) in this component. Thus, we conclude \(|S| \leq |\partial S|\).

            In the other direction, let \(S \subseteq K\) be a finite subset of lattice sites. Then we also have that \(\mathcal{O}_S\) is a subalgebra of \(\mathcal{T}_{B_r(S)}\), whereupon the proof proceeds as in the previous paragraph. Indeed, suppose that \(A \in \mathcal{O}_S \subseteq \mathcal{O}_{B_{2r}(S)}\). As \(\mathcal{T}_{B_r(S)}\subseteq \mathcal{O}_{B_{2r}(S)}\) occurs as a tensor factor in \(\mathcal{O}_{B_{2r}(S)}\), \(A\) is a sum of products of operators in \(\mathcal{T}_{B_r(S)}\) and its commutant. We need to show that if \(A\) is in the commutant, then it is a scalar. Indeed, as their supports do not overlap, \(A\) must also be in the commutant of \(\mathcal{T}_{K\setminus B_r(S)}\). Thus, \(A\) commutes with all of the LIOMs, and by Definition~\ref{def:MBL} we must have \(A \in \cexp{\tau^z_j}_{j\in J}\). The only such operator which also commutes with all the \(\tau^x_j\) operators is a scalar, \(A = a \mathbbm{1}\).
        \end{proof}

        We can now prove Theorem~\ref{thm:QCA_exists}.

        \begin{proof}[Proof of Theorem~\ref{thm:QCA_exists}]
            Let \(j:K \to J\) be a perfect matching as in Lemma~\ref{lem:matching}. By Lemma~\ref{lem:matrix_iso} there is an isomorphism of algebras
            \begin{equation}
                \alpha_k: \mathcal{O}_k \to \mathcal{T}_{j(k)}.
            \end{equation}
            Then
            \begin{equation}
                \alpha = \bigotimes_{k\in K} \alpha_k
            \end{equation}
            is an algebra isomorphism between \(\bigotimes_{k\in K} \mathcal{O}_k\) and \(\bigotimes_{j \in J}\mathcal{T}_{j} \). It is bijective because the matching \(k(j)\) is perfect, and it is a homomorphism because all the algebras \(\mathcal{T}_{j} \) commute with one another.

            Further, \(\alpha\) is \(2r\)-locality preserving: we have \(\mathrm{dist}(j(k),k) \leq r\) by the definition of the graph \(B\), and then \(\alpha(\mathcal{O}_k) = \mathcal{T}_{j(k)} \subseteq \mathcal{O}_{B_{2r}(k)}\) by the fact that \(\mathcal{T}_{j(k)}\) is supported in \(\mathcal{O}_{B_{r}(j(k))}\) and the triangle inequality.

            All such complex matrix algebra homomorphisms are of the form \(\alpha(A) = V A V^\dagger\) for some unitary \(V\), which is the QCA we need. It maps all LIOMs to single-site operators, thus meeting the requirements of the theorem.
        \end{proof}

        Define the single-site operators \(\sigma^z_{k} = V^\dagger \tau^z_{j(k)} V\). Then we have that
        \begin{equation}
            V^\dagger H V = \sum_k h_k \sigma^z_{k} + \sum_{kl} h_{kl} \sigma^z_{k} \sigma^z_{l} + \cdots
        \end{equation}
        is a sum of products of the commuting single-site operators \(\sigma^z_k\) (by Definition~\ref{def:MBL}).
        We note that this is just a diagonalization of \(H\) in the simultaneous eigenbasis of all the \(\sigma^z_j\). The nontrivial conclusion of this section is that this diagonalization can be made locally.

        This characterization of localized system with short-ranged entanglement by QCA motivates a similar definition for driven systems. We consider an \(n\)-tone-driven system to be localized with short-ranged entanglement when it has a generalized Floquet decomposition~\cite{Else2019longlived,Long2022QPMBL}
        \begin{equation}
            U(t) = V(\Bt_t) e^{-i t H_F} V(\Bt_0)^\dagger,
            \label{eqn:FB_app}
        \end{equation}
        where the static Floquet Hamiltonian
        \begin{equation}
            H_F = \sum_k h_k \sigma^z_{k} + \sum_{kl} h_{kl} \sigma^z_{k} \sigma^z_{l} + \cdots
        \end{equation}
        is a sum of products of the fixed \(\sigma^z_k\) operators and the \emph{micromotion operator} \(V(\Bt)\) is a continuous parametrized family of QCA. By a similar proof to the above, we see that this definition is equivalent to the existence of a maximal set of LIOMs \(\tau^z_j(\Bt)\) with explicit continuous \(\Bt\) dependence and a family of conjugate \(\tau^x_j(\Bt)\) operators also with continuous \(\Bt\) dependence.
        Alternatively, the definition is also equivalent to the \emph{quasienergy operator}~\cite{Sambe1973,Ho1983,Blekher1992,Verdeny2016,Martin2017} having a complete set of LIOMs~\cite{Long2022QPMBL} and conjugate operators which create or destroy excitations locally.
        
        The requirement of the family \(\tau^x_j(\Bt)\) being continuous is nontrivial. For instance, even for 2-tone-driven qubits we may construct continuous families of LIOMs such that each family of LIOM eigenstates has a nonzero Chern number~\cite{Martin2017,Crowley2019qptop}. However, there is an obstruction to continuously defining a family of conjugate operators, as the eigenstates cannot be given a smooth gauge. We do not consider such systems to be localized as they have continuous spectrum~\cite{Crowley2019qptop}, and so do not enjoy the same stability to generic perturbations we expect of localized systems~\cite{Long2022QPMBL}. Indeed, in a \emph{frequency lattice} picture, such systems are delocalized~\cite{Crowley2019qptop}.

        We emphasize that Floquet localized phases with short-ranged entanglement correspond to loops of QCA---\emph{not} loops in the space of QCA modulo generalized permutations, even if QCA which differ by generalized permutations correspond to the same LIOMs. Indeed, a loop of QCA modulo permutations could correspond to a path from the identity QCA to a permutation. This path of QCA does not define a legitimate micromotion operator \(V(\Bt)\).

    \subsection{Non-uniqueness}
        \label{subapp:nonunique}

        A QCA which maps \(H\) to a sum of product of \(\sigma^z_k\) operators always exists, but this QCA is not unique.
        Here, we characterize this redundancy, and show that it does not affect the classification of driven short-range entangled MBL phases unless \((d,n) = (0,1)\).

        Indeed, diagonalization of the Hamiltonian \(H = V H_{\mathrm{triv}} V^\dagger\) (where \(H_{\mathrm{triv}}\) is a sum of products of \(\sigma^z_k\) operators) only determines \(V\) up to right-multiplication by a \emph{generalized permutation} QCA.
        We define a QCA \(P\) to be a generalized permutation exactly when it maps \(\sigma^z_k\) operators to sums of products of other  \(\sigma^z_k\)s.
        This name is motivated as follows.
        In a finite dimensional Hilbert space, \(V\) is a unitary operator which diagonalizes \(H\) in the simultaneous eigenbasis of \(\{\sigma^z_{k}\}_{k\in K}\). Even if \(H\) is nondegenerate, such unitaries are only unique up to reordering of eigenvalues and multiplying eigenstates by arbitrary phases. Specifically, if \(P\) is a unitary generalized permutation matrix (a unitary matrix with exactly one unit modulus complex number in each row and column) in the eigenbasis of \(\{\sigma^z_{k}\}_{k\in K}\), then \((VP)^\dagger H (VP)\) is still diagonal. We extend this ``generalized permutation'' language to the infinite dimensional case via the property of preserving the algebra of \(\sigma^z_k\) operators.

        Translations form an important class of QCA permutation operators~\cite{Gross2012QCA,Freedman2020higherd}. In one dimension, all QCA are given by translation operators, up to a finite depth circuit. Thus, while one-dimensional QCA are classified by \(\log \Q^\times\)~\cite{Gross2012QCA}, all these nontrivial QCA land within the redundancy of the description of localized phases. As such, all one-dimensional localized phases of spins without symmetry are trivial. The same is not true of fermions, where the Kitaev chain~\cite{Kitaev2001majorana} represents a nontrivial localized phase which can be prepared by a QCA~\cite{Po2017radical}. It has invertible eigenstate order, and so is not an ALT phase.

        Mathematically, the classification of localized phases with short-ranged entanglement is given by a quotient of \(\pi_0(\QCA_d)\). Denoting the set of generalized permutation QCA by \(\mathrm{GP}_d\), there is a canonical inclusion map \(\iota: \mathrm{GP}_d \to \QCA_d\), and an induced homomorphism \(\iota^*: \pi_0(\mathrm{GP}_d) \to \pi_0(\QCA_d)\). We write the image of this homomorphism as \(I_{d,0} = \iota^*(\pi_0(\mathrm{GP}_d))\). Then the quotient which classifies localized topological phases is 
        \begin{equation}
            \pi_0(\QCA_d)/I_{d,0}.
        \end{equation}
        This group is isomorphic to \(\pi_0(\QCA_d/\mathrm{GP}_d)\). Indeed, there is a fibration sequence
        \begin{equation}
            \mathrm{GP}_d \xrightarrow{\iota} \QCA_d \xrightarrow{q} \QCA_d/\mathrm{GP}_d
        \end{equation}
        (where \(q\) is the quotient map which sends a QCA \(V\) to its coset \(V \mathrm{GP}_d\)) with associated long exact sequence of homotopy groups
        \begin{equation}
            \cdots \to \pi_0(\mathrm{GP}_d) \xrightarrow{\iota^*} \pi_0(\QCA_d) \xrightarrow{q^*} \pi_0(\QCA_d/\mathrm{GP}_d) \to 0,
        \end{equation}
        and \(\pi_0(\QCA_d/\mathrm{GP}_d) \cong \pi_0(\QCA_d)/I_{d,0}\) follows from the first isomorphism theorem applied to \(q^*\).
    
        In the driven case, the redundancy is captured by a family of generalized permutation matrices \(P(\Bt)\) with elements
        \begin{equation}
            P_{\alpha \beta}(\Bt) = \delta_{\pi(\alpha) \beta} e^{-i \Bm_\beta \cdot \Bt - i \phi_\beta},
        \end{equation}
        where \(\pi\) is a permutation of the simultaneous eigenstates of the \(\{\sigma^z_k\}_{k\in K}\) operators and \(\Bm_\alpha \in \Z^n\) are vectors of integers. [Such \(P(\Bt)\) are generalized permutation matrices in a frequency lattice picture~\cite{Long2021ALTP,Long2022QPMBL}.] Any \(P\) of this form defines an identical unitary evolution~\eqref{eqn:FB_app} with
        \begin{align}
            V'(\Bt) &= V(\Bt) P(\Bt), \\
            H_F' &= H_F - i [\partial_t P(\Bt_t)] P(\Bt_0)^\dagger.
        \end{align}
        Note that \(i [\partial_t P(\Bt_t)] P(\Bt_t)^\dagger\) is a static diagonal operator with elements
        \begin{equation}
            (i [\partial_t P(\Bt_t)] P(\Bt_t)^\dagger)_{\alpha \beta} =\Bm_{\pi^{-1}(\alpha)} \cdot \Bw \, \delta_{\alpha \beta},
        \end{equation}
        where \(\Bw = (\omega_1, \ldots, \omega_n)\) is the vector of drive frequencies.
        If \(P(\Bt)\) is also a QCA for all \(\Bt\), then \(V'(\Bt)\) also defines an acceptable generalized Floquet decomposition.

        The only example for which this redundancy affects the classification of driven localized phases for \(d\leq 3\) is the classification of Floquet (\(n=1\)) phases in zero dimensions. 
        From \autoref{sec:Classification}, we have that \(\pi_1(\QCA_0) \cong \Q/\Z\), and that a complete index is given by the winding number \(\widetilde{W}_1[V]\). In zero dimensions, all unitaries are QCA, so the diagonal unitary with elements
        \begin{equation}
            P_m(\theta)_{\alpha \beta} = e^{i m \theta \delta_{\alpha 0}} \delta_{\alpha \beta}
        \end{equation}
        defines an acceptable \(V'(\theta) = V(\theta) P_m(\theta)\). However, we can compute that
        \begin{equation}
            \widetilde{W}_1[V P_m] = \widetilde{W}_1[V] + \frac{m}{N} \mod 1,
        \end{equation}
        where \(N\) is the Hilbert space dimension. Thus, multiplication by \(P_m(\theta)\) can take \(\widetilde{W}_1[V']\) to any allowed value, and we conclude that all loops of zero-dimensional QCA (loops of projective unitaries) define trivial localized phases.

        To see that the redundancy by generalized permutations does not affect the driven classification for other value of \((d,n)\), we must be more systematic. We can classify driven localized phases by a quotient of the group \(\ho{\T^n}{\QCA_d}\), where \(\T^n = [\R/(2\pi\Z)]^n\) is the \(n\)-torus. The generalized permutations \(P(\Bt)\) are particular maps from \(\T^n\) to \(\QCA_d\), and we denote the subgroup of \(\ho{\T^n}{\QCA_d}\) that includes such maps by \(I_{d,n}\). Then the classification of \(n\)-tone driven localized phases in \(d\)-dimensions with short-ranged entanglement is given by
        \begin{equation}
            \ho{\T^n}{\QCA_d}/I_{d,n}.
        \end{equation}
        Note that this group is \emph{not} isomorphic to \(\ho{\T^n}{\QCA_d/\mathrm{GP}_d}\) in general. Taking, for instance, \(n=1\), there are loops in \(\ho{\T^1}{\QCA_d/\mathrm{GP}_d} \cong \pi_1(\QCA_d/\mathrm{GP}_d)\) which correspond to paths from the identity to a locally generated permutation. This is not a legitimate micromotion operator \(V(\theta)\), as it is not periodic as a (projective) unitary.

        Assuming that a generalized permutation \(P(\Bt)\) acts as a family of finite depth circuits, we show that the homotopy class of \(P(\Bt)\) is trivial unless \((d,n) = (0,1)\).
        This implies that \(I_{d,n} \leq \pi_0(\QCA_d)\) for \((d,n) \neq (0,1)\), and we suspect that \(I_{d,n} \cong 0\) unless \(d+n =1\), so that the only redundancy in the description of MBL phases by QCA phases is for \((d,n) \in \{(0,1),(1,0)\}\).
        
        As we assume \(P(\Bt)\) acts by finite depth circuits, we only need to consider \(n \geq 1\)
        In fact, Ref.~\cite{Long2021ALTP} already showed that \(I_{0,n} \cong 0\) for \(n \neq 1\), so we can also restrict to \(d \geq 1\).
        In this case, we argue that \(P(\Bt)\) can always be constructed as a depth-two circuit of \((d-1)\)-dimensional generalized permutations. 
        Such a circuit does not pump any nontrivial QCA to its boundary on open boundaries, so if the bulk-boundary correspondence discussed in \autoref{sec:Classification} holds, then all such circuits are trivial. More constructively, the deformation discussed in Appendix~\ref{subapp:deformation} can be applied to such a circuit to explicitly give a homotopy to the identity.

        To arrive at this depth-two circuit, we observe that, by definition, any generalized permutation \(P(\Bt)\) is generated by time evolution under a time-independent Hamiltonian composed of \(\sigma^z\) LIOMs. That is, there are trivial LIOM Hamiltonians \(H_j\) such that
        \begin{equation}
            [i \partial_{\theta_j} P(\Bt)] P(\Bt)^\dagger = H_j.
        \end{equation}
        All the \(H_j\) operators commute, so we can solve this equation as \(P(\Bt) = e^{-i \Bt \cdot \vec{H}} P(0)\), where \(\vec{H} = (H_1, ..., H_n)\).
        Split \(P(0)\) into a brickwork circuit of \((d-1)\)-dimensional blocks
        \begin{equation}
            P(0) = \left( \prod_{i \in 2\Z +1} b_i^\dagger \right) \left( \prod_{i \in 2\Z} b_i \right),
        \end{equation}
        where each \(b_i\) only overlaps in support with \(b_{i\pm 1}\) near their boundaries, and such that the width of these blocks are much larger than the support of any term in \(H_j\). 
        We want to restrict the Hamiltonians \(H_j\) to \(H^{(2i)}_j\) with support contained within that of \(b_{2i}\), such that \(H^{(2i)}_j\) agrees with \(H_j\) near the center of the block, and such that \(e^{-i \Bt \cdot \vec{H}^{(2i)}} \) is still periodic.
        To construct \(H^{(2i)}_j\), first define \(\tilde{H}^{(2i)}_j\) to include every term in \(H_j\) with support overlapping the \(b_{2i}\) block. Then pick an eigenvalue \(\lambda_k\) of each \(\sigma^z_k\) with \(k\) not inside the \(b_{2i}\) block, and make a replacement
        \begin{equation}
            \sigma^z_{k_1} \cdots \sigma^z_{k_l} \otimes A \mapsto \lambda_{k_1} \cdots \lambda_{k_l} \mathbbm{1} \otimes A
        \end{equation}
        for every term in \(\tilde{H}^{(2i)}_j\) which involves these \(\sigma^z_k\)s. The result of this replacement is \(H_j^{(2i)}\). It is certainly supported within the relevant block, as every term in \(H_j\) (and hence \(\tilde{H}^{(2i)}_j\)) is a sum of products of \(\sigma^z_k\) LIOMs, and every factor not supported in the block gets replaced by a scalar. To see that \(e^{-i 2\pi H_j^{(2i)}} = \mathbbm{1}\), observe that \(H_j^{(2i)}\) acts on the \(b_{2i}\) block as \(H_j\) would if we projected every site \(k\) outside the block into the \(\lambda_k\) eigenstate of \(\sigma^z_k\) (up to a constant). But \(e^{-i 2\pi H_j} = \mathbbm{1}\), so projecting into any subspace still results in the identity action, and thus \(e^{-i 2\pi H_j^{(2i)}} = \mathbbm{1}\).\footnote{This argument can be expressed more succinctly and precisely using the notion of projected algebras in Ref.~\cite{Tu2025anomalies}.}
        Then the even-indexed lower-dimensional blocks for the generalized permutation \(P(\Bt)\) will be \(e^{-i \Bt \cdot \vec{H}^{(2i)}} b_{2i}\). We observe that the product
        \begin{equation}
            P(\Bt) \left( \prod_{i \in 2 \Z} b_{i}^\dagger e^{i \Bt \cdot \vec{H}^{(i)}} \right) = \prod_{i \in 2\Z + 1} b_i^\dagger(\Bt)
        \end{equation}
        must split into a tensor product of disjoint odd-indexed blocks, as \(P(\Bt)\) and \(e^{-i \Bt \cdot \vec{H}^{(2i)}} b_{2i}\) agree deep inside the even-indexed blocks. As products of generalized permutaions, these blocks are also generalized permutations.
        Thus, we have the desired brickwork circuit for \(P(\Bt)\), and conclude that \(P(\Bt)\) is homotopic to the identity as a map from \(\T^{n\geq1}\) to \(\QCA_{d \geq 1}\).

\section{The \texorpdfstring{\(\Omega\)}{Omega}-spectrum property}
    \label{app:Omega_spectrum}

    We conjecture that the spaces of quantum cellular automata on the integer lattice \(\Z^d\) stabilized by adding ancillae, \(\QCA_d\), form an \(\Omega\)-spectrum,
    \begin{equation}
        \QCA_{d-1} \simeq \Omega \QCA_{d}.
    \end{equation}
    For \(d < 0\), we define \(\QCA_d = \Omega^{-d} \QCA_0\), making the negative degrees of this statement trivial. In positive degree, we propose that the swindle map \(S\) (\autoref{fig:swindle}) provides the homotopy equivalence. 

    In this appendix, we provide further arguments in support of this conjecture. In Appendix~\ref{subapp:blending}, we recall that if loops of QCA can be put on open boundary conditions such that they act by QCA on their edges---known as \emph{blending with the identity}---then loops of QCA can all be parametrized as finite depth circuits of lower-dimensional QCA. We also argue that this condition holds for loops of one-dimensional QCA. Then, in Appendix~\ref{subapp:deformation}, we construct an explicit deformation retract from the space of circuits to the image of the swindle map. As the swindle map is injective, this shows
    \begin{equation}
        \QCA_{d-1} \cong \mathrm{im}(S) \simeq \Omega\QCA_{d}
        \label{eqn:deformation_result}
    \end{equation}
    under the assumptions of Appendix~\ref{subapp:blending}, where ``\(\cong\)'' is homeomorphism.

    Making these arguments rigorous requires being more specific about the topology on \(\QCA_d\), even in \(d=1\). Here, we implicitly use the \emph{strong topology} on bounded-range QCA---a sequence of QCA \(V_n\) converges if and only if it is of bounded range and \(V_n A V_n^\dagger\) converges for all local operators \(A\).\footnote{Being more technical, this means we are using the final topology on the direct limit \(\varinjlim \QCA_d(r)\), where \(\QCA_d(r)\) is the space of \(d\)-dimensional QCA with range bounded by \(r\), equipped with the strong topology, and the direct limit is taken with respect to the inclusion \(\QCA_d(r) \hookrightarrow \QCA_d(r')\) for \(r \leq r'\).}

    While it may be useful to keep spin systems in mind, the arguments we present in this appendix are agnostic to local operators being fermionic or spin operators. Provided that the additional ancillae used to stabilize a system may support fermionic operators, everything in the following sections also applies to systems with fermions.

    \subsection{Blending and circuits}
        \label{subapp:blending}

        A \emph{blend} between \(\alpha\) and \(\beta \in \QCA_d\) is a third QCA \(\gamma\) such that
        \begin{align}
            \gamma(A) &= \alpha(A) \quad\text{for }\mathrm{supp}(A) \subseteq (-\infty,x_{l}]\times \Z^{d-1}, \\
            \gamma(A) &= \beta(A) \quad\text{for }\mathrm{supp}(A) \subseteq [x_{r}, \infty)\times \Z^{d-1},
        \end{align}
        where \(\mathrm{supp}(A)\) is the smallest subset of \(\Z^d\) on which the operator \(A\) acts non-trivially, \(x_{l,r}\) are positions along the first axis in \(\Z^d\), and \(\alpha(\cdot)\) is the map on local operators defined by conjugation by the locality-preserving unitary defining \(\alpha\).
        That is, \(\gamma\) agrees with \(\alpha\) to the left of some cut, and agrees with \(\beta\) to the right of the cut. We can similarly define a blend \(\gamma(\Bt)\) between parametrized families of QCA \(\alpha(\Bt)\) and \(\beta(\Bt)\). Two (families of) QCA are said to blend into each other if a blend exists between them.

        Two QCA blend into one another if and only if they are related by a finite depth circuit of lower-dimensional QCAs~\cite{Freedman2022group}. That is, \(\alpha\) blends with \(\beta\) at any position of the cut \(x\) if and only if \(\alpha \beta^{-1}\) is equal to a finite depth unitary circuit. (Actually, it is sufficient for a blend to exist at any position~\cite{Haah2023invertible}.)

        Indeed, it is sufficient to consider a blending between \(\gamma = \alpha \beta^{-1}\) and the identity. Then \(\gamma\) should be equal to a finite depth circuit. 
        Let the blending at the cut across \(\{x\}\times \Z^{d-1}\) be \(\tilde{\gamma}_x\), and coarse grain such that \(\tilde{\gamma}\) agrees with \(\gamma\) for \(x'<x\) and acts as the identity for \(x'>x\).
        Then we have that \(\tilde{\gamma}_x \tilde{\gamma}_{x+1}^{-1}\) and \(\tilde{\gamma}^{-1}_x \tilde{\gamma}_{x+1}\) are QCA supported on \(\{x,x+1\}\times \Z^{d-1}\).
        
        We can split \(\gamma\) into a depth-two circuit by iteratively splitting off gates of the form \(\tilde{\gamma}_x \tilde{\gamma}_{x+1}^{-1}\) or \(\tilde{\gamma}^{-1}_x \tilde{\gamma}_{x+1}\). (The gates are themselves QCA in \(d-1\) dimensions.) First, observe that we can split \(\gamma\) into a product of two semi-infinite QCA
        \begin{equation}
            \gamma = \tilde{\gamma}_0 (\tilde{\gamma}_0^{-1} \gamma).
        \end{equation}
        Then we can express the \(\tilde{\gamma}_0\) factor as
        \begin{equation}
            \tilde{\gamma}_0 = (\tilde{\gamma}_0 \tilde{\gamma}_{-1}^{-1}) \tilde{\gamma}_{-1}.
        \end{equation}
        The first factor is the gate we need. We can proceed similarly with the second factor, decomposing \(\tilde{\gamma}_{-1} = \tilde{\gamma}_{-2} (\tilde{\gamma}^{-1}_{-2} \tilde{\gamma}_{-1})\) and so on, alternating whether we perform the multiplication by \(\tilde{\gamma}^{-1}_x \tilde{\gamma}_x\) on the right or \(\tilde{\gamma}_x \tilde{\gamma}^{-1}_x\) on the left. Mirroring the same procedure for positive \(x\), we obtain a depth-two circuit decomposition
        \begin{equation}
            \gamma = \left(\prod_{x\in 2\Z+1} w_{x,x+1} \right) \left(\prod_{x\in 2\Z} w_{x,x+1} \right).
        \end{equation}

        The converse direction for the equivalence of blending and circuit equivalence is straightforward. If \(\gamma\) is a finite depth circuit, we can simply delete all gates to one side of any cut to obtain a blending with the identity.

        It is intuitive that a blending should exist for any path of QCA beginning at the identity. Such a path of unitary operators \(\gamma(t)\) can be generated by a local Hamiltonian \(H(t)\) if it is differentiable. Then tracing over the sites to the right of any cut gives a new Hamiltonian \(\mathrm{Tr}_{(x,\infty)\times \Z^{d-1}}{H(t)}\) which, assuming it continues to generate a QCA, generates a blending between \(\gamma(t)\) and the identity. This is certainly true in one dimension, where any unitary evolution at the zero-dimensional edge is a QCA, but in higher dimensions the new Hamiltonian evolution may develop exponential tails along the edge. This is not much of a problem physically, but such tails are an obstacle to performing proofs.

        Nonetheless, we expect that the path-space \(P\QCA_{d}\) can be continuously parametrized in terms of a strictly local path of circuits in dimensions higher than one too. The loop space \(\Omega\QCA_d\) can then be parametrized in terms of circuits which begin and end at the identity. Note that this does not require that the individual gates must return to the identity, only that they cancel at the end of the path.

    \subsection{Deformation to the swindle}
        \label{subapp:deformation}

        We now argue that Eq.~\eqref{eqn:deformation_result} holds.
        
        First, we show that \(S: \QCA_{d-1} \to \Omega \QCA_{d}\) is injective, such that \(\mathrm{im}(S) \cong \QCA_{d-1}\) is homeomorphic to \(\QCA_{d-1}\). That is, \(S(v,t) = S(w,t)\) for all \(t\) only if \(v=w\). Indeed, multiplying \(S(v,t)\) on the left by the inverse of the top layer in \(S(w,t)\) and on the right by the inverse of the bottom layer (\autoref{fig:swindle}), we have
        \begin{multline}
            \left(\prod_{x\in 2\Z+1} u_{x,x+1}(w,t)^\dagger u_{x,x+1}(v,t) \right) \\
            \times\left(\prod_{x\in 2\Z}  u_{x,x+1}(v,t) u_{x,x+1}(w,t)^\dagger \right) = \mathbbm{1}.
        \end{multline}
        This equation can only be satisfied if, for all \(t \in [0,1]\),
        \begin{align}
            u_{2x,2x+1}(v,t) u_{2x,2x+1}(w,t)^\dagger &= q_{2x}(t) q_{2x+1}(t) \\
            u_{2x-1,2x}(w,t)^\dagger u_{2x-1,2x}(v,t) &= q_{2x-1}(t)^\dagger q_{2x}(t)^\dagger \label{eqn:inj_step}
        \end{align}
        for some sequence of \(q_x(t)\) supported on \(\{x\}\times \Z^{d-1}\). However, the gates \(u_{x,x+1}\) are translationally invariant, so considering \(x=0\) gives
        \begin{equation}
            u_{0,1}(v,t) u_{0,1}(w,t)^\dagger = q_{0}(t) q_{1}(t) = u_{0,1}(v,t)^\dagger u_{0,1}(w,t),
        \end{equation}
        where the second equality comes from translating Eq.~\eqref{eqn:inj_step} by one unit towards positive \(x\) and inverting the equation. We see that \(u_{0,1}(v,t)^2 = u_{0,1}(w,t)^2\). Making a spectral decomposition of both sides, we conclude that \(u_{0,1}(v,t) = u_{0,1}(w,t)\) up to multiplying eigenvalues by \(\pm1\). However, these paths both start at \(t=0\) with the same eigenvalues \(\lambda = +1\), which remain nonzero throughout all \(t\). Thus, all the discrete choices of \(\pm1\) must be \(+1\), and \(u_{0,1}(v,t) = u_{0,1}(w,t)\). Then comparing these at \(t=1\) gives \(v_0 v^\dagger_1 = w_0 w^\dagger_1\), and hence \(v = w\) up to a phase. We regard QCA as projective unitaries, so \(v=w\).
        
        We now assume that all loops of QCA \(\gamma(t)\) can be represented by paths of circuits (Appendix~\ref{subapp:blending})
        \begin{equation}
            \gamma(t) = \left(\prod_{x\in 2\Z+1} w_{x,x+1}(t) \right) \left(\prod_{x\in 2\Z} w_{x,x+1}(t) \right)
            \label{eqn:gamma_char}
        \end{equation}
        such that \(w_{x,x+1}(0) = \mathbbm{1}\) and the circuit is the identity at \(t=1\). This requires that \(w_{x,x+1}(1) = v^\gamma_x \otimes v^{\gamma\dagger}_{x+1}\) for some sequence \(v^\gamma_x \in \QCA_{d-1}\) which are not necessarily all equal.
        
        \begin{figure}
            \centering
            \includegraphics[width=\linewidth]{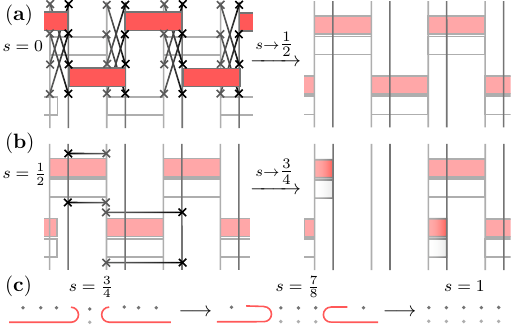}
            \caption{The loop of finite depth circuits \(\gamma(t) \otimes S(v_0^{\gamma},s)^\dagger\) is deformed to the constant loop at the identity as a function of \(s \in [0,1]\). (a)~At \(s=0\), the circuit consists of two layers. Between \(s=0\) and \(s=1/2\) the gates \(u^\dagger_{x,x+1}(v_0^{\gamma},t)\) in the top layer (red) are swapped to the bottom layer while having their order of multiplication reversed.
            (b)~The gate \(u^\dagger_{0,1}(v_0^{\gamma},t) w_{0,1}(t)\) [\(u^\dagger_{-1,0}(v_0^{\gamma},t) w_{-1,0}(t)\)] acts as the identity on site \(0\) for \(t\in\{0,1\}\). It can thus be swapped (over \(s \in [\tfrac{1}{2},\tfrac{3}{4}]\)) so that it acts on site \(1\) [\(-1\)] and its ancilla while ensuring that the whole circuit is still a loop. (c)~Between \(s=\tfrac{3}{4}\) and \(s=1\) the circuit (red) is ``rolled up'', such that it acts as the identity (gray dots) in any finite region.} 
            \label{fig:swindle_deformation}
        \end{figure}
        
        Then, up to stable equivalence, we have
        \begin{equation}
            \gamma(t) = \gamma(t) \otimes \mathbbm{1}
            = \gamma(t) \otimes S(v^\gamma_0,t)^\dagger S(v^\gamma_0,t).
            \label{eqn:gamma_stab}
        \end{equation}
        We now make a sequence of deformations which cancel the \(\gamma(t)\) and \(S(v_0,t)^\dagger\) parts of Eq.~\eqref{eqn:gamma_stab}, leaving just \(\mathbbm{1}\otimes S(v_0,t)\). That is, we define a continuous function
        \begin{align}
            E: \Omega \QCA_d \times [0,1] &\to \Omega\QCA_d \\
            (\gamma(t), s) &\mapsto E(\gamma, t, s)
        \end{align}
        such that \(E(\gamma,t,0) = \gamma(t) \otimes \mathbbm{1}\) and \(E(\gamma,t,1) = \mathbbm{1}\otimes S(v_0^\gamma,t)\). The deformations all take the form of swapping the positions of gates in the circuit representation of Eq.~\eqref{eqn:gamma_stab} (\autoref{fig:swindle_deformation}). The \(S(v_0^\gamma,t)\) factor is a spectator, and so we will focus on \(\gamma(t)\otimes S(v_0,t)^\dagger\).
        
        First, in \(s \in [0,\tfrac{1}{2}]\), we swap the gates \(u_{x',x'+1}^\dagger(v^\gamma_0,t)\) (we denote ancillae sites with a prime, \(x'\)) from the second tensor factor---the second \emph{layer}---to the first, such that the new gates in the first layer become
        \begin{equation}
            w'_{x,x+1}(t) = u^\dagger_{x,x+1}(v^\gamma_0,t) w_{x,x+1}(t).
        \end{equation}
        This also requires reversing the order of multiplication in
        \begin{equation}
            S^\dagger = \left(\prod_{x'\in2\Z} u^\dagger_{x',x'+1}\right) \left(\prod_{x'\in2\Z+1} u^\dagger_{x',x'+1}\right),
        \end{equation}
        so that it matches the circuit structure of \(\gamma(t)\). This can be done by adding the gates \(\mathrm{SWAP}_{x,x'}(2s)\) to the positions marked in \autoref{fig:swindle_deformation}(a).
        
        Observe that
        \begin{equation}
            w'_{0,1}(1) = \mathbbm{1}\otimes v^{\gamma}_0 v^{\gamma\dagger}_1 \text{ and } w'_{-1,0}(1) = v^{\gamma\dagger}_0 v^\gamma_{-1} \otimes \mathbbm{1}
        \end{equation}
        both act as the identity on site \(x=0\) at \(t=1\). Thus, in \(s \in[\tfrac{1}{2},\tfrac{3}{4}]\), we can introduce \(\mathrm{SWAP}(4(s-\tfrac{1}{2}))\) gates to the positions shown in \autoref{fig:swindle_deformation}(b) which ``roll up'' the circuit---the gates previously acting on site \(0\) now act on-site at \(x=\pm1\) between ancillae. Explicitly, for \(s \in [\tfrac{1}{2},\tfrac{3}{4})\), the gate in \(E(\gamma,t,s)\) acting across the \((00',11')\) edge is
        \begin{equation}
            \mathrm{SWAP}_{0,1'}(4(s-\tfrac{1}{2})) w'_{0,1}(t) \mathrm{SWAP}_{0,1'}(4(s-\tfrac{1}{2}))^\dagger.
        \end{equation}
        As \(w'_{0,1}(1)\) acts as the identity on site 0, the introduced swap gates cancel at \(t=1\), and the entire circuit remains a loop for all \(s\). By \(s=\tfrac{3}{4}\), the gate acts as the identity on site 0 for all \(t \in [0,1]\). We repeat an analogous deformation on the other gate with support on site \(0\). The part of the circuit which is supported on \(x=0\) is thus rolled up to act on the ancillae at sites \(x=\pm1\).
        
        This rolling up procedure can be iterated. At \(s = \frac{3}{4}\), we see that the gate \(w'_{1,2}(t) w'_{1',1}(t)\) is the identity on sites \(1\) and \(1'\) when \(t=1\). Thus, by adding ancillae to site \(2\), this gate can also be rolled up within \(s\in[\tfrac{3}{4},\tfrac{7}{8}]\) [\autoref{fig:swindle_deformation}(c)]. We continue rolling up sites \(x\) and \(-x\) within \(s\in[1-2^{-x-1},1-2^{-x-2}]\), such that, as \(s\to 1\), we have \(E(\gamma,t,s) \to \mathbbm{1} \otimes S(v_0^\gamma,t)\) on any finite patch of the lattice. Technically, this makes \(E(\gamma,t,s)\) continuous in the strong operator topology.
        
        The deformation is completed by adding a parallel layer of swap gates which takes \(\mathbbm{1} \otimes S(v_0^\gamma,t)\) to \(S(v_0^\gamma,t) \otimes \mathbbm{1} \in \mathrm{im}(S)\). If the circuit decomposition Eq.~\eqref{eqn:gamma_char} gives \(w_{x,x+1} = u_{x,x+1}(v)\) [as it does when using the \(\mathrm{Tr}_{(x,\infty)\times \Z^d}{H(t)}\) construction of Appendix~\ref{subapp:blending}] then \(E\) is a deformation retract, and Eq.~\eqref{eqn:deformation_result} holds.

\section{Higher homotopy groups of symmetric QCA}
    \label{app:homotopy_symmetry}

    The symmetric implementation of the swap gate constructed in \autoref{sec:LoopsWithSymmetry} allows all the arguments demonstrating the \(\Omega\)-spectrum property to be carried out identically for symmetric QCA. As in the case without symmetry, the homotopy groups \(\pi_n(\QCA^{\mathcal{R}}_d)\) with \(d\leq n\) are determined by \(\pi_{n-d}(\QCA^{\mathcal{R}}_0)\). Thus, we make significant progress on characterizing symmetric QCA in higher dimensions through an in-depth analysis of zero-dimensional QCA with symmetry. 

    We will assume that the symmetry acts linearly on the local degrees of freedom. Indeed, as we are concerned with the stable classification, we can always consider stacking with an ancilla representation such that the symmetry becomes linear.

    In Appendix~\ref{subapp:symm_zerod} we compute all the homotopy groups of symmetric QCA in zero dimensions. We verify a functorial property of the classification in Appendix~\ref{subapp:functor}, and apply the classification to the QP pump of \autoref{subsec:QP_pump} in Appendix~\ref{subapp:QPpump_invariant}.

    \subsection{Zero dimensions with symmetry}
        \label{subapp:symm_zerod}

        Our strategy to classify higher homotopy groups of QCA with symmetry is the same as the case without symmetry: we identify unstable invariants (at fixed Hilbert space dimension) and subsequently find a stable presentation of these invariants which does not change upon adding ancillae. Finally, we characterize the abstract classifying group.

        \subsubsection{Unstable invariants}

        We consider a Hilbert space \(\mathcal{H}\) with finite dimension \(N\) carrying a continuous unitary representation \(\rho: G \to U(N)\) of an arbitrary topological group \(G\). Then \(\rho\) is a direct sum of irreducible representations (irreps) \(r\) each with multiplicity \(n_r\),
        \begin{equation}
            \rho = \bigoplus_{r\text{ irrep of }G} r \otimes \mathbbm{1}_{n_r},
            \label{eqn:rho_irrep}
        \end{equation}
        and we have
        \begin{equation}
            N = \sum_{r\text{ irrep of }G} n_r d_r,
        \end{equation}
        where \(d_r\) is the dimension of the irrep \(r\).
        In particular, all but finitely many \(n_r\) are zero.

        We first describe the classification of symmetric unitaries, before explaining how this classification is modified for projective unitaries. Any family of unitaries \(\tilde{V}: S^n \to U(N)\) which commutes with \(\rho\) must decompose as
        \begin{equation}
            \tilde{V} = \bigoplus_{r\text{ irrep of }G} \mathbbm{1}_{d_r} \otimes \tilde{V}_r.
        \end{equation}
        Any deformation of this family which continues to commute with \(\rho\) must maintain this structure. Thus, the families \(\tilde{V}_1(\Bt)\) and \(\tilde{V}_2(\Bt)\) are homotopic through a symmetric homotopy if and only if \(\tilde{V}_{1r}(\Bt)\) and \(\tilde{V}_{2r}(\Bt)\) are homotopic for all \(r\). We have a complete set of invariants given by the winding numbers of each \(\tilde{V}_r\),
        \begin{equation}
            W^\rho_n[\tilde{V}] = (W_n[\tilde{V}_{r}])_{r\text{ irrep of }G}.
        \end{equation}
        
        With some foresight, we choose to identify \(W^\rho_n[\tilde{V}]\) with an element of the module
        \begin{multline}
            \mathrm{Rep}(G)[\pi_n(U(N))] = \Bigg\{ \sum_{\rho \in \mathrm{Rep}(G)} \rho \cdot W_n^\rho \\
            :\, W_n^\rho \in \pi_n(U(N)) \Bigg\},
        \end{multline}
        which consists of finite formal sums of winding numbers with coefficients in \(\mathrm{Rep}(G)\), regarded as a rig (a ri\emph{n}g without \emph{n}egatives). That is, addition in \(\mathrm{Rep}(G)\) is direct sum, multiplication is tensor product (which distributes over direct sums) and elements are regarded up to unitary equivalence of representations, so that addition and multiplication are commutative. We do not append the formal negatives used to define the representation ring \(R_G\). 
        However, as \(\pi_n(U(N)) \in \{0, \Z\}\) for sufficiently large \(N\), the module \(\mathrm{Rep}(G)[\pi_n(U(N))]\) is isomorphic to either \(0\) or \(R_G\).
        
        Then we define the complete invariant for symmetric unitaries
        \begin{equation}
            W^\rho_n[\tilde{V}] = \sum_{r\text{ irrep of }G} r \cdot W_n[\tilde{V}_r].
        \end{equation}
        
        QCA are properly projective unitaries, rather than unitaries. \(W^\rho_n[\tilde{V}]\) continues to function as a complete invariant for \(n \geq 2\), as in the case without symmetry. We may thus denote them \(W^\rho_n[V]\), where \(V: S^n \to PU(N)\) is a family of symmetric projective unitaries. However, there are symmetric projective unitaries \(V\) which cannot be expressed as symmetric proper unitaries (requiring a larger group of \(n=0\) invariants), and there are loops \(V(\theta)\) of symmetric proper unitaries which become equivalent when regarded as symmetric projective unitaries (requiring a smaller group of \(n=1\) invariants).

        We first recall that connected components of symmetric projective unitaries (\(n=0\)) are labeled by one-dimensional representations of \(G\). Given a symmetric projective unitary \(V\), we can arbitrarily choose a proper unitary \(\tilde{V}\) in the preimage of \(V\) under the quotient by global phases. The fact that \(V\) is symmetric only implies that \(\tilde{V}\) is symmetric up to a global phase,
        \begin{equation}
            \rho(g)^\dagger \tilde{V} \rho(g) = \chi[V](g) \tilde{V},
            \label{eqn:V_lift_proj}
        \end{equation}
        where we assume that \(\chi[V](g)\) is continuous as a function of \(g\).
        Note that \(\chi[V]\) is the same when \(\tilde{V}\) is multiplied by a global phase, justifying the notation \(\chi[V]\) rather than \(\chi[\tilde{V}]\). Using that \(\rho\) is a representation, we see that \(\chi[V]:G \to U(1)\) is a one-dimensional representation of \(G\). From Eq.~\eqref{eqn:rho_irrep}, we see that \(\tilde{V}\) must decompose as
        \begin{equation}
            \tilde{V} = \sum_{r\text{ irrep of }G} \sigma_{r \to \chi[V] r} \otimes \tilde{V}_r,
            \label{eqn:Vtilde_decomp}
        \end{equation}
        where \(\sigma_{r \to \chi[V] r}\) is a matrix with a single nonzero element equal to 1, such that \(\tilde{V}\) maps the irrep \(r\) block of \(\rho\) to the irrep \(\chi[V] r\) block.  Any path of projective unitaries \(V(t)\) lifts to a path of proper unitaries \(\tilde{V}(t)\) [as the quotient \(q : U(N) \to PU(N)\) is a fibration], and the space of permutations \(\sigma_{r \to \chi[V] r}\) is discrete, so the permutation label \(\chi[V]\) is an invariant for path components of projective unitaries. Any unitaries with the same \(\chi[V]\) in the decomposition Eq.~\eqref{eqn:Vtilde_decomp} can be connected by a path, simply by connecting each \(\tilde{V}_r\) by a path. Thus, \(\chi[V]\) is a complete invariant.
        
        However, not all \(\chi[V] \in \mathrm{Hom}(G \to U(1)) =: X_G\) can occur in a finite Hilbert space. For the \(\tilde{V}\) in Eq.~\eqref{eqn:Vtilde_decomp} to be unitary, we need that \(\sum_r \sigma_{r \to \chi[V] r}\) is a finite-dimensional permutation matrix, and so in particular \(\chi[V]\) has finite order: there is an \(n\) such that \(\chi[V](g)^n = 1\) for all \(g\). Thus, \(\chi[V]\) must belong to the torsion subgroup of \(X_G\) (the subgroup of elements of finite order, with the group operation in \(X_G\) defined pointwise), which we denote \(X_G^T\). Conversely, given any \(\chi[V] \in X_G^T\) with order \(n\), we can take
        \begin{equation}
            \rho = \bigoplus_{p=0}^{n-1} \chi[V]^p
        \end{equation}
        and then the unitary  \(\sum_{p=0}^{n-1}\sigma_{\chi[V]^p \to \chi[V]^{p+1}}\) achieves the invariant \(\chi[V]\).
        
        For \(n=1\), we can multiply a loop of proper unitaries \(\tilde{V}(\theta)\) by \(e^{i m \theta}\) and change all the winding numbers, while still forming the same loop of projective unitaries after the quotient by global phases. Indeed, we find
        \begin{align}
            W^\rho_1[e^{i m \theta} \tilde{V}] &= \sum_{r\text{ irrep of }G} r\cdot W_1[e^{i m \theta}\tilde{V}_r] \\
            &= \sum_{r\text{ irrep of }G} r \cdot (m n_r + W_1[\tilde{V}_r]) \\
            &= \left(\sum_{r\text{ irrep of }G} r\otimes\mathbbm{1}_{n_r} \right)\cdot m + W^\rho_1[\tilde{V}] \\
            &= \rho \cdot m + W^\rho_1[\tilde{V}],
        \end{align}
        where we used
        \begin{equation}
            r\cdot n_r = r\cdot(1+\cdots+1) = r+\cdots +r = r\otimes \mathbbm{1}_{n_r} \cdot 1
        \end{equation}
        for \(n_r \geq 0\) and the decomposition of \(\rho\) into irreps, \(\rho = \sum_r r \otimes \mathbbm{1}_{n_r}\). Thus, we must consider the coset
        \begin{equation}
            w^\rho_1[V] = W^\rho_1[\tilde{V}] + \rho \cdot \pi_1(U(N)) \cong W^\rho_1[\tilde{V}] + \rho \cdot \Z,
        \end{equation}
        which is now a well-defined invariant for loops of projective unitaries \(V(\theta)\).

        \subsubsection{Stable invariants}

        With the unstable invariants identified, we must now consider the stabilization process of appending ancillae. We have a directed system among symmetric projective unitary groups \(PU(N)^\rho\) [the subset of \(PU(N)\) which commutes with the representation \(\rho\)] given by
        \begin{align}
            \iota_{\rho \to \rho \sigma} : PU(N)^\rho &\to PU(NM)^{\rho \sigma}, \\
            V &\mapsto V \otimes \mathbbm{1} \nonumber
        \end{align}
        where we take \(\rho, \sigma \in \mathcal{R}\) (the monoid of ancillae), \(N = \mathrm{dim}(\rho)\), and \(M = \mathrm{dim}(\sigma)\). The stabilized space \(\QCA^\mathcal{R}_0\) is the direct limit
        \begin{equation}
            \QCA^{\mathcal{R}}_0 = \varinjlim PU(N)^\rho.
        \end{equation}
        Strictly, we require that the representation \(\rho\) grows through the directed system such that each representation in \(\mathcal{R}\) eventually occurs as a tensor factor in \(\rho\).

        Starting at \(n=0\), we find the invariant \(\chi[V] = e^{i \omega^\rho[V]}\) of \(\pi_0(PU(N)^\rho)\) is actually already stable to appending ancillae. With \(\omega^\rho[V] \in \R/(2\pi \Z) \cong U(1)\) defined modulo \(2\pi\), we can see from Eq.~\eqref{eqn:V_lift_proj} that
        \begin{equation}
            \iota_{\rho \to \rho \sigma}^*(\omega^\rho[V]) = \omega^{\rho\sigma}[V \otimes \mathbbm{1}] = \omega^\rho[V]
        \end{equation}
        and more generally
        \begin{equation}
            \omega^{\rho \sigma}[V \otimes Q] = \omega^\rho[V] + \omega^\sigma[Q].
        \end{equation}

        To compute the action of the induced maps on the remaining unstable invariants, \(\iota^*_{\rho \to \rho \sigma}(W_n^\rho[V]) = W_n^{\rho\sigma}[V\otimes \mathbbm{1}]\), we suppose that we have Clebsch-Gordan multiplicities \(C_{r r'}^k \in \Z\) such that, for \(r, r', k\) irreducible,
        \begin{equation}
            r r' = \sum_k C^{r r'}_k k \in \mathrm{Rep}(G),
        \end{equation}
        where \(C^{r r'}_k k = k + \cdots + k = \mathbbm{1}_{C^{r r'}_k} \otimes k\). Taking
        \begin{equation}
            \rho = \sum_r n^\rho_r r \quad\text{and}\quad
            \sigma = \sum_{r'} n^\sigma_{r'} r',
        \end{equation}
        we have, with summation implied,
        \begin{equation}
            \rho \sigma = n_r^\rho n_{r'}^\sigma C_k^{r r'} k.
        \end{equation}
        Then \(V \otimes \mathbbm{1}\) decomposes into the irreducible subspaces of \(\rho \sigma\) (possibly after multiplication by a permutation of the irrep subspaces) as
        \begin{equation}
            V \otimes \mathbbm{1} \cong \bigoplus_k \left[\mathbbm{1}_{d_k} \otimes \left(\bigoplus_{r, r'} V_r \otimes \mathbbm{1}_{n_{r'}^\sigma C^{r r'}_k}\right) \right],
        \end{equation}
        and the invariant becomes
        \begin{align}
            W^{\rho\sigma}_n[V\otimes \mathbbm{1}] &= \sum_{k, r, r'} k\cdot W_1[V_r \otimes \mathbbm{1}_{n_{r'}^\sigma C^{r r'}_k}] \\
            &= \sum_{k,r,r'} (n_{r'}^\sigma C^{r r'}_k k)\cdot W_1[V_r] \\
            &= \sum_{r} (\sigma r)\cdot W_1[V_r] \\
            &= \sigma \cdot W^{\rho}_n[V].
        \end{align}
        An analogous calculation produces the generalization
        \begin{equation}
            W^{\rho \sigma}_n[V\otimes Q] = \sigma \cdot W^\rho_n[V] + \rho\cdot W^\sigma_n[Q].
            \label{eqn:symm_Wadd}
        \end{equation}
        
        Then the induced directed system is
        \begin{align}\label{eqn:dir_pi_symm}
            \iota^*_{\rho \to \rho \sigma}: \mathrm{Rep}(G)[\pi_n(U(N))] &\to \mathrm{Rep}(G)[\pi_n(U(NM))] \n
            W^\rho_n[V] &\mapsto \sigma \cdot W^\rho_n[V]. 
        \end{align}
        We recover Eq.~\eqref{eqn:dir_pi} by taking \(G = 0\) to be the trivial group.

        Thus, we can identify a stable invariant by formally dividing \(W^\rho_n[V]\) by the representation \(\rho\),
        \begin{equation}
            \widetilde{W}^\rho_n[V] = \frac{W^\rho_n[V]}{\rho} \in \mathcal{R}^{-1} \mathrm{Rep}(G)[\pi_n(U(N))],
        \end{equation}
        which then obeys \(\widetilde{W}^{\rho\sigma}_n[V\otimes \mathbbm{1}] = \widetilde{W}^\rho_n[V]\).
        Here, \(\mathcal{R}^{-1} \mathrm{Rep}(G)[\pi_n(U(N))]\) is the algebraic localization of the module \(\mathrm{Rep}(G)[\pi_n(U(N))]\), defined in the next section.

         For \(n=1\), the invariant is a coset
        \begin{equation}
            \widetilde{W}^\rho_1[V] = \widetilde{W}^\rho_1[V] + 1\cdot \Z \in \mathcal{R}^{-1} \mathrm{Rep}(G)[\pi_1(U(N))]/(1\cdot \Z).
            \label{eqn:first_winding_symm}
        \end{equation}
        
        \subsubsection{Classification}

        If \(M\) is a module over \(R\) and \(S\subseteq R\) is a multiplicatively closed subset of \(R\), then the \emph{algebraic localization} of \(M\) is the set \(S^{-1} M := (M\times S)/\sim\), where the equivalence relation is
        \begin{multline}
            (m_1,s_1) \sim (m_2,s_2) \iff \exists \sigma \in S \text{ such that}\\
            \sigma\cdot (s_1\cdot m_2  - s_2 \cdot m_1) = 0 \in M.
            \label{eqn:alg_loc_def}
        \end{multline}
        We write the equivalence class of \((m,s)\) as \(m/s\) and equip this set with the addition operation and \(R\)-action
        \begin{align}
            \frac{m_1}{s_1} + \frac{m_2}{s_2} &= \frac{s_2 \cdot m_1 + s_1 \cdot m_2}{s_1 s_2}, \label{eqn:loc_add}\\
            \rho \cdot \frac{m}{s} &= \frac{\rho \cdot m}{s}.
        \end{align}

        The formal ratios \(\widetilde{W}^\rho_n[V]\) are properly defined as elements of the localization \(\mathcal{R}^{-1} \mathrm{Rep}(G)[\pi_n(U(N))]\), which from Eqs.~(\ref{eqn:symm_Wadd},~\ref{eqn:loc_add}) gives
        \begin{equation}
            \widetilde{W}^{\rho\sigma}_n[V\otimes Q] = \widetilde{W}^\rho_n[V] + \widetilde{W}^\sigma_n[Q].
        \end{equation}
        The one-dimensional stable symmetric winding number is an additive coset in
        \begin{equation}
            \widetilde{W}^\rho_1[V] \in \mathcal{R}^{-1} \mathrm{Rep}(G)[\pi_1(U(N))]/(1\cdot \pi_1(U(N))).
        \end{equation}
        Again, as \(\pi_n(U(N)) \in \{0,\Z\}\), these groups are always isomorphic to either the trivial group or \(\mathcal{R}^{-1} R_G\).

        Note that Eq.~\eqref{eqn:alg_loc_def} implies that \(\widetilde{W}^{\rho_1}_n[V] = \widetilde{W}^{\rho_2}_n[Q]\) if and only if there are representations \(\sigma \otimes \rho_1\) and \(\sigma \otimes \rho_2\) in \(\mathcal{R}\) such that the unstable winding numbers
        \begin{equation}
            W_n[V \otimes \mathbbm{1}_{\sigma \rho_2}] = \sigma \rho_2 \cdot W_n[V] = \sigma\rho_1 \cdot W_n[Q] = W_n[Q \otimes \mathbbm{1}_{\sigma \rho_1}]
        \end{equation}
        are equal. Thus, \(V\) and \(Q\) are strongly equivalent, exactly as we want.

        We have identified a complete set of stable invariants among the elements of the localization \(\mathcal{R}^{-1} R_G\), but not all elements in \(\mathcal{R}^{-1} R_G\) necessarily correspond to achievable invariants with a given \(\mathcal{R}\). The subset which can occur is that where the numerator involves irreps which also occur in the denominator. We define the additive subgroup of \(\mathcal{R}^{-1} R_G\)
        \begin{multline}
            \mathcal{K}^{\mathcal{R}} = \Bigg\{\frac{\sum_{r\text{ irrep}} n_r r}{\rho} \in \mathcal{R}^{-1} R_G \,:\, n_r \neq 0 \implies \\
            \rho = s + r \text{ for some }s \in \mathrm{Rep}(G)\Bigg\}.
        \end{multline}
        This is precisely the subgroup of winding number invariants which can occur.
        
        To conclude, the homotopy groups of symmetric zero-dimensional QCA are
        \begin{equation}
            \pi_n(\QCA_0^\mathcal{R}) \cong 
            \left\{
            \begin{array}{l l}
                \leq X^T_G & \quad n=0, \\
                \mathcal{K}^{\mathcal{R}}/\Z & \quad n=1, \\
                0 & \quad n\geq 2\text{ even,} \\
                \mathcal{K}^{\mathcal{R}} & \quad n \geq 3\text{ odd,}
            \end{array}
            \right.
            \label{eqn:zerodim_symm_class}
        \end{equation}
        where by \(\Z\) we mean the subgroup generated by \(1 \in \mathcal{K}^{\mathcal{R}}\). We have also not precisely characterized the restriction on \(\pi_0(\QCA_0^\mathcal{R}) \leq X_G^T\) imposed by the limited ancillae in \(\mathcal{R}\), as our focus is on the less familiar higher homotopy groups.

        As an example, we consider \(G = U(1)\). Then all irreducible representations are one-dimensional, and given by \(U(1) \ni e^{i \theta} \mapsto e^{i m \theta}\) (in multiplicative notation). Identifying this irrep with the polynomial \(x^m\), the representation ring is isomorphic to the ring of integer Laurent polynomials in \(x\),
        \begin{equation}
            R_{U(1)} \cong \Z[x, x^{-1}].
        \end{equation}
        We have that \(X_{U(1)} \cong \Z\), so \(X_{U(1)}^T \cong 0\) is trivial.
        Let us suppose that the monoid of ancillae \(\mathcal{R}\) contains all finite dimensional representations of \(U(1)\), so \(\mathcal{R}\) is the full representation rig and thus isomorphic to Laurent polynomials with non-negative coefficients,
        \begin{equation}
            \mathcal{R} \cong \N[x, x^{-1}].
        \end{equation}
        Then  \(\mathcal{K}^{\mathcal{R}}\) is the additive group of rational polynomials with positive denominators where all powers which occur in the numerator also occur in the denominator,
        \begin{multline}
            \mathcal{K}^{\mathcal{R}} \cong \bigg\{ \frac{p(x)}{q(x)} \in \mathrm{Frac}\, \Z[x] \,:\,
            q(x) \in \N[x], \\
            x^n \in p(x) \implies x^n \in q(x) \bigg\}.
            \label{eqn:U(1)_loopclass}
        \end{multline}
        We used that one can divide the numerator and denominator of a ratio of Laurent polynomials by the most negative power of \(x\) in either to obtain a ratio of usual polynomials. Then the homotopy groups of \(\QCA^{\mathcal{R}}_0\) with this set of ancillae is given by Eq.~\eqref{eqn:zerodim_symm_class}. We note that the rational polynomial structure which appeared in Ref.~\cite{Zhang2021U1} is again occurring here, as it is a consequence of the representation theory of \(U(1)\).

    \subsection{Functoriality}
        \label{subapp:functor}

        Any classification of symmetric objects should, ideally, be compatible with \emph{forgetting the symmetry}. That is, it should be possible to map the classification of symmetric objects into the classification of objects without symmetry in a consistent way. This is generalized by the notion of the classification being functorial~\cite{Xiong2018minimalist}.

        Suppose that we have a homomorphism of topological groups
        \begin{equation}
            \phi : H \to G,
        \end{equation}
        and a representation \(\rho: G \to U(N)\) of \(G\). Then \(\rho \phi\) is a representation of \(H\),
        \begin{equation}
            \rho\phi: H \xrightarrow{\phi} G \xrightarrow{\rho} U(N).
        \end{equation}
        Further, we have that \((\rho_1 \oplus \rho_2) \phi = \rho_1 \phi \oplus \rho_2 \phi\) and \((\rho_1 \otimes \rho_2) \phi = \rho_1 \phi \otimes \rho_2 \phi\), so that \(\rho \mapsto \rho \phi\) defines an induced rig homomorphism
        \begin{equation}
            \phi^*_{\mathrm{Rep}}: \mathrm{Rep}(G) \to \mathrm{Rep}(H).
        \end{equation}
        Note that the order of \(G\) and \(H\) are reversed. The map is contravariant.

        What is more, any symmetric \(V \in \QCA_d^{\mathcal{R}}\) can also be regarded as being symmetric with respect to the \(H\) action \(\rho\phi\). Thus, we have a continuous inclusion map
        \begin{equation}
            \phi^*_{\QCA} : \QCA_d^{\mathcal{R}} \to \QCA_d^{\mathcal{R}'}
        \end{equation}
        where \(\phi^*_{\mathrm{Rep}}(\mathcal{R}) \subseteq \mathcal{R}'\). In turn, this continuous map induces homomorphisms
        \begin{equation}
            \phi^*_{n} : \pi_n(\QCA_d^{\mathcal{R}}) \to \pi_n(\QCA_d^{\mathcal{R}'}).
        \end{equation}

        In Appendix~\ref{subapp:symm_zerod} and \autoref{sec:SymmetryOneDimension}, we classified homotopy groups of symmetric QCA with \(d-n\leq 1\). More concretely, we constructed isomorphisms between \(\pi_n(\QCA^{\mathcal{R}}_d)\) and certain groups defined in terms of \(\mathrm{Rep}(G)\). One should want these isomorphisms to be functorial---the induced map \(\phi^*_{n}\) should be reproduced by applying \(\phi^*_{\mathrm{Rep}}\) (or more correctly the induced map on the localization or quotient) to the invariant in \(\mathcal{R}^{-1}\mathrm{Rep}(G)\).

        It is not difficult to see that this is the case. Indeed, as the map \(\phi^*_{\QCA}\) does not change the QCA \(V\), only the representation, carrying through the calculation of the invariant just involves replacing all the representations \(\rho\) which occur with \(\rho \phi\), which is precisely the action of \(\phi^*_{\mathrm{Rep}}\).

        An important example is completely forgetting the symmetry, corresponding to the inclusion of the identity, \(\phi: 1 \to G\). Then we have that all representations reduce to the identity operator \(\rho \phi = \mathbbm{1}_{\mathrm{dim}(\rho)}\). In the rig \(\mathrm{Rep}(G)\), this corresponds to the dimension map,
        \begin{equation}
            \mathrm{dim}: \rho \mapsto \mathrm{dim}(\rho),
        \end{equation}
        where the right hand side is to be interpreted as the sum of \(\mathrm{dim}(\rho) \in \N\) copies of the trivial one-dimensional representation \(1\). For \(G = U(1)\), identifying \(\mathrm{Rep}(U(1))\) with \(\N[x,x^{-1}]\), the dimension map just corresponds to evaluating the rational polynomial invariant \(f(x)\) at \(x=1\)~\cite{Zhang2021U1},
        \begin{equation}
            \mathrm{dim}(f(x)) = f(1).
        \end{equation}

    \subsection{Application to the QP pump}
        \label{subapp:QPpump_invariant}

        Finally, we briefly explain where the instance of the QP pump of \autoref{subsec:QP_pump} sits in the classification of Appendix~\ref{subapp:symm_zerod}. Specifically, we find its associated invariant.

        The QP pump is \(U(1)\)-symmetric and has \(d=1\), \(n=2\). It thus has the same classifying group as loops of unitaries in \(d=0\) with \(U(1)\) symmetry. Indeed, this is the object from which our model of the QP pump was constructed. From Appendix~\ref{subapp:symm_zerod}, to get the classification of \(v(\theta)\) we must find its one-dimensional winding number in each symmetry sector. We have
        \begin{equation}
            v(\theta) = e^{-iW \theta/2}
            \begin{pmatrix}
                1 & 0 \\
                0 & e^{i W \theta}
            \end{pmatrix}, \quad
            Q_0 = 
            \begin{pmatrix}
                1 & 0 \\
                0 & 0
            \end{pmatrix}.
        \end{equation}
        Thus, up to a global phase, \(v(\theta)\) has winding number \(W\) in the charge 0 sector, and winding number \(0\) in the charge 1 sector.
    
        The invariant associated to \(v(\theta)\) is a coset [Eq.~\eqref{eqn:first_winding_symm}] of  formal ratios of representations with the numerator weighted by the winding numbers. We identify the charge 0 representation with 1 and the charge 1 representation with \(x \in \N[x,x^{-1}] \cong \mathrm{Rep}(U(1))\). Then the invariant is
        \begin{equation}
            \widetilde{W}_1^Q[v] = \frac{1\cdot W + x \cdot 0}{1 + x} + \Z = \frac{W}{1 + x} + \Z.
        \end{equation}
        We see that all distinct values of \(W\) indeed define inequivalent values of the invariant.
    
        We can also see the functoriality of Appendix~\ref{subapp:functor} in action. Suppose that we explicitly break the symmetry. Then the charge sector \(1 \leftrightarrow x\) should be mapped to having charge \(0\). The new invariant can be found by evaluating the rational function \(\widetilde{W}_1^Q[v]\) at \(x=1\). We find
        \begin{equation}
            \widetilde{W}_1^Q[v] \mapsto \frac{W}{2} + \Z \in \Q/\Z.
        \end{equation}
        Now, all even \(W\) are equivalent to the trivial phase. Interestingly, the odd \(W\) phases remain non-trivial. They no longer support a quantized edge current, as we can deform the model with \(W=+1\) to \(W=-1\), but there is still some nontrivial dynamics at the edge.
        
        Indeed, we believe the odd \(W\) phase without symmetry to be the two-tone-driven phase of Refs.~\cite{Friedman2022qpedge,Dumitrescu2022}. While the solvable limit of the model in Ref.~\cite{Friedman2022qpedge} is highly singular, which makes computing invariants from integral formulas inconvenient, the edge phenomenology described in Ref.~\cite{Friedman2022qpedge} is consistent with what we expect from having \(\widetilde{W}_1 \neq 0 \bmod 1\). Namely, the edge dynamics was argued to be anomalous in that it could not be achieved by a zero-dimensional localized [in the sense that a continuous micromotion operator exists as in Eq.~\eqref{eqn:FB_app}] system which did not occur at the edge of an extended one-dimensional system. As the model in Ref.~\cite{Friedman2022qpedge} is based around a version of the Affleck-Kennedy-Lieb-Tasaki (AKLT) model, which can be trivialized when stacked with itself, it is thus reasonable to suppose that \(\widetilde{W}_1 = 1/2 \bmod 1\) for that model, as this is the only element of \(\Q/\Z\) of order 2. Thus, the abstract classification reveals a relationship between the interacting QP pump and Refs.~\cite{Friedman2022qpedge,Dumitrescu2022} which was not clear previously. It also extends the arguments for stability in Ref.~\cite{Friedman2022qpedge} to the nonperturbative regime, where a picture of emergent symmetry (as employed in that reference) becomes inapplicable.

\section{QCA on manifolds or with spatial symmetries}
    \label{app:manifolds}

    In this appendix, we present conjectures regarding the classification of QCA on manifolds other than \(\R^d\) (as considered in the main text), or with lattice symmetries.
    
    The \(\Omega\)-spectrum property discussed in \autoref{sec:OmegaSpectrum} assumes that QCA are defined on the integer lattice \(\Z^d\). By coarse graining, any locally finite set of points in \(\R^d\) can be deformed to \(\Z^d\), so the QCA is more properly based on the manifold \(\R^d\). This distinction is important, as the classification of QCA based on the \(d\)-dimensional manifold \(M\), \(\QCA_M\), depends on the homology of \(M\)~\cite{Freedman2020higherd} and not just its dimension.

    It is natural to conjecture that QCA on an orientable manifold \(M\) can all be patched together from QCA defined on the open charts in an atlas for \(M\). That is, we construct a map
    \begin{equation}
        \QCA_M \to \map{M}{\QCA_d}
    \end{equation}
    (where \(\map{M}{\QCA_d}\) is the space of continuous functions from \(M\) to \(\QCA_d\)) by identifying the action of \(V \in \QCA_M\) near each point \(x\in M\) as a QCA on \(\R^d\) via a chart in \(M\) containing \(x\). We suspect that a homotopy equivalence \(\QCA_M \simeq \map{M}{\QCA_d}\) can be constructed in this way.
        
    This leads to the hypothesis that, for orientable \(M\)~\cite{Thorngren2018crystal},
    \begin{equation}
        \ho{X}{\QCA_M} \cong \ho{X\times M}{\QCA_d} \cong h^{d}(X \times M).
        \label{eqn:manifolds}
    \end{equation}

    While we leave the general verification of this conjecture to future work, we remark here that Eq.~\eqref{eqn:manifolds} reproduces the classification of QCA on two-dimensional orientable manifolds in Ref.~\cite{Freedman2020higherd}. For non-orientable manifolds, one must consider twisted generalized cohomology---that is, \(\QCA_M\) would be homotopy equivalent to spaces of sections in a \(\QCA_d\)-bundle over \(M\)~\cite{Xiong2018minimalist,Gaiotto2019gencohomology}.
    
    The conjecture of Eq.~\eqref{eqn:manifolds} can be extended to include systems with lattice symmetries.
    Indeed, the generalized cohomology classification of invertible states is also believed to capture lattice symmetries and antiunitary symmetries~\cite{Xiong2018minimalist,Gaiotto2019gencohomology,Thorngren2018crystal}, and this structure can be extended to QCA.

    Continuing to assume that a QCA on \(M\) can be identified by a map from \(M\) to \(\QCA_d\)
    \begin{equation}
        V\in \QCA_M \longleftrightarrow v: M \to \QCA_d,
    \end{equation}
    lattice symmetries can be treated as actions of the symmetry group \(G\) (\(x \mapsto g \cdot x\)) on the base manifold \(M\)~\cite{Thorngren2018crystal}. 
    Then, identifying the QCA \(V\) with the map \(v\), the QCA is symmetric if
    \begin{equation}
        v(x) = \rho_g v(g \cdot x) \rho_g^\dagger
    \end{equation}
    for all \(g \in G\) and \(x \in M\). Then QCA with lattice symmetries would be fixed points of the mapping space
    \begin{equation}
        \map{M}{\QCA_d}^{\mathcal{R}}
    \end{equation}
    under the simultaneous action of \(G\) on \(\QCA_d\) and \(M\). The action on \(\QCA_d\) will generically not be on-site. For instance, a reflection should also act to reflect the QCA \(v(x)\).

    Abstractly, we are lead to equivariant Bredon cohomology, which classifies such fixed points of mapping spaces when the target space is an \(\Omega\)-spectrum~\cite{Greenlees1995equivariant}. The cohomology may, in general, need to be twisted when either \(M\) is not orientable or, as would be more common for physics, the symmetry action is orientation reversing (involves reflections) or antiunitary.

\bibliography{ALTphase}

\end{document}